\documentclass[11pt]{article}

\usepackage{epsfig}
\usepackage{amsfonts}
\usepackage{amssymb}
\usepackage{amsthm}
\usepackage{amstext}
\usepackage{amsmath}
\usepackage{xspace}
\usepackage{graphicx}
\usepackage{graphics}
\usepackage{colordvi}
\usepackage{color}
\usepackage{psfrag}
\usepackage{subfigure}
\usepackage{algorithm,algorithmic}
\usepackage{wrapfig}
\usepackage{url}
\usepackage{enumitem}

\usepackage{verbatim}
\usepackage{tikz}
\usetikzlibrary{shapes,arrows,positioning,shadows,snakes}
\usepackage[margin=1in]{geometry}

\newcommand{\cfbox}[3]{%
    \colorlet{currentcolor}{.}%
    {\color{#1}%
    \framebox[#3]{\color{currentcolor}#2}}%
}
\newcommand{\yi}{{\sc Yes-Instance}\xspace}
\renewcommand{\ni}{{\sc No-Instance}\xspace}

\newcommand{\restatethm}[3]{
  \medskip\noindent{\bf #1~#2.}{\rm (restated)}
  {\it #3}
}

\newcommand{\x}{{\bf x}}

\newcommand{\paren}[1]{\left ( #1 \right ) }

\newcommand{\dimension}[1]{\ensuremath{{\sf dim}(#1)}\xspace}
\newcommand{\induce}[1]{\ensuremath{{\sf im}(#1)}\xspace}
\newcommand{\sinduce}[1]{\ensuremath{{\sf sim}(#1)}\xspace}
\newcommand{\sinducesigma}[2]{\ensuremath{{\sf sim}_{#2}(#1)}\xspace}
\newcommand{\hstrong}{\ensuremath{\times_e}}

\newcommand{\MFS}{{\sc Mrfs}\xspace}
\newcommand{\UDP}{{\sc Udp-Min}\xspace}
\newcommand{\SMP}{{\sc Smp}\xspace}
\newcommand{\MES}{{\sc Maximum Expanding Sequence}\xspace}

\newcommand{\NP}{\mbox{\sf NP}}

\newcommand{\ZPP}{\mbox{\sf ZPP}}

\newcommand{\DTIME}{\mbox{\sf DTIME}}

\newcommand{\opt}{\mbox{\sf OPT}}

\newcommand{\set}[1]{\left\{ #1 \right\}}

\newcommand{\pset}{{\mathcal{P}}}

\newcommand{\bset}{{\mathcal{B}}}
\newcommand{\aset}{{\mathcal{A}}}
\newcommand{\cset}{{\mathcal{C}}}

\newcommand{\mset}{{\mathcal M}}
\newcommand{\iset}{{\mathcal{I}}}

\newcommand{\uset}{{\mathcal{U}}}

\newcommand{\sset}{{\mathcal{S}}}

\newtheorem{theorem}{Theorem}[section]
\newtheorem{lemma}[theorem]{Lemma}

\newtheorem{corollary}[theorem]{Corollary}
\newtheorem{claim}[theorem]{Claim}
\newtheorem{definition}[theorem]{Definition}
\renewcommand{\paragraph}[1]{\medskip\noindent{\bf #1.}\xspace}
\newenvironment{prog}[1]{
\begin{minipage}{5.8 in}
{\sc\bf #1}
\begin{enumerate}}
{
\end{enumerate}
\end{minipage}
}

\renewcommand{\phi}{\varphi}

\newcommand{\poly}{\operatorname{poly}}

\newcommand{\reals}{{\mathbb R}}

\newcommand{\R}{\ensuremath{\mathbb R}}

\newcommand{\Z}{\ensuremath{\mathbb Z}}

\def\danupon#1{}
\def\parinya#1{}
\def\bundit#1{}
\def\danuponb#1{{}}

\newif\iffullversion
\fullversiontrue
\newcommand{\fullversion}[2]{\iffullversion#1\xspace\else#2\xspace\fi}

\usepackage[footnotesize]{caption}

\newcommand{\squishlist}{
 \begin{list}{$\bullet$}
  { \setlength{\itemsep}{0pt}
     \setlength{\parsep}{2pt}
     \setlength{\topsep}{2pt}
     \setlength{\partopsep}{0pt}
     \setlength{\leftmargin}{1.5em}
     \setlength{\labelwidth}{1em}
     \setlength{\labelsep}{0.5em} } }
\newcommand{\squishend}{
  \end{list}  }

\newcommand{\squishnum}{\begin{enumerate}
}
\newcommand{\squishnumend}{
\end{enumerate}}

\newcounter{Lcount}
\newcommand{\squishlisttwo}{
\begin{list}{D\arabic{Lcount}. }
{ \usecounter{Lcount} \setlength{\itemsep}{0pt}
\setlength{\parsep}{0pt} \setlength{\topsep}{0pt}
\setlength{\partopsep}{0pt} \setlength{\leftmargin}{2em}
\setlength{\labelwidth}{1.5em} \setlength{\labelsep}{0.5em} } }

\newcommand{\squishendtwo}{
\end{list} }

\newcommand{\bip}{\mbox{B}\xspace}
\newcommand{\bipm}{\ensuremath{\bip}}
\newcommand{\bipp}{\ensuremath{\bip_e}}

\begin{document}

\title{Graph Products Revisited: Tight Approximation Hardness of
  Induced Matching, Poset Dimension and More}
\author{
Parinya Chalermsook\thanks{
  IDSIA, Lugano, Switzerland.
  Supported by the Swiss National
  Science Foundation project 200020-122110/1
  and by Hasler Foundation Grant 11099. Partially supported by Julia Chuzhoy's CAREER grant CCF-0844872.
  \hbox{E-mail}:~{\tt parinya@cs.uchicago.edu}
} \and
Bundit Laekhanukit\thanks{
  School of Computer Science,
  McGill University, Montreal QC, Canada.
  Supported by the Natural Sciences and Engineering Research Council
  of Canada (NSERC) grant No.~28833 and by European Research Council
  (ERC) Starting Grant 279352.
  \hbox{E-mail}:~{\tt blaekh@cs.mcgill.ca}
} \and
Danupon Nanongkai\thanks{Theory and Applications of Algorithms Research Group, University of Vienna, Austria, and Nanyang Technological University, Singapore. \hbox{E-mail}:~{\tt danupon@gmail.com}.}
%\thanks{University of Vienna, Austria, and Nanyang Technological University, Singapore.\hbox{E-mail}:~{\tt danupon@gmail.com}.}
}
%\date{\today}
\date{}

\maketitle
%{\bf Possible titles:} ``The Art of Graph Products in Hardness of Approximation'', ``Applications of Graph Products in Hardness of Approximation'', ``Tight Hardness of Approximation via Graph Products'', ``It's All About Products'', ``Graph Product Revisited''. (CAREFUL NOT TOO ATTRACT COMPLEXITY PEOPLE. Probably add ``simple'' to avoid hardcore people). REMEMBER TO CITE SHEPHERD ...

%{\bf Note:} Adjacent poset?
%\thispagestyle{empty}
\danupon{To do: $\reals_+$}
\begin{abstract}
Graph product is a fundamental tool with rich applications in both graph theory and theoretical computer science. It is usually studied in the form $f(G*H)$ where $G$ and $H$ are graphs, $*$ is a graph product and $f$ is a graph property. For example, if $f$ is the {\em independence number} and $*$ is the {\em disjunctive product}, then the product is known to be {\em multiplicative}: $f(G*H)=f(G)f(H)$.

In this paper, we study graph products in the following non-standard form: $f((G\oplus H)*J)$ where $G$, $H$ and $J$ are graphs, $\oplus$ and $*$ are two different graph products and $f$ is a graph property. We show that if $f$ is the {\em induced and semi-induced matching number}, then for some products $\oplus$ and $*$, it is {\em subadditive} in the sense that $f((G\oplus H)*J)\leq f(G*J)+f(H*J)$. Moreover, when $f$ is the {\em poset dimension number}, it is {\em almost} subadditive.
% Bun: Change function f to f to make it parallel with the previous sentence.

As applications of this result (we only need $J=K_2$ here), we obtain tight hardness of approximation for various problems in discrete mathematics and computer science: bipartite induced and semi-induced matching (a.k.a. maximum expanding sequences), poset dimension, maximum feasible subsystem with 0/1 coefficients, unit-demand min-buying and single-minded pricing, donation center location, boxicity, cubicity, threshold dimension and independent packing.
\end{abstract}

\pagestyle{plain}

\section{Introduction}
Graph products generally refer to a way to use two graphs, say $G$ and
$H$, to construct a new graph, say $G*H$ for some product
$*$. Studying properties of graphs resulting from a graph product,
i.e., $f(G*H)$ for some function $f$, has been an active research area
with countless applications in graph theory and computer science.
For example, the fact that the {\em independence number $\alpha$} of
the {\em disjunctive product} $G\vee H$ is {\em multiplicative},
i.e., $\alpha(G\vee H)=\alpha(G)\alpha(H)$, has been used to amplify
the hardness of approximating the maximum independent set problem.

In this paper, we study some graph properties when we apply graph
products in a non-standard fashion to improve approximation hardness of
several problems. We will study graph products in the form $(G\oplus H)*J$
where $G$, $H$ and $J$ are any graphs and $*$ and $\oplus$ denote two
different products. This form may look strange at first, but it will
be clear later that it in fact captures a ``generic'' graph
transformation technique that has been used a lot in the past
(cf. Section~\ref{sec: simplified subadditivity}).

The products we will study are the {\em tensor product} $G\times H$,
the {\em extended tensor product}\footnote{To the best of our knowledge,
  this product has not been considered before. Interestingly, it is
  mentioned in \cite[pp 42]{HammackIK11} as ``not worthy of
  attention''.} $G\hstrong H$, the {\em disjunctive product} $G\vee H$ and
the {\em lexicographic product} $G\cdot H$. These products produce a graph
whose vertex set is $V(G)\times V(H)=\{(u, v): u\in V(G), v\in V(H)\}$ with different edge sets. Their
exact definitions are not necessary at this point, but if you are
impatient, see Section~\ref{sec:prelim}.

Properties of a graph $G$ that we are interested in are the
{\em induced matching number} $\induce{G}$, the
{\em semi-induced matching number} $\sinduce{G}$ and the
{\em poset dimension number} $\dimension{G}$.
Informally, an induced matching of an undirected graph $G$ is a
matching $\mset$ of $G$ such that no two edges in $\mset$ are joined by an
edge in $G$. We let the induced matching number of $G$, denoted by
$\induce{G}$, be the size of the maximum induced matching.
See the formal definition in Section~\ref{sec:prelim}.
Definitions of {\em poset} and other graph properties are not needed in this section
and are deferred to Section~\ref{sec:prelim}.

%We will represent a {\em partially ordered set} (poset) by a directed acyclic graph $P$. We define its dimension, denoted by $\dimension{P}$, to be the minimum dimension of a Euclidean space that preserves the poset structure of $P$.
%
%The precise definitions are again not necessary here and can be found in Section~\ref{sec:prelim}.

Now that we have introduced all notations necessary, we are ready to
state our result. We show that for $f={\sf im}$ or $f={\sf sim}$ and
for an appropriate choice of product $\oplus$ and $*$, $f$ will be
{\em subadditive}, i.e., $f((G\oplus H)*J)\leq f(G*J)+f(H*J)$ for any graphs
$G$, $H$ and $J$.
For $f={\sf dim}$, we will get {\em almost} subadditivity instead.
The precise statement is as follows.

%Recall that $\times$, $\hstrong$, $\vee$ and $\cdot$ are graph products.

\begin{theorem}[(Almost) Subadditivity]\label{theorem:subadditive}\label{theorem:main}
For any undirected graphs $G$, $H$ and $J$ and a height-two poset $\vec P$,

\begin{align}
\induce{(G\vee H)\times J} &\leq \induce{G\times J}+\induce{H\times J}\label{eq:subadditive induce}\\
\sinduce{(G\vee H)\times J} &\leq \sinduce{G\times J}+\sinduce{H\times J}\label{eq:subadditive semi-induce}\\
\dimension{(G\cdot H)\hstrong \vec P} &\leq \dimension{G\hstrong \vec P}+\chi(G)\dimension{H\hstrong \vec P}+\dimension{\vec P}\label{eq:subadditive dimension}
\end{align}

\noindent where $\chi(G)$ is the chromatic number of $G$.
\end{theorem}

%Note that we will show that $G\hstrong P$ is a poset as long as $P$ is a poset. So, it makes sense to define its dimension.

Note that Eq.\eqref{eq:subadditive dimension} suggests that the poset
dimension number is almost subadditive in the sense that if $\chi(G)$
and $\dimension{\vec P}$ are small, then it will be subadditive (with a small multiplicative factor). This will
be the case when we use it to prove the hardness of approximation; see
Section~\ref{sec: applications} for more detail.

\danupon{For later: We may want to note that Eq.\eqref{eq:subadditive induce} is tight in the sense that there is an example that achieves this value: I think $(K_2\vee K_2)*K_2$ should work. Note that it's not that this is true for all graphs. I think a counter example is $(C_4\vee C_4)*K_2$. Can we say the same thing for the other inequality?}

\paragraph{Organization} In Subsection~\ref{sec:intro proof}, we give
an example of how Theorem~\ref{theorem:subadditive} plays a role in proving hardness of approximation.
In Subsection~\ref{sec:intro approx}, we discuss problems whose
hardness of approximation can be obtained via our
technique. In Section~\ref{sec:prelim}, we define formal terms needed in the
rest of the paper. We then prove Theorem~\ref{theorem:subadditive} in
Section~\ref{sec: simplified subadditivity} (for the special case which gives more intuition) and
Section~\ref{sec: subadditivity full} (for the general case).
Section~\ref{sec: applications} and \ref{sec:rest} show the
hardness results. Our results are summarized in Fig.~\ref{fig:summary}.

\begin{figure*}
%\setlength{\fboxsep}{0pt}
%\centering{\fbox{\begin{minipage}{\textwidth}
\begin{center}
\begin{tikzpicture}[auto, node distance=0.1cm]
    \scriptsize
    % Place nodes
    \tikzstyle{mybox}=[rectangle, draw=black, text centered, text=black, text width=3cm]
    %\tikzstyle{myarrow}=[->, >=open triangle 90, thick]
    %\tikzstyle{line}=[-, thick]
%
    \node (independent) [mybox] {Independent Set and Coloring};
    \node (sinduce) [mybox, right=4 of independent] {Bip. Semi-Ind. Matching\\ (Max. Expand. Seq.)};
    \node (induce) [mybox, above=0.4 of sinduce] {Bipartite Induced Matching};
    \node (dimension) [mybox, below=1.5 of sinduce] {Poset Dimension};
    \node (packing) [mybox, right = 2 of induce] {Bip. Ind. Packing};
    \node (mfs) [mybox, right = 2 of sinduce] {(0/1)-Feas. Subsys.};
    \node (donation) [mybox, above = 0.3 of mfs] {Donation Center};
    \node (pricing1) [mybox, below = 0.3 of mfs] {Single-Minded Pricing};
    \node (pricing2) [mybox, below = 0.3 of pricing1] {Unit-Demand Pricing};
    \node (box) [mybox, right = 2 of dimension] {Boxicity and Cubicity};
    \node (threshold) [mybox, below = 0.3 of box] {Threshold dimension};
    %\draw [-, very thick, red, double=red]  (independent) edge node {\bf Subadditivity (Thm~\ref{theorem:subadditive})} ([xshift=-0.5cm] sinduce.west);
    \draw[->, very thick, red, double=red]  (independent) -- ([xshift=-0.5cm] sinduce.west);
    \node (theorem) [right=0.7 of independent,fill=red!25,text width=1.8cm, text centered]{Subadditivity (Thm~\ref{theorem:subadditive})};
    \draw[->, thin, blue] ([xshift=-0.5cm] sinduce.west) -- (sinduce);
    \draw[->, thin, blue] ([xshift=-0.5cm] sinduce.west) |- (induce.west);
    \draw[->, thin, blue] ([xshift=-0.5cm] sinduce.west) |- (dimension.west);
    \draw[->, thick, dotted, gray] (induce) -- (packing);
    \draw[->, thick, dotted, gray] (sinduce) -- (mfs);
%    \draw[->, very thick, green, dashed] (sinduce.east) -| ([xshift=-1cm] pricing1.west) -- (pricing1.west);
%    \draw[->, very thick, green, dashed] (sinduce.east) -| ([xshift=-1cm] pricing2.west) -- (pricing2.west);
    \draw[->, thick, dotted, gray] (sinduce.east) -| ([xshift=-1cm] pricing1.west) -- (pricing1.west);
    \draw[->, thick, dotted, gray] (sinduce.east) -| ([xshift=-1cm] pricing2.west) -- (pricing2.west);
    \draw[->, thick, dotted, gray] (sinduce.east) -| ([xshift=-1cm] donation.west) -- (donation.west);
    \draw[->, thick, dotted, gray] (dimension) -- (box);
    \draw[->, thick, dotted, gray] (dimension.east) -| ([xshift=-1cm] threshold.west) -- (threshold.west);
    \node (blank) [below=0.8 of independent.west, text width=0cm, text=blue] {};
    \node (subadditivity) [right=0.5of blank, text width=5cm, text=red] {{\scriptsize Sec. \ref{sec: simplified subadditivity},\ref{sec: subadditivity full}}};
    \draw [-, very thick, red, double=red] ([xshift=-0.5cm] subadditivity.west) -- (subadditivity.west);
    \node (mainreduction) [below=0.2 of subadditivity, text width=5cm, text=blue] {{\scriptsize Sec. \ref{sec: applications}}};
    \draw [-, thin, blue] ([xshift=-0.5cm] mainreduction.west) -- (mainreduction.west);
    \node (rest) [below=0.2 of mainreduction, text width=5cm, text=gray] {{\scriptsize Sec. \ref{sec:rest}}};
    \draw [-, thick, gray, dotted] ([xshift=-0.5cm] rest.west) -- (rest.west);
%    \node (pricing) [below=0.2 of rest, text width=5cm, text=green] {{\scriptsize Sec. \ref{sec:rest}}};
%    \draw [-, very thick, green, dashed] ([xshift=-0.5cm] pricing.west) -- (pricing.west);
\end{tikzpicture}
\end{center}
%\end{minipage}}}
\caption{Paper Outline. Problems in the third column are defined in Section~\ref{sec:prelim}. Other problems are defined in Section~\ref{sec:rest}.}\label{fig:summary}
\end{figure*}
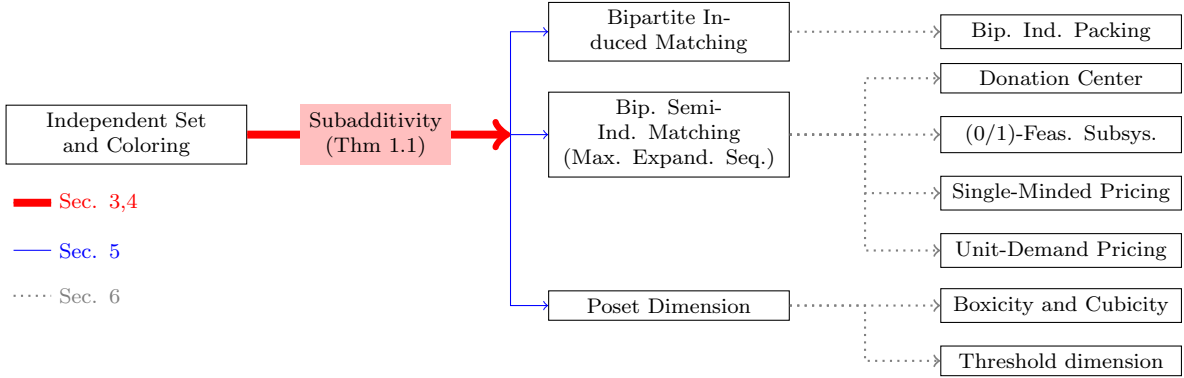

\subsection{Using Subadditivity (Theorem~\ref{theorem:subadditive}).}\label{sec:intro proof}
We now sketch the proof idea of the $n^{1-\epsilon}$ hardness of the bipartite induced matching problem, where $n$ is the number of vertices, which shows how Theorem~\ref{theorem:subadditive} plays a role in proving the hardness of approximation. The full proof appears in Section~\ref{sec: applications}. We build on the idea of \cite{ElbassioniRRS09} and apply
Theorem~\ref{theorem:subadditive} with $J=K_2$.

For any graph $G$, let $\alpha(G)$ be the size of the maximum independent set. We use the following connection between independence and induced matching numbers which was implicitly shown in \cite{ElbassioniRRS09}.
% \fullversion
% {(For completeness, we provide the proof in Appendix~\ref{sec:proofs-apps}.)}
% {(The proof is provided in the full version of our paper.)}
%
\begin{align}
\alpha(G)\leq \induce{G\hstrong K_2} \leq \induce{G\times K_2}+\alpha(G).\label{eq:elbassioni}
\end{align}
If $\induce{G\times K_2}$ is relatively small, i.e.,
$\induce{G\times K_2}=O(\alpha(G))$, then we will already have the
hardness of $n^{1-\epsilon}$ using the hardness of approximating the
independent set number (e.g., \cite{Hastad96}).
However, $\induce{G\times K_2}$ could be as large as $|V(G)|$, and in
such case, we do not get any hardness result (not even NP-hardness).
To remedy this, we apply Eq.\eqref{eq:subadditive induce} in
Theorem~\ref{theorem:subadditive} repeatedly to show that
%
% ***SPACE***: Should use align* if space allows
%
\begin{align*}
\induce{(G^k)\times K_2}
  \leq \induce{(G^{k-1}) \times K_2} +
       \induce{G\times K_2}
  \leq \ldots \leq \induce{G\times K_2}k
%\label{eq:induce power upper bound}
\end{align*}
where $G^k=G\vee G\vee \ldots \vee G$ is a $k$-fold product. Combining this with Eq.\eqref{eq:elbassioni}, we have
\begin{align*}
%\alpha^k(G)=
\alpha(G^k)\leq \induce{(G^k)\hstrong K_2}\leq \induce{G\times K_2}k+\alpha(G^k)\,.
%\leq |V(G)|k+\alpha^k(G)\,.
\end{align*}
It is well known that $\alpha(G^k)=(\alpha(G))^k$. Thus, after
applying the $k$-fold product, the term $\induce{G \times K_2}$ only
grows linearly in terms of $k$, while the term $\alpha(G^k)$ grows
exponentially! So, by choosing large enough $k$, the
induced matching number and the independence number coincide,
i.e., $\induce{(G^k)\hstrong K_2} \approx \alpha(G^k)$. Now, any
hardness of approximating the independence number implies immediately
roughly the same hardness of approximating the induced matching
number. Note that it can be checked with the definition of $\hstrong$ in Section~\ref{sec:prelim} that $(G^k)\hstrong K_2$ is a bipartite graph, so we get the hardness of the bipartite induced matching problem as desired.

\subsection{List of Applications.}\label{sec:intro approx}
The almost subadditivity properties shown in
Theorem~\ref{theorem:subadditive} are useful in proving many other
hardness of approximation results as listed in the following
theorem.

%The following theorem summarizes the hardness of approximation results that can be obtained using subadditivity inequalities.

\begin{theorem}\label{theorem:hardness}
For any $\epsilon>0$, unless $\ZPP=\NP$ there is no
$n^{1-\epsilon}$-approximation algorithm, where $n$ is the number of
vertices in the input graph, for the following problems: bipartite induced and semi-induced matching (a.k.a. maximum expanding sequence), poset dimension, bipartite independent packing,  donation center location, maximum feasible subsystem with 0/1 coefficients, boxicity, cubicity and threshold dimension.

Additionally, there is no $d^{1/2-\epsilon}$-approximation algorithm for the induced and semi-induced matching problem on $d$-regular bipartite graphs.

Moreover, unless $\NP \subseteq {\sf ZPTIME}(n^{\poly \log n})$, there is no
$\log^{1-\epsilon} m$-approximation algorithm and no
$k^{1/2-\epsilon}$-approximation algorithm for the single-minded and
unit-demand pricing problems, where $m$ is the number of consumers and
$k$ is the maximum consumer's set size.
\end{theorem}

\noindent Fig.~\ref{fig:summary} summarizes the results and reductions. Almost all reductions are done in a systematic way.
We take a hard instance of the maximum independent set problem or the graph
coloring problem. Then we perform an appropriate graph product and
output a result. Depending on the applications, we need hard
instances in various forms. All the results here are essentially {\em tight} except the hardness of $k^{1/2-\epsilon}$ and $d^{1/2-\epsilon}$ of the pricing problems and the induced-matching problem on $d$-regular graphs, respectively.

\medskip\noindent{\em Remark on Stronger Results.} We note that for problems having $n^{1-\epsilon}$-hardness stated in Theorem~\ref{theorem:hardness}, we can actually prove a slightly stronger result: for any $\gamma >0$, unless $\NP \subseteq {\sf ZPTIME}(n^{\poly \log n})$, there is no $\frac{n}{2^{(\log n)^{3/4+\gamma}}}$-approximation algorithm. This is achieved by applying the result of Khot and Ponnuswami~\cite{KhotPonnuswami}.
%
%following problems: bipartite induced and semi-induced matching (a.k.a. maximum expanding sequence), bipartite independent packing,  donation center location, and maximum feasible subsystem with 0/1 coefficients. Furthermore, for any $\gamma > 0$, unless $\NP \subseteq {\sf BPTIME}(n^{\poly \log n})$, there is no $\frac{n}{2^{(\log n)^{3/4+ \gamma}}}$ approximation algorithm for poset dimension, boxicity, cubicity, and threshold dimension. However,
%
For the sake of presentation, we will prove only $n^{1-\epsilon}$-hardness using the result of H{\aa}stad~\cite{Hastad96}.

%Also note that all results, except for the pricing problems, give a

\medskip \noindent Now, we are ready to discuss our hardness results. We provide the formal definitions of the first three problems in  Section~\ref{sec:prelim} and provide the definitions of the remaining problems in Section~\ref{sec:rest}.

\paragraph{Bipartite Induced Matching} One immediate application is
the {\em tight} $n^{1-\epsilon}$ hardness of {\em approximating the
  induced matching problem on bipartite graphs}, improving
upon the previous best hardness of $n^{1/3-\epsilon}$ \cite{ElbassioniRRS09}. Our
result also implies the tight hardness of the independent packing of
graphs \cite{CameronH06} as well.
A similar technique can also be used to show that the induced
matching problem on {\em $d$-regular bipartite graphs} is hard to
approximate to within a factor of $d^{1/2-\epsilon}$, improving upon the
{\sf APX}-hardness of~\cite{DuckworthMZ05,Zito99}. (This result
is not tight as the best known upper bound is $\Theta(d)$~\cite{GotthilfL05}.)

The notion of induced matching has naturally arisen in discrete
mathematics and computer science. It is, for example, studied as the
``risk-free'' marriage problem in \cite{StockmeyerV82} and is a
subtask of finding a {\em strong edge coloring}. This problem and its
variations also have connections to various problems such as {\em
  storylines extraction} \cite{KumarMS04} and {\em network scheduling,
  gathering and testing}
(e.g. \cite{EvenGMT84,StockmeyerV82,JooSSM10,Milosavljevic11,BonifaciKMS11}).
%
%maximum feasible subsystem \cite{ElbassioniRRS09}, maximum expanding sequence \cite{BriestK11}, unit-demand min-buying and single-minded pricing} \cite{BriestK11},
%
The problem was shown to be \NP-complete in
\cite{StockmeyerV82,Cameron89} and was later shown to be hard to
approximate to within a factor of $n^{1/3-\epsilon}$ unless $\NP=\ZPP$ by
\cite{ElbassioniRRS09}. We have sketched the proof of the tight hardness
of $n^{1-\epsilon}$ in Section~\ref{sec:intro proof}, and more detail
can be found in Section~\ref{sec: applications}.

%\paragraph{Independent Packing of Graphs} Cameron and Hell~\cite{CameronH06} considered the problem of packing independent copies of graph $H$ into $G$. This problem is at least as hard as computing the maximum induced matching in graph $G$, so our hardness result immediately implies $|V(G)|^{1-\epsilon}$ hardness of this problem.\danupon{What is previously known?}

\paragraph{Bipartite Semi-induced Matching (a.k.a. Maximum Expanding Sequence)}
%
%We note the following a variation of an induced matching called an {\em semi-induced matching}.
%
%In this case, in addition to the input bipartite graph $G=(U\cup U', E)$, we are given the ordering $\sigma$ of vertices in $U$. For any matching $M$, we order its edges using left vertices, e.g. $M=\{u_1u'_1, u_2u'_2, \ldots\}$ where for any $i<j$, $u_i$ appears before $u_j$ in ordering $\sigma$. We then say that $M$ is an {\em semi-induced matching} if for any $i<j$, there is no edge between $u_j$ and $u'_i$.
%
%\danupon{This is still not consistent with above} To be precise, given a bipartite graph $G= ([n], [n'],E)$, a collection $\mset$ of edges is a {\bf semi-induced matching} if for any $(i,a')(j,b')\in \mset$ such that $i< j$, there is no edge of the form $ib' \in E(G)$ (but we allow an edge $ja' \in E(G)$). For any bipartite graph $G$, we denote by $\sinduce{G}$ the maximum integer $r$ such that there exist semi-induced matching of size $r$ in $G$.
%
The same technique used in proving the hardness of the bipartite
induced matching problem can be extended (with some additional work) to its interesting
variation which captures a few other problems. This
variation was introduced independently by Briest and Krysta
\cite{BriestK11} as the {\em maximum expanding sequence problem} and
by Elbassioni et~al.~\cite{ElbassioniRRS09} as the
{\em bipartite semi-induced matching problem}.
There it was used as an intermediate problem that captures
the hardness of some important algorithmic pricing problems and the
maximum feasible subsystem problem, which we will see shortly.

\paragraph{Poset Dimension} Another immediate application of
Theorem~\ref{theorem:subadditive} is the tight $n^{1-\epsilon}$ hardness
of approximating the poset dimension, improving upon the hardness of
$n^{1/2-\epsilon}$ of Hegde and Jain \cite{HJ07}.

The notion of poset dimension has long been a central subject of study
in discrete mathematics (e.g., \cite{TrotterBook01}) and has
connections with many other notions, e.g., transitive-closure spanners
\cite{Raskhodnikova10} as well as the boxicity and the threshold dimension of
graphs~\cite{AdigaBC10}.
A variant called the {\em fractional dimension} is shown to have a
connection to some classical scheduling problems
(e.g., \cite{AmbuhlMMS08}). We note that our technique also
implies the tight hardness of approximating the fractional dimension of
posets.\danupon{Later: show detail of fractional dimension hardness}

The computational complexity of the poset dimension problem was one of
the twelve outstanding open problems in Garey and Johnson's
treatise on NP-completeness~\cite{GareyJ79}. It was independently
shown by Yannakakis \cite{Yannakakis82} and Lawler and Vornberger
\cite{LawlerV81} that the problem is \NP-complete. More recently,
Hegde and Jain showed that the problem is hard to approximate to
within a factor of $n^{1/2-\epsilon}$ unless $\NP=\ZPP$.
Here we resolve the approximability of this problem using graph
products.

We note that our result actually implies the tight hardness of
approximating the dimension of {\em adjacency poset}. This is the
notion, along with the {\em incidence poset}, of the dimension of posets
arising from graphs. They have been extensively studied due to their
connections with graph's planarity and chromatic number (e.g.,
\cite{FelsnerLT10,Schnyder89,Schnyder90}).

\paragraph{Unit-demand Min-buying (\UDP) and Single-minded (\SMP)
  Pricing} A result that is not so immediate from
Theorem~\ref{theorem:subadditive} is the hardness of approximating the
two combinatorial pricing problems, called \UDP and \SMP.
%%%% BUN: expanding sequences comes out of nowhere....
% This gives alternate proofs of the tight hardness results (the first tight result was recently shown by Chalermsook et~al.~\cite{ChalermsookPricing}). Moreover, our proof confirms the role of expanding sequences in the hardness of pricing problems suggested in \cite{BriestK11}.
The tight hardness of these two problems were recently proved by Chalermsook et~al.~\cite{ChalermsookPricing}. Here we give alternate proofs of the results in~\cite{ChalermsookPricing} by employing the tight hardness of an intermediate problem -- the maximum expanding sequence problem, thus confirming the role of expanding sequences in the hardness of pricing problems suggested in \cite{BriestK11}.

%that the maximum expanding sequence problem is the source of hardness of some pricing problems \cite{BriestK11}.

%Moreover, our proof confirms that the maximum expanding sequence problem is the source of hardness of some pricing problems \cite{BriestK11}.\danupon{I added this, please check!!!}

%In \UDP, we have a collection of items $[n]$ and consumer sets $S_1,\ldots, S_m$ where $S_j \subseteq [n]$. Each consumer $S_c$ is associated with budget $B_c$. Once the price $p: [n] \rightarrow \R_{+}$ is fixed, each consumer $c$ buys the cheapest item in $S_c$ if the price of such item is at most $B_c$; otherwise, the consumer buys nothing. Our goal is to set prices $p$ so that the profit is maximized. In \SMP, the setting is the same, except that now each consumer $c$ buys the whole set of items if $\sum_{i \in S_c} p(i) \leq B_c$; otherwise, the consumer $c$ buys nothing.

Both \UDP and \SMP are among the most basic pricing problems in
the literature and have received a lot of attention (e.g.,
\cite{BriestK11,GuruswamiPricing,Rusmevichientong03,Rusetal,Balcan-Blum-Pricing,ChalermsookPricing}).
Briest and Krysta showed the hardness of $\log^\epsilon m$, assuming
the (rather non-standard) hardness of the bounded-degree bipartite
independent set problem.
% They introduced the maximum expanding sequence problem, showed a reduction from it to \UDP and \SMP, and finally proved the hardness of approximating maximum expanding sequence.
% BUN: The sentence is not well phrased. XXX
To prove the hardness of \UDP and \SMP, they introduced the maximum expanding sequence problem and showed that it can be reduced to \UDP and \SMP.
Thus, by proving the hardness of the maximum expanding sequence problem, they obtain the hardness results for these pricing problems.
% for \UDP and \SMP.
%
As mentioned in \cite{BriestK11}, this ``indicates that expanding sequences
are a common source of hardness for quite different combinatorial
pricing problems''.
Chalermsook et~al.~\cite{ChalermsookPricing} recently showed the tight
hardness of $\log^{1-\epsilon} m$ of these problems, assuming a
standard assumption
(i.e., $\NP\not\subseteq \DTIME(n^{\poly\log n})$), by avoiding the
maximum expanding sequence problem and proving the hardness of
\UDP and \SMP directly.
In this paper, we revisited Briest and Krysta's original proposal to
prove the hardness of these problems via the maximum expanding
sequence problem. We show the hardness of approximation result for a special case\footnote{Our special case is different
  from that in~\cite{BriestK11}. (If we use their special case, we only obtain
  the hardness of $\log^{1/2-\epsilon} m$.) In particular, our special case requires that the input must be in some specific form of a {\em graph product.}} of the maximum
expanding sequence problem via graph products, which then implies the hardness of \UDP and
\SMP.
Our results confirm that the maximum expanding sequence problem
is indeed the main source of hardness for both \UDP and \SMP.

\paragraph{Maximum Feasible Subsystem with 0/1 Coefficients}
In the maximum feasible subsystem (\MFS) problem, we are given a
system of $m$ linear inequalities $\ell_i \leq a_i^T x \leq \mu_i$, where
$a_i \in \set{0,1}^n$, and $\ell_i, \mu_i \in \R_{+}$.
The goal is to find a non-negative solution $x \in
\R^n_{+}$ that maximizes the number of constraints satisfied.
%
%A natural relaxation of \MFS is to consider a bicriteria approximation algorithm. We say that a constraint $i$ is $\beta$-satisfied by solution $x$ if $\ell_i \leq a_i^T x \leq \beta u_i$. A solution $x \in \R_{+}^n$ is an $(\alpha, \beta)$ approximation algorithm if at least $1/\alpha$-fraction of the linear constraints is $\beta$-satisfied by $x$.
%
When coefficients are not necessarily 0/1, the
$m^{1-\epsilon}$-hardness of \MFS was proved by Guruswami and
Raghavendra \cite{GuruswamiR09}\footnote{
Indeed, Guruswami and Raghavendra proved the hardness of the
{\em Max 3LIN} problem, which can be seen as a special case of \MFS}.
Elbassioni et~al.~\cite{ElbassioniRRS09} showed that even in the 0/1-coefficient
case, the problem has the hardness of $m^{1/3-\epsilon}$. They
actually showed a gap-preserving reduction from the semi-induced matching problem to 0/1-\MFS.  This means that our
hardness of the semi-induced matching problem immediately
implies the tight hardness of $m^{1-\epsilon}$ for any $\epsilon$ for
\MFS.
This hardness result holds even when we allow a violation of the upper bounds by
at most an $O(n)$ factor.\danupon{I'm being vague here.}
We also show the tight hardness of $\log^{1-\epsilon} (\max_{i \in [n]} \ell_i)$, matching an upper bound in~\cite{ElbassioniRRS09}.
%$\log^{1-\epsilon} L$, where $L= \max_{i \in [n]} \ell_i$. \danupon{We have to be consistent. We used $n^{1/3-\epsilon}$ in the theorem.}
% REMOVE TO SAVE SPACE
%These results are tight as they match the upper bounds in~\cite{ElbassioniRRS09}.

\paragraph{Boxicity, Cubicity and Threshold Dimension of Graphs} %For any graph $G$, the boxicity of $G$, denoted by ${\sf box}(G)$, is the minimum number of dimension $d$ such that graph $G$ can be represented as an intersection graph of $d$-dimensional boxes.
The notion of {\em boxicity} arose from the study of
{\em intersection graphs}.
It was introduced by Roberts~\cite{RobertsBoxicity} and studied
extensively in discrete mathematics.
It also has connections to important graph theoretic
measures such as treewidths~\cite{ChandranS07},
genuses~\cite{FelsnerLT10,AdigaBC11}, crossing
numbers~\cite{AdigaCM11} and the maximum degree of
graphs~\cite{ChandranFS08}.
In computer science, optimization problems on graphs with
bounded boxicity (e.g., graphs arising from intervals, rectangles, and
cubes) have also received a lot of attention.
%
%Similarly to the notion of boxicity, cubicity of $G$, denoted by ${\sf cub}(G)$ is the minimum number of dimension $d$ such that graph $G$ is an intersection graph of $d$-dimensional cubes.
%
Adiga et~al.~\cite{AdigaBC10} showed that the hardness of
approximating poset dimension implies the hardness of approximating
boxicity  (and other closely related measures called {\em cubicity}
and {\em threshold dimension}).
Combining this with our hardness of approximating poset dimension, we get the tight  $n^{1-\epsilon}$-hardness
 for all these problems.

%Our hardness of approximating poset dimension, combined with the reduction in~\cite{AdigaBC10}, implies tight hardness of $n^{1-\epsilon}$ for all these problems.

%\paragraph{Threshold Dimension of Graphs} A graph $G$ is a {\bf threshold graph} if there is a number $\gamma$ and a weight function $w: V(G) \rightarrow \R$ such that, for any two vertices $u,v \in V(G)$, we have $(u,v) \in E(G)$ if and only if $w(u) +w(v) \geq \gamma$. A {\bf threshold cover} of graph $G$ is a set of threshold graphs $\set{G_i}_{i=1}^k$ on the same vertex set $V(G)$ such that $E(G) = \bigcup_{i=1}^k E(G_i)$. The threshold dimension of $G$ is the minimum number $d$ such that a threshold cover of size $d$ exists. Our result on the hardness of computing poset dimension, combined with the reduction in~\cite{AdigaBC10}, implies that it is hard to approximate the threshold dimension of graphs to within a factor of $|V(G)|^{1-\epsilon}$.

\paragraph{Donation Center Location} In this problem, we are given a set of agents and a set of centers, where
agents have preferences over centers and centers have capacities. The goal is to open a subset of centers and to assign a maximum-sized subset of agents to their most-preferred opened centers, while respecting the capacity constraints.
%See the formal definition in Section~\ref{sec:rest}.

Huang and Svitkina~\cite{HuangS09} introduced this problem and showed an $n^{1/2-\epsilon}$ approximation hardness by a reduction from the maximum independent set problem. We show a straightforward reduction from the semi-induced matching problem, hence giving the tight $n^{1-\epsilon}$-hardness. This hardness result holds even when all agents have the same preference over centers, and each center has unit capacity.

\section{Preliminaries}\label{sec:prelim}

In this section, we define graph products and graph properties we will use.
%
%We defer the definition of semi-induced matching to Appendix \ref{sec: subadditivity full} as we do not need it here.
%
For any directed or undirected graph $G$, we use $V(G)$ and $E(G)$ to denote its vertex and edge sets, respectively. Note that if $G$ is directed, then it is possible that, for some $u, v \in V(G)$, $uv\in E(G)$ but $vu\notin E(G)$. This is not the case when $G$ is undirected.
When a graph is directed, we shall put an arrow above $\vec G$ to emphasize that $\vec G$ is a directed graph.

\begin{definition}[Graph Products]\label{def:product}
%
% For any directed or undirected graph $G$, let $V(G)$ and $E(G)$ be its vertex and edge sets, respectively. Note that if $G$ is directed, then it is possible that, for some $u, v \in V(G)$, $uv\in E(G)$ and $vu\notin E(G)$. This is not the case when $G$ is undirected.
%
A {\em graph product} is a binary operation that constructs from two graphs $G$ and $H$ a graph with vertex set
$V(G)\times V(H)=\{(u, v): u\in V(G), v\in V(H)\}$, and the edge set is determined by adjacency of vertices of $G$ and $H$.
%%%%%%%%%%%
% The graph products we study include the {\em tensor product} $G\times H$, the {\em extended tensor product}, $G\hstrong H$, the {\em disjunctive product} $G\vee H$ and the {\em lexicographic product} $G\cdot H$.
%
% The edge sets of these products are as in Fig.~\ref{fig:product definition}.
%
%%%%%%%%%%%
%% The graph products of $G$ and $H$ we study include the {\em tensor product} $G\times H$, {\em extended tensor product}\footnote{To the best of our knowledge, this product has not been considered before. Interestingly, it is mentioned in \cite[pp 42]{HammackIK11} as ``not worthy of attention''.} $G\hstrong H$, {\em disjunctive product} $G\vee H$ and {\em lexicographic product} $G\cdot H$.
%% %
%% The vertex set of the resulting graph of any of these products is  $V(G)\times V(H)=\{(u, v): u\in V(G), v\in V(H)\}.$ The edge sets are:
%\begin{align*}
%\mbox{\footnotesize (tensor)} && E(G\times H) &= \{(u, a)(v, b) : \mbox{$uv\in E(G)$ \underline{and} $ab\in E(H)$}\}\\
%\mbox{\footnotesize (extended tensor)} && E(G\hstrong H) &= \{(u, a)(v, b) : \mbox{{\bf (}$uv\in E(G)$ \underline{or} $u=v${\bf )} \underline{and} $ab\in E(H)$}\}\\
%\mbox{\footnotesize (disjunctive)} && E(G\vee H) &= \{(u, a)(v, b) : \mbox{$uv\in E(G)$ \underline{or} $ab\in E(H)$}\}\\
%\mbox{\footnotesize (lexicographic)} && E(G\cdot H) &= \{(u, a)(v, b) : \mbox{$uv \in E(G)$ \underline{or} {\bf (}$u=v$ \underline{and} $ab\in E(H)${\bf )}}\}\,.
%\end{align*}
%
\end{definition}

The graph products we study include the {\em tensor product} $G\times H$, the {\em extended tensor product} $G\hstrong H$, the {\em disjunctive product} $G\vee H$ and the {\em lexicographic product} $G\cdot H$.
The edge sets of these products are as belows.
\begin{align*}
\mbox{(tensor)} && E(G\times H) &= \{(u, a)(v, b) : \mbox{$uv\in E(G)$ \underline{and} $ab\in E(H)$}\}\\
\mbox{(extended tensor)} && E(G\hstrong H) &= \{(u, a)(v, b) : \mbox{{\bf (}$uv\in E(G)$ \underline{or} $u=v${\bf )} \underline{and} $ab\in E(H)$}\}\\
\mbox{(disjunctive)} && E(G\vee H) &= \{(u, a)(v, b) : \mbox{$uv\in E(G)$ \underline{or} $ab\in E(H)$}\}\\
\mbox{(lexicographic)} && E(G\cdot H) &= \{(u, a)(v, b) : \mbox{$uv \in E(G)$ \underline{or} {\bf (}$u=v$ \underline{and} $ab\in E(H)${\bf )}}\}\,
\end{align*}

%The edge sets of these products are as in Fig.~\ref{fig:product definition}.

% \begin{figure*}
% \setlength{\fboxsep}{0pt}
% \centering{
% \fbox{
% \begin{minipage}{\columnwidth}
% \noindent
% \begin{align*}
% \mbox{\footnotesize (tensor)} && E(G\times H) &= \{(u, a)(v, b) : \mbox{$uv\in E(G)$ \underline{and} $ab\in E(H)$}\}\\
% \mbox{\footnotesize (extended tensor)} && E(G\hstrong H) &= \{(u, a)(v, b) : \mbox{{\bf (}$uv\in E(G)$ \underline{or} $u=v${\bf )} \underline{and} $ab\in E(H)$}\}\\
% \mbox{\footnotesize (disjunctive)} && E(G\vee H) &= \{(u, a)(v, b) : \mbox{$uv\in E(G)$ \underline{or} $ab\in E(H)$}\}\\
% \mbox{\footnotesize (lexicographic)} && E(G\cdot H) &= \{(u, a)(v, b) : \mbox{$uv \in E(G)$ \underline{or} {\bf (}$u=v$ \underline{and} $ab\in E(H)${\bf )}}\}\,
% \end{align*}
% \end{minipage}}}
% \caption{Definition of graph products.}\label{fig:product definition}
% \end{figure*}

\begin{definition}[Induced Matching Number, $\induce{G}$]\label{def:induce}
Let $G=(V, E)$ be any undirected graph. The {\em induced matching} of $G$ is the set of edges $\mset \subseteq E(G)$ such that $\mset$ is a matching and no two edges in $\mset$ are joined by an edge in $G$, i.e., for any edges $uu', vv' \in \mset$, $G$ has none of the edges in $\set{uv, uv', u'v, u'v'}$.
The {\em induced matching number} of $G$, denoted by $\induce{G}$, is the cardinality of the maximum-cardinality induced matching of $G$.
\end{definition}

%%%%%%%%%%%%%%%%%%%%%%%%%%%%%%%%%%%%%%%%%%%%%%%%%%%%%%%%%%%%%%%%%%%%%%
%%%%%%%%%%%%%%%%%%%%%%%%%%%%%%%%%%%%%%%%%%%%%%%%%%%%%%%%%%%%%%%%%%%%%%
%%%%%%%%%%%%%%%%%%%%%%%%%%%%%%%%%%%%%%%%%%%%%%%%%%%%%%%%%%%%%%%%%%%%%%

% We note the following variation of an induced matching called an {\em semi-induced matching}.
% For any finite set $S$, a {\em total order $\sigma$ on $S$} is a bijection $\sigma: S \rightarrow [|S|]$. Note that one can view $\sigma$ as a total ordering $\sigma(x_1)<\sigma(x_2)<\ldots <\sigma(x_n)$ of elements $x_i$ in $S$, where $x_i$ is such that $\sigma(x_i)=i$, for all $i$.

% $\sigma^{-1}(1), \sigma^{-1}(2), \ldots, \sigma^{-1}(|S|)$ of elements in $S$.

Now, we shall define a variant of an induced matching called a semi-induced matching; this notion is with respect to a total order.
For any finite set $S$, a {\em total order $\sigma$ on $S$} is a bijection $\sigma: S \rightarrow [|S|]$.
The total order $\sigma$ gives an ordering of elements $S$ as we may order elements $x_i\in S$ so that  $\sigma(x_1)<\sigma(x_2)<\ldots <\sigma(x_n)$, where $x_i$ is such that $\sigma(x_i)=i$, for all $i$.

\begin{definition}[Semi-induced Matching, ${\sf sim}(G)$]\label{def:semi induced number}
Given any graph $G=(V,E)$ and any total order $\sigma$, we say that a
matching $\mset$ is a $\sigma$-semi-induced matching if, for any pair of edges
$uu', vv'  \in \mset$ such that $\sigma(u) < \sigma(u')$ and
$\sigma(u) < \sigma(v) < \sigma(v')$, there are no edges $uv'$ and $uv$ in $E$.
\end{definition}

We can check if a matching $\mset$ is a $\sigma$-semi-induced matching as follows.
First, we order edges in $\mset$ as $u_1v_1, u_2v_2, \ldots, u_qv_q$ where, for any $i$, $\sigma(u_i)<\sigma(v_i)$ and $\sigma(u_1)<\sigma(u_2)<\ldots<\sigma(u_q)$.
Now, $\mset$ is a $\sigma$-semi-induced matching if and only if, for any $i<j$, $\mset$ has no edge in $\set{u_iu_j, u_i v_j}$.

For any graph $G$, we define $\sinducesigma{G}{\sigma}$ to be the size of a maximum $\sigma$-induced matching, and we define $\sinduce{G}= \max_{\sigma}\sinducesigma{G}{\sigma}$.
In the {\em semi-induced matching problem}, we are given a graph $G$, and the goal is to compute $\sinduce{G}$.

%%%%%%%%%%%%%%%%%%%%%%%%%%%%%%%
% This is previous definition  of maximum induced matching problem where \sigma is involved
%%%%%%%%%%%%%%%%%%%%%%%%%%%%%%%
\begin{comment}
\begin{definition}[Maximum Semi-induced Matching]\label{def:semi induced problem}
In this problem, we are given an $n$-vertex graph $G=(V,E)$ and ordering of vertices $\sigma: V(G) \rightarrow [n]$. Our goal is to compute the maximum-cardinality $\sigma$-semi-induced matching in $G$, i.e. $\sinducesigma{G}{\sigma}$.
%This problem takes a graph $G$ and a mapping $\sigma:V(G)\rightarrow [|V(G)|]$ as input. The goal is to compute $\sinducesigma{G}{\sigma}$.
\end{definition}
\end{comment}

%%%%%%%%%%%%%%%%%%%%%%%%%%%%%%%%%%%%%%%%%%%%%%%%%%%%%%%%%%%%%%%%%%%%%%
%%%%%%%%%%%%%%%%%%%%%%%%%%%%%%%%%%%%%%%%%%%%%%%%%%%%%%%%%%%%%%%%%%%%%%
%%%%%%%%%%%%%%%%%%%%%%%%%%%%%%%%%%%%%%%%%%%%%%%%%%%%%%%%%%%%%%%%%%%%%%

\begin{definition}[Partially Ordered Set (poset)]\label{def:poset}
% (the arrow above $\vec P$ is used to emphasize that it is a directed graph)
A directed graph $\vec P$ is a {\em partially ordered set (poset)}\footnote{Our definition follows \cite{Johnson81}.} if it is directed, acyclic and transitive (i.e., $uv,vw\in E(\vec P) \Rightarrow uw\in E(\vec P)$).
\end{definition}

An important class of posets is a height-two poset.
Given a poset $\vec P$, we say that a vertex is {\em minimal} (resp., {\em maximal}) if it has zero in-degree (resp., out-degree) in $\vec P$. A poset is a {\em height-two} poset if every vertex is minimal or maximal (some vertex might be both).

A poset can be defined by an ordering of $d$-dimensional points.
This leads to the notion of the poset dimension number of graphs.
For any $d$-dimensional points $p, q\in \reals^d$, we say that $p< q$ if for any $1\leq i\leq d$, $p[i]\leq q[i]$, and there exists $j$ such that $p[j]<q[j]$; otherwise, we say that $p\not< q$.

\begin{definition}[Poset Dimension Number, $\dimension{G}$\footnote{This concept is also sometime called the {\em order dimension} or the {\em Dushnik--Miller dimension}. It is usually defined using the notion of {\em realizer}. This is however equivalent to our ``embedding to $\reals^d$'' definition, and we use the embedding definition as it will be easier to use in our proofs.}]\label{def:dimension}
Let $\vec{P}$ be any poset. We say that a mapping $\phi: V(\vec P) \rightarrow \R^d$ {\em realizes} poset $\vec P$ if for any distinct vertices $u,v \in V(\vec P)$, $uv \in E(\vec P)$ if and only if $\phi(u) < \phi(v)$. The dimension of poset $\vec P$, denoted by $\dimension{\vec P}$, is the smallest integer $d$ such that there is a mapping $\phi: V(\vec P) \rightarrow \R^d$ that realizes $\vec P$.
\end{definition}

\medskip\noindent\textbf{Why Height-two Posets?} Recall that our main theorem only applies to height-two posets. There is a good reason for this. Note that it is not always the case that the product $G\hstrong \vec P$ between an undirected graph $G$ and poset $\vec P$ (that we use in Theorem~\ref{theorem:subadditive}) will result in a poset (an example is when $G$ and $\vec P$ are a path and a directed path of three vertices, respectively); so, the term $\dimension{G \hstrong \vec P}$ does not necessarily make sense. However, if $\vec P$ is a {\em height-two} poset, then $G\hstrong \vec P$ is always a poset (in fact, a height-two one). (We prove this fact in Lemma~\ref{lem:poset graph times poset} in Section~\ref{sec: subadditivity full}.) For this reason, Theorem~\ref{theorem:subadditive} is stated only for a height-two poset $\vec P$.

%In the {\em bipartite maximum induced matching problem} ({\sc Im}) we wish to approximate $\induce{G}$ for any input bipartite graph $G$. In the {\em poset dimension problem} ({\sc Dim}) we wish to approximate $\dimension{G}$ for any input poset $G$. Note that {\sc Im} is a maximization problem while {\sc Dim} is a minimization one.
%\danupon{We should make sure that we didn't use {\sc Im} and {\sc Dim} anywhere. We should also avoid the word embedding}

%Observe that for any $\sigma$, we have $\induce{G} \leq \sinducesigma{G}{\sigma} \leq \sinduce{G} $.

 % Sec 2: Preliminaries
\section{Proof of the Special case of Theorem~\ref{theorem:subadditive}}
\label{sec: simplified subadditivity}

%Let us remark a few things before proving Theorem~\ref{theorem:subadditive}. First,

In this section, we will focus on the special case of Theorem~\ref{theorem:subadditive} where we consider the products $(G\oplus H)*J$ when $J=K_2$ only. This is partly because the proof is more intuitive (and easier to illustrate by pictures) in this special case.
Moreover, this case itself is sufficient for our purpose in proving hardness results. We will also use $B[G\oplus H]$ instead of $(G\oplus H)*K_2$ to simplify the notation (more on this in Section \ref{sec:bip}).
Proofs of the general cases can be found in Section~\ref{sec: subadditivity full}. They are relatively short, perhaps easier to verify, and can be read without understanding any material in this section. However, since the proofs in the general case are less intuitive, some readers might find the intuition in this section helpful.

%Although our main theorem appears in terms of graph products of $G, H, J$, all of our current applications only need a very special case of the main theorem (i.e. when $J= K_2$). Therefore, in this section, we will only consider this special case and defer the proof of the general case to Appendix.

\subsection{Why Multiplying Graphs by $K_2$?}\label{sec:bip}
We first give a motivation of studying the graph products in this specific (and rather peculiar) form. First, notice that both $G\times K_2$ and $G\hstrong K_2$ are {\em bipartite}. To see this, let $V(K_2) = \set{1,2}$; so, the vertices of $G\times K_2$ and $G\hstrong K_2$ are in the form $(v, i)$ where $v\in V(G)$ and $i\in \{1, 2\}$ (e.g., Fig. \ref{fig:induce_13} and \ref{fig:dimension}). Thus, we can partition vertices of these graphs into $V_1\cup V_2$ where $V_i=\{(v, i): v\in V(G)\}$. There is no edge between vertices in the same partition in both graphs so they are bipartite.

So, we may think of the product $G\times K_2$ and $G\hstrong K_2$ as a ``transformation'' of graph $G$ into a bipartite graph. To emphasize this point of view and simplify the notation later on, we will write $\bipm[G]$ and $\bipp[G]$ instead of $G\times K_2$ and $G\hstrong K_2$ in this and next sections.

An intuitive way to think of $\bipm[G]$ and $\bipp[G]$ is to imagine the following construction. We construct $\bipm[G]$ by making two copies of each vertex $v\in V(G)$ to get vertices $(v, 1)$ and $(v, 2)$ in $\bipm[G]$. Then, for each edge $uv\in E(G)$, we add an edge between $(u, 1)$ and $(v, 2)$. To get $\bipp[G]$, we simply add an edge connecting the two copies of each vertex $v \in V(G)$; i.e., we add an edge between $(v, 1)$ and $(v, 2)$ for all $v\in V(G)$.

This transformation might now sound familiar to many readers. The graph $\bipm[G]$ is actually known as the {\em bipartite double cover} and has been used repeatedly as a natural way to transform any graph into a bipartite graph; for example, one can use this transformation to reduce the problem of computing cycle covers to the maximum bipartite matching problem.

In the context of posets, we will abuse the notation and define $\bip[G]=G \times \vec{K}_2$ and $\bipp[G]=G\hstrong \vec K_2$. Edges in these graphs always point from vertices in $V_1$ to vertices in $V_2$. The graph $\bip[G]$ also has its name in this case --  an {\em adjacency poset} \cite{FelsnerLT10}.

%In the remaining of this section, we prove Theorem~\ref{theorem:main} when $J= K_2$. The proof of the general form of this theorem appears in Appendix.

\begin{figure*}[t]
\setlength{\fboxsep}{0pt}
\centering{\fbox{\begin{minipage}{\textwidth}
\centering
\subfigure[{$G$, $H$, $\bipm[G]$ and $\bipm[H]$}]{
  \includegraphics[height=0.25\textwidth, clip=true, trim=3cm 8.6cm 6cm 2.3cm]{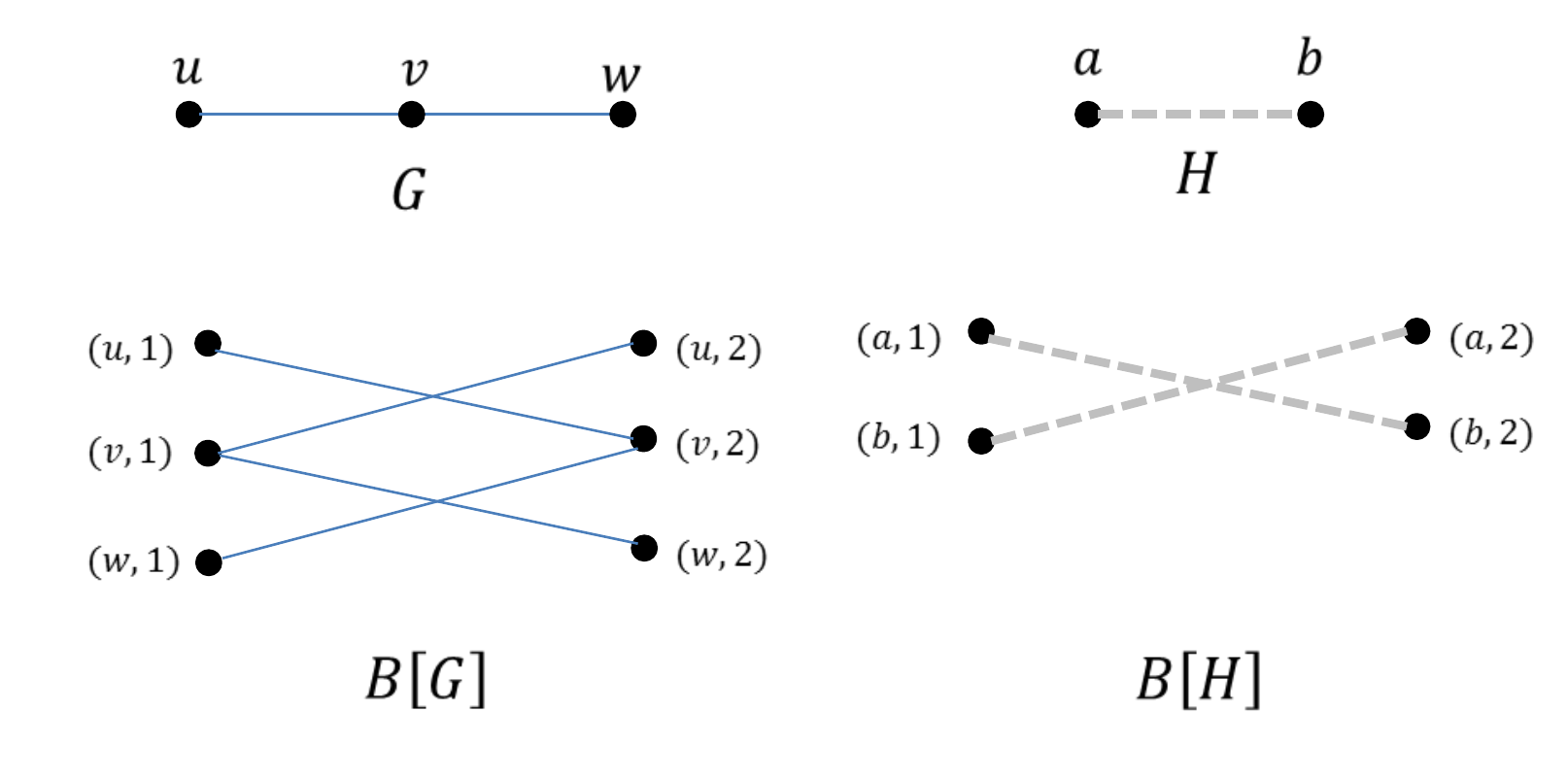}
    \label{fig:induce1}
}
\subfigure[{$\bip[G\vee H]$, super vertices $V^x_i$ and induced matching $\mset_G$}]{
  \includegraphics[height=0.25\textwidth, clip=true, trim=6cm 6cm 5cm 4.8cm]{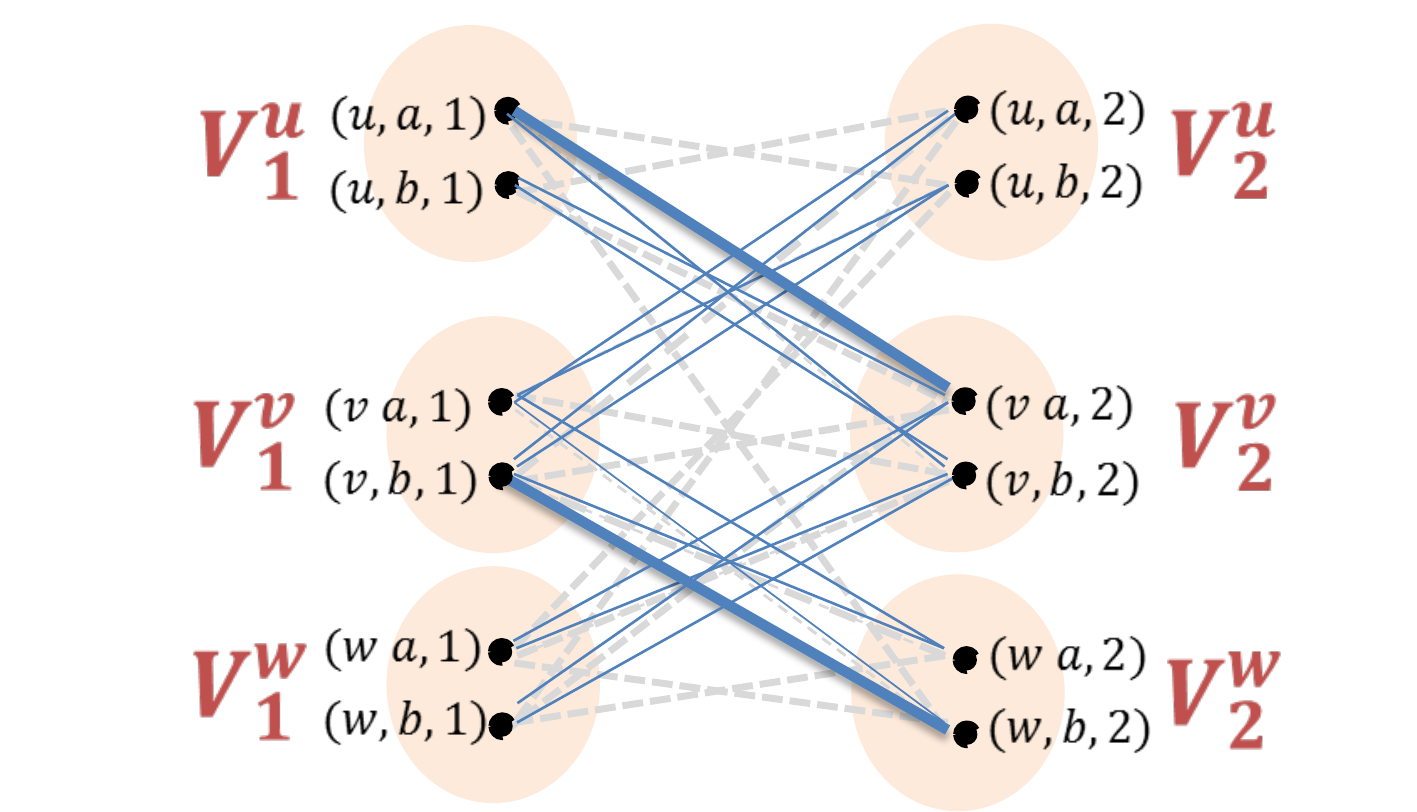}
    \label{fig:induce3}
}
\end{minipage}}}
\caption{Example of graphs $G$, $H$, $\bipm[G]$, $\bipm[H]$ and $\bip[G\vee H]$, as well as super vertices $V_i^x$, set of edges $E_G$ and induced matching $\mset_G$ (defined in Section~\ref{sec:im-ind}). Bold edges are in $\mset_G$. Solid edges (in blue, including bold edges) are edges assigned to  $E_G$, and dashed edges (in gray) are edges assigned to $E_H$. Observe that if we view $V_i^x$ as a vertex (by unifying vertices in them) and consider only edges in $E_G$, then the graph looks exactly like $\bip[G]$. Moreover, the induced matching $\mset_G$ becomes an induced matching $\{V^u_1V^v_2, V^v_1V^w_2\}$ in this graph of super vertices. This is the main fact we use to prove Eq.\eqref{eq:subadditive induce special}.}\label{fig:induce_13}
\end{figure*}

\subsection{Induced Matching Number (Eq.\eqref{eq:subadditive induce}).}
\label{sec:im-ind}
%
%\danupon{We have think again whether we really need $\bip[G]$. To me, we rarely use it at least in this section. It might become useful in other sections though.}
%
In this section, we aim to prove the following special case of Eq.\eqref{eq:subadditive induce}:
\begin{align}
\induce{\bipm[G\vee H]} \leq \induce{\bipm[G]}+\induce{\bipm[H]}.\label{eq:subadditive induce special}
\end{align}
Let $V_1$ and $V_2$ be the two partitions of vertices in $\bip[G\vee H]$ and $\mset$ be an induced matching in $\bip[G \vee H]$.
Recall that each edge in $\bip[G\vee H]$ is of the form
$(u,a,1)(v,b,2)$, where
$u, v\in V(G)$ and $a, b\in V(H)$, and it appears in
$\bip[G \vee H]$ if and only if at least one of the following conditions holds: (1) $uv\in{E(G)}$ or (2) $ab\in{E(H)}$. Our strategy is to consider edges satisfying each condition separately.

In particular, we let $E(\bip[G\vee H]) = E_G\cup E_H$, where $E_G$ and $E_H$ consist of edges $(u,a,1)(v,b,2)$ that satisfy the first and second condition, respectively. That is,
% $uv \in E(G)$
$E_G=\set{(u,a,1)(v,b,2): uv \in E(G)}$ and $E_H=\set{(u,a,1)(v,b,2): ab \in E(H)}$. For example, in Fig.~\ref{fig:induce3}, $E_G$ consists of solid edges (in blue) and $E_H$ consists of dashed edges (in gray). Note that some edges, e.g., edge $(u,a,1)(v, b, 2)$ in Fig.~\ref{fig:induce3}, are in both $E_G$ and $E_H$. We also partition the induced matching $\mset$ into $\mset=\mset_G\cup \mset_H$ where $\mset_G=\mset\cap E_G$ and $\mset_H=\mset\cap E_H$.
Obviously, $|\mset|\leq |\mset_G|+|\mset_H|$.
Our goal is to show that $|\mset_G| \leq \induce{\bip[G]}$ and
$|\mset_H|\leq\induce{\bip[H]}$.
We will only show the former claim because the latter can be argued similarly.

To prove this claim, we partition vertices in $V_1$ and $V_2$ according to which vertices in $G$ they ``inherit'' from.
That is, for any vertex $u\in V(G)$, we let
$V_1^u = \set{(u,a,1): a \in V(H)}$ and $V_2^u = \set{(u,a,2): a \in V(H)}$ (e.g., see Fig. \ref{fig:induce3}).

We can think of each set $V_i^u$ as a ``super vertex'' corresponding to a vertex $(u, i)$ in $\bipm[G]$ in the sense that if we unify all vertices in $V_i^u$ into one vertex, for all $u\in V(G)$ and $i\in V(K_2)$, and remove duplicate edges, then we will get the graph $\bipm[G]$.
In fact, we can show more than this. We can show that if we look at $\mset_G$ in the graph of super vertices, then we will get an induced matching of $\bipm[G]$ having the same size as $\mset_G$! For example, in Fig.~\ref{fig:induce3} the induced matching $\mset_G$ in $\bip[G\vee H]$ consisting of bold edges becomes a set of two edges $\{V^u_1V^v_2, V^v_1V^w_2\}$ in the graph of super vertices, which is still an induced matching.

The key idea in proving this fact is an observation that for any pair of super vertices $V_1^u$ and $V_2^v$, either there is no edge between any pair of vertices in $V_1^u$ and $V_2^v$, or there will be edges between all pairs of vertices in $V_1^u$ and $V_2^v$. For example, in Fig.~\ref{fig:induce3}, there is no edge between any pair of vertices $x\in V^u_1$ and $y\in V^w_2$ while there is an edge between every pair of vertices $x\in V^u_1$ and $y\in V^v_2$.
Using this observation, we can easily prove the two lemmas below. The first lemma says that $\mset_G$ becomes a matching in the graph of super vertices, and the second one says that this matching is, in fact, an induced matching.

\begin{figure*}[t]
\setlength{\fboxsep}{0pt}
\centering{\fbox{\begin{minipage}{\textwidth}
\centering
\subfigure[{$G$, $H$, $\bip[G]$ and $\bip[H]$}]{
% \centering{\fbox{\begin{minipage}{\columnwidth}
  \includegraphics[height=0.25\textwidth, clip=true, trim=3cm 8.6cm 6cm 2.3cm]{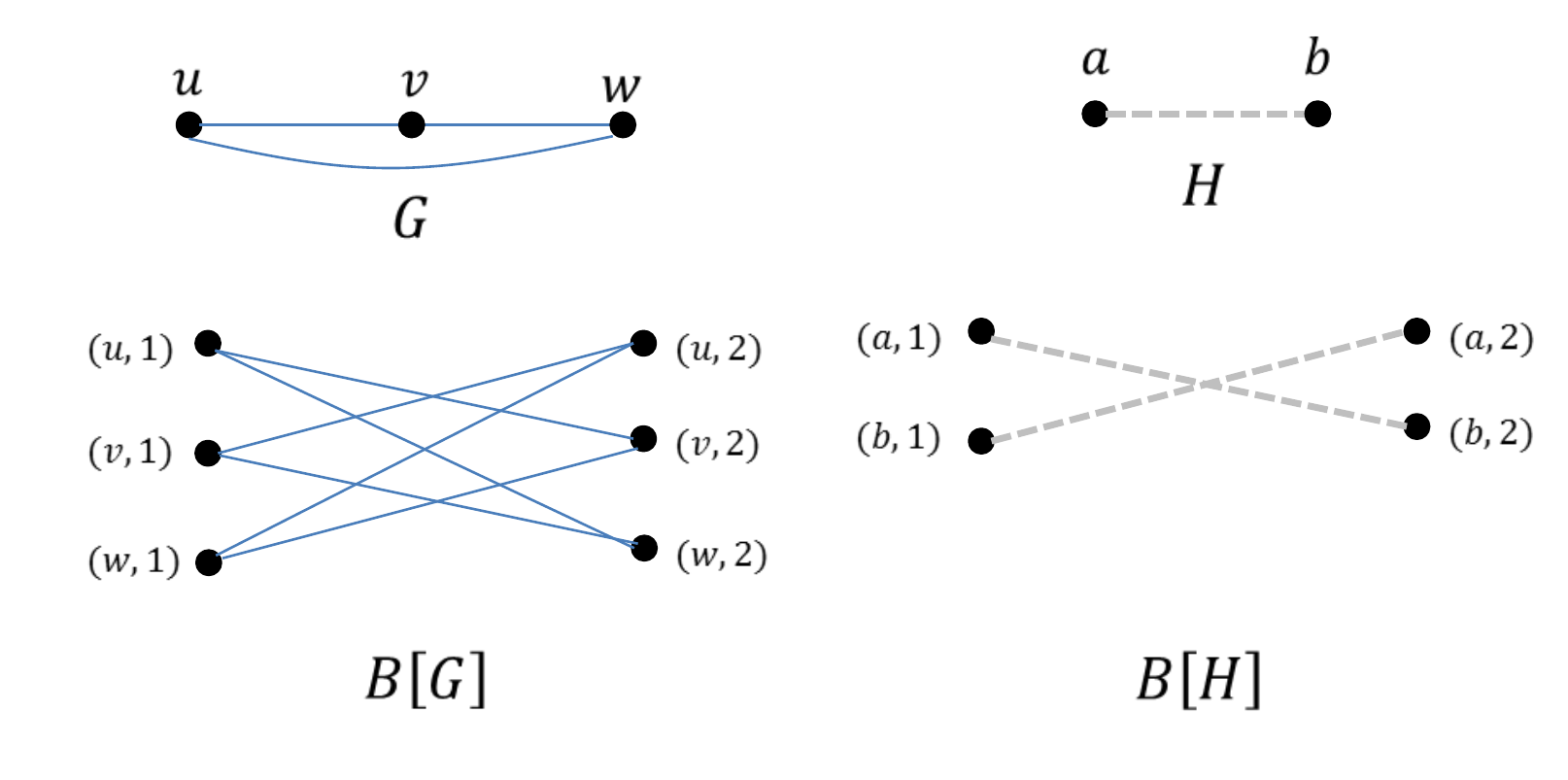}
    \label{fig:induce1_2}
% \end{minipage}}}
}
\subfigure[{$\bip[G\vee H]$}]{
% \centering{\fbox{\begin{minipage}{\columnwidth}
  \includegraphics[height=0.25\textwidth, clip=true, trim=6cm 5.5cm 5cm 4.2cm]{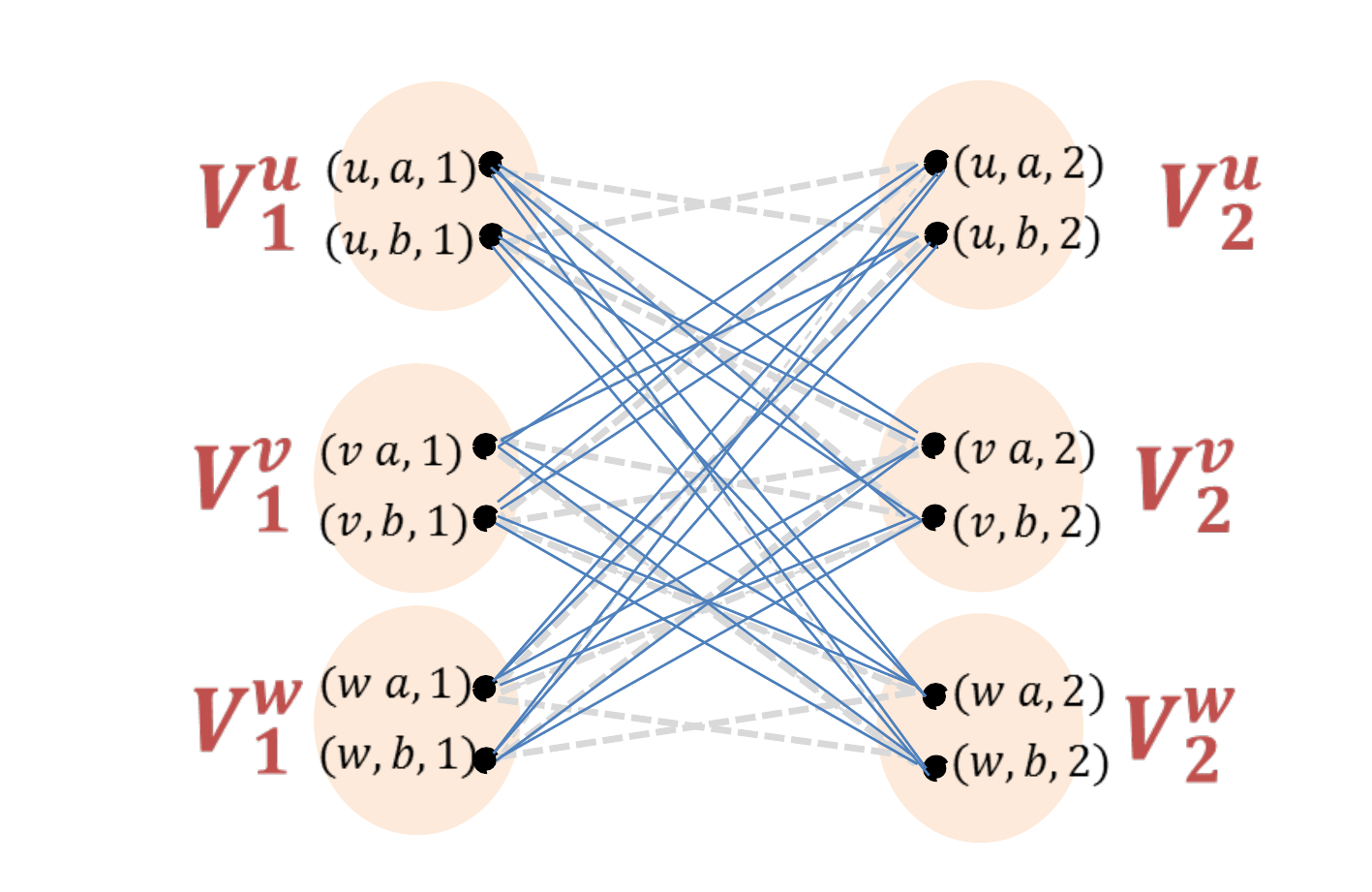}
    \label{fig:induce3_2}
% \end{minipage}}}
}
\end{minipage}}}
\caption{Example of graphs $G$, $H$, $\bipm[G]$, $\bipm[H]$ and $\bip[G\vee H]$, as well as super vertices $V_i^x$. Solid edges (in blue) are edges assigned to  $E_G$, and dashed edges (in gray) are edges assigned to $E_H$. ($E_G$ and $E_H$ are defined in Section~\ref{sec:im-ind}.)}
\end{figure*}

Before proceeding to the proofs, recall that we write the edge set of  $\bip[G\vee H]$ as  $E(\bip[G\vee H]) = E_G\cup E_H$, where $E_G=\set{(u,a,1)(v,b,2):uv \in E(G)}$ and  $E_H=\set{(u,a,1)(v,b,2): ab \in E(H)}$.

\begin{lemma}\label{lem:induce special 1}
For any $u\in V(G)$ and $i\in \{1, 2\}$, $V_i^u$ contains an endpoint of at most one edge in $\mset_G$.
\end{lemma}
\begin{proof}
For the sake of contradiction, assume that there is a vertex $u\in
V(G)$ such that  $V_1^u$ contains two endpoints of two edges in
$\mset_G$, say $(u,a,1)(v,b,2)$ and $(u,a',1)(v',b',2)$.
(The case of $V_2^u$ is proved analogously.)
Since $(u,a,1)(v,b,2)$ is in $E_G$ (recall that $\mset_G=\mset\cap E_G$), we have that $uv\in E(G)$ and thus
$(u,1)(v,2)$ is in $E(\bipm[G])$.
This fact then implies that there is an edge between $(u,a',1)$ and
$(v,b,2)$ in $E_G$ as well, contradicting to the fact that $\mset_G$ (and
thus $\mset$) is an induced matching.\qedhere
\end{proof}

\paragraph{Example}
Here we illustrate the proof of Lemma~\ref{lem:induce special 1}.
Consider Fig.~\ref{fig:induce3_2} and let us say that $\mset_G$
contains edges $(u, a, 1)(v, b, 2)$ and $(u, b, 1)(v, a, 2)$ which
means that $V_1^u$ contains endpoints of two edges in
$\mset_G$. Having the first edge in $E_G$ means that $uv\in E(G)$ and
thus $(u, 1)(v, 2)$ is in $E(\bipm[G])$ (as witnessed in
Fig.~\ref{fig:induce1_2}). But then it means
that the edge $(u, a, 1)(v, a, 2)$ must be in $E_G$ as well, making
$\mset_G$ (and thus $\mset$) not an induced matching.

\begin{lemma}\label{lem:induce special 2}
For any $u, u', v, v'\in V(G)$, if $\mset_G$ contains an edge between a pair of vertices in $V_1^u$ and $V_2^v$ and an edge between another pair of vertices in $V_1^{u'}$ and $V_2^{v'}$, then there must be no edge between vertices in $V_1^u$ and $V_2^{v'}$ in $E_G$.
\end{lemma}

\begin{proof}
Assume for a contradiction that $\mset_G$ contains edges $(u, a, 1)(v, b, 2)$ and $(u', a', 1)(v', b', 2)$ and there is an edge, say $(u, c, 1)(v', d, 2)$ in $E(G)$.
% Having the last edge in $E_G$ implies that $uv'\in E(G)$ which means that $(u, 1)(v', 2)\in E(\bipm[G])$ which then implies that $(u, a, 1)(v', b', 2)$ is in $E_G$.
Since the edge $(u, c, 1)(v', d, 2)$\bundit{the word ``last edge'' is not clear, so I re-phrase. The old one is ``Having the last edge in $E_G$ implies that $uv'\in E(G)$ which means that $(u, 1)(v', 2)\in E(\bipm[G])$ which then implies that $(u, a, 1)(v', b', 2)$ is in $E_G$.''}
 is in $E_G$, we have $uv'\in{E(G)}$ and thus $(u, 1)(v', 2)\in E(\bipm[G])$.
This implies that $(u, a, 1)(v', b', 2)$ is in $E_G$,
which contradicts the fact that $\mset$ is an induced matching in $\bipm[G\vee H]$.\qedhere
\end{proof}

\paragraph{Example}
Here we illustrate the proof of Lemma~\ref{lem:induce special 2}.
Consider Fig.~\ref{fig:induce3_2}, and let us say
that the matching $\mset_G$ contains $(v, a, 1)(u, a, 2)$ and $(u, a,
1)(w, a, 2)$ and there is an edge $(v, b, 1)(w, b, 2)$ which prevents
$\mset_G$ from being an induced matching in the graph of super
vertices. Having the last edge in $E_G$ implies that $(v, 1)(w, 2)\in
E(\bipm[G])$ which in turns implies that $(v, a, 1)(w, a, 2)\in E_G$,
making $\mset_G$ (and thus $\mset$) not an induced matching in
$\bipm[G\vee H]$.

\paragraph{Note on Proving the General Version (Section~\ref{sec: subadditivity full})} To prove the general version, i.e., Eq.\eqref{eq:subadditive induce}, we may define $E_G$ and $E_H$ analogously to the proof in this section. We may then define super vertices in a similar way and prove the lemmas that are similar in spirit to Lemma \ref{lem:induce special 1} and \ref{lem:induce special 2}. However, we choose an {\em alternative way} which seems more suitable in proving the general version by decomposing the graph into products of some well-structured graphs and show the {\em associativity} property of graph products we use.

\begin{figure*}[t]
\setlength{\fboxsep}{0pt}
\centering{\fbox{\begin{minipage}{\textwidth}
\centering
\subfigure[{$G$, $H$, $\bipp[G]$ and $\bipp[H]$}]{
  \includegraphics[height=0.28\textwidth, clip=true, trim=2cm 6cm 4.2cm 1.5cm]{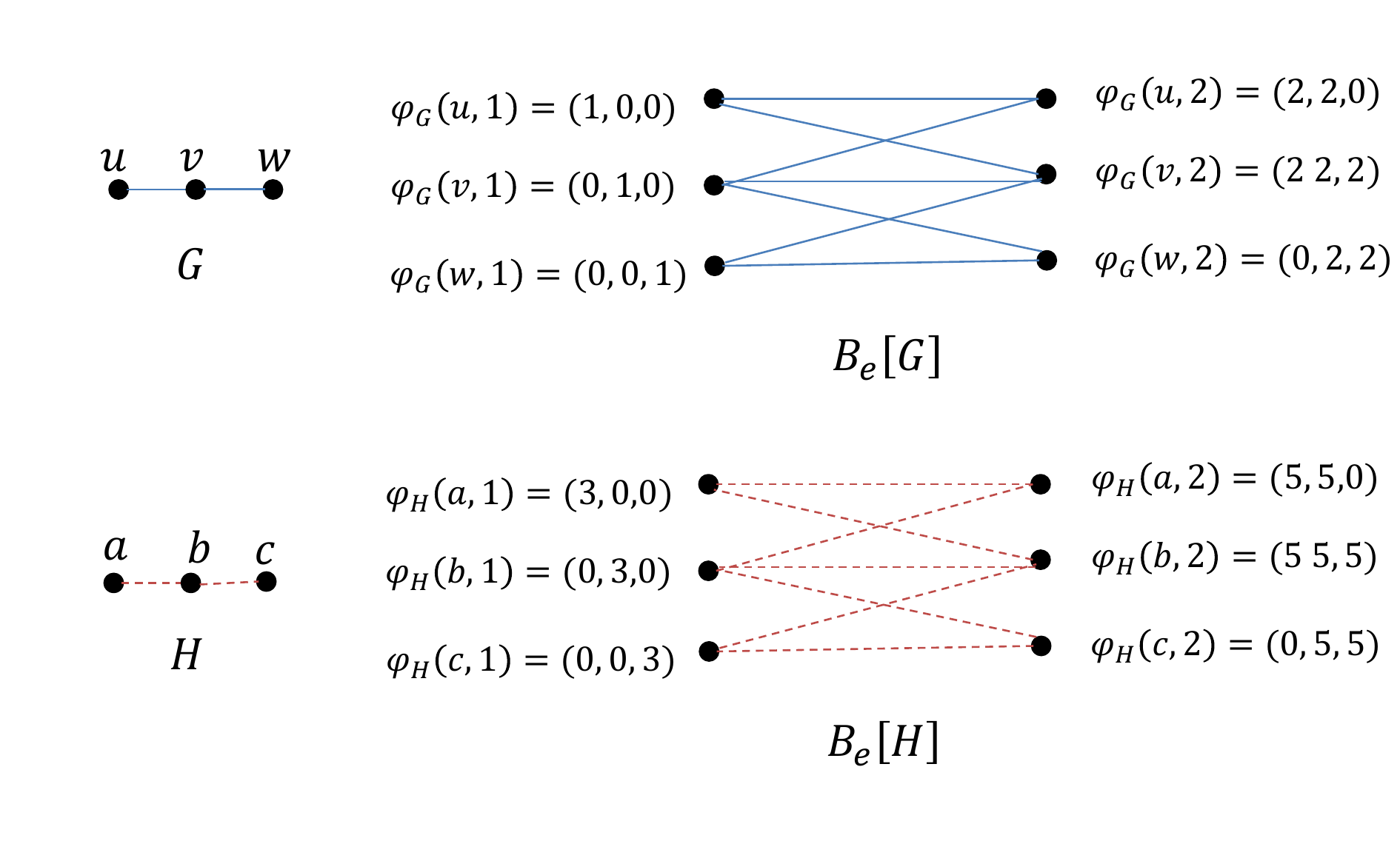}
    \label{fig:dimension1}
}
\subfigure[{$\bipp[G\cdot H]$}]{
  \includegraphics[height=0.28\textwidth, clip=true, trim=6cm 6cm 5cm 1.3cm]{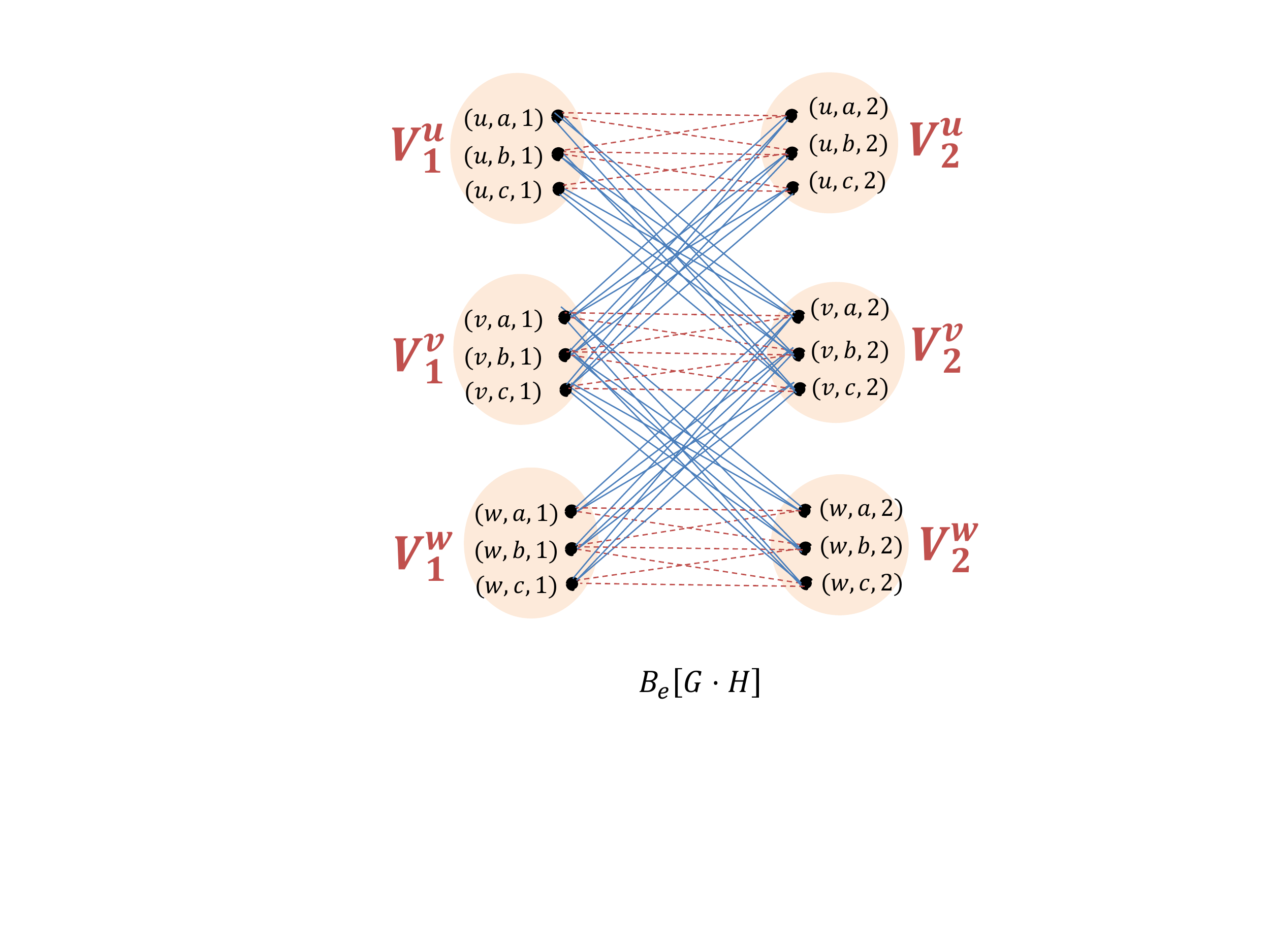}
    \label{fig:dimension3}
}
\end{minipage}}}
\caption{Graphs $G$, $H$, $\bipp[G]$, $\bipp[H]$ and $\bipp[G\cdot H]$, as well as super vertices $V_i^x$ and mappings $\varphi_G$ and $\varphi_H$ that realize $\bipp[G]$ and $\bipp[H]$, respectively. Note that directions of edges are omitted in the pictures. They are always from left to right.}\label{fig:dimension}
\end{figure*}

\subsection{Poset Dimension Number (Eq.\eqref{eq:subadditive dimension}).}
\label{sec:poset-ineq}
We now prove the special case of Eq.~\eqref{eq:subadditive dimension}\footnote{Note that Eq.\eqref{eq:subadditive dimension} implies that $\dimension{\bipp[G\cdot H]}\leq \dimension{\bipp[G]}+\chi(G)\dimension{\bipp[H]} +1$, so we are proving a slightly stronger statement for this special case.}:
%{

\begin{align}
\label{eq:subadditive dimension special} \dimension{\bipp[G\cdot H]} &\leq \dimension{\bipp[G]} +\chi(G)\dimension{\bipp[H]}\nonumber \,.
\end{align}

%}
%
Throughout this section, we will think of $\bipp[G]$, for any undirected graph $G$, as a poset $G\hstrong \vec{K}_2$. Thus, edges in $\bipp[G \cdot H]$ are directed edges in the form $(u, a, 1)(v, b, 2)$ for some $u, v\in V(G)$ and $a, b\in V(H)$.
%
%Let $V_1$ and $V_2$ be the two partitions of vertices in $\bipp[G\cdot H]$. Recall that each edge in $\bipp[G\cdot H]$  is of the form $(u,a,1)(v,b,2)$, where $uv\in E(G)$ \underline{or} ($u=v$ and $ab\in E(H))$).
%
Let $d_G = \dimension{\bipp[G]}$, $d_H = \dimension{\bipp[H]}$, and $\phi_G: V(\bipp[G]) \rightarrow \R^{d_G}, \phi_H: V(\bipp[H]) \rightarrow \R^{d_H}$ be mappings that realize the posets $\bipp[G]$ and $\bipp[H]$, respectively. This means that, for example, $(u, 1)(v, 2)\in E(\bipp[G])$ if and only if $\varphi_G(u, 1)<\varphi_G(v, 2)$.
%
%As in Section~\ref{sec: simplified subadditivity}, we will also partition vertices in $\bipp[G\cdot H]$ into ``super vertices'': For any $u\in V(G)$ and $i\in \{1, 2\}$, we let $V^u_i=\{(u, a, i) : a\in V(H)\}$. We view each set $V^u_i$ as a super vertex.\danupon{We may not need ``super vertex''}
%
We may assume without loss of generality that all coordinates of $\phi_G$ and $\phi_H$ are non-negative (by adding appropriate positive numbers).
See an example in Fig.~\ref{fig:dimension}.

Our strategy is to use $\varphi_G$ and $\varphi_H$ to define a mapping $\varphi:V(\bipp[G\cdot H])\rightarrow \reals^{d_G+\chi(G)d_H}$ that realizes $\bipp[G\cdot H]$. Again, this means that we want $\varphi$ such that for any vertices $(u, a, i)$ and $(v, b, j)$,
$(u, a, i)(v, b, j)\in E(\bipp[G\cdot H])$ if and only if $\varphi(u, a, i)<\varphi(v, b, j)\,.$
To simplify our discussion, we will focus on the case where $i=1$ and $j=2$ (The cases when $i=j$ are easy to deal with).

%%%%%%%%%%%%%%%%%%%%%%%%%%%%%%%%%%%%%%%%%%%%%%%%%%%%%%%%%%%%%%%%%

\begin{figure*}[t]
\setlength{\fboxsep}{0pt}
\centering{\fbox{\begin{minipage}{\textwidth}
\centering
\subfigure[]{
  \includegraphics[height=0.22\textheight, clip=true, trim=5cm 6cm 3cm 1cm]{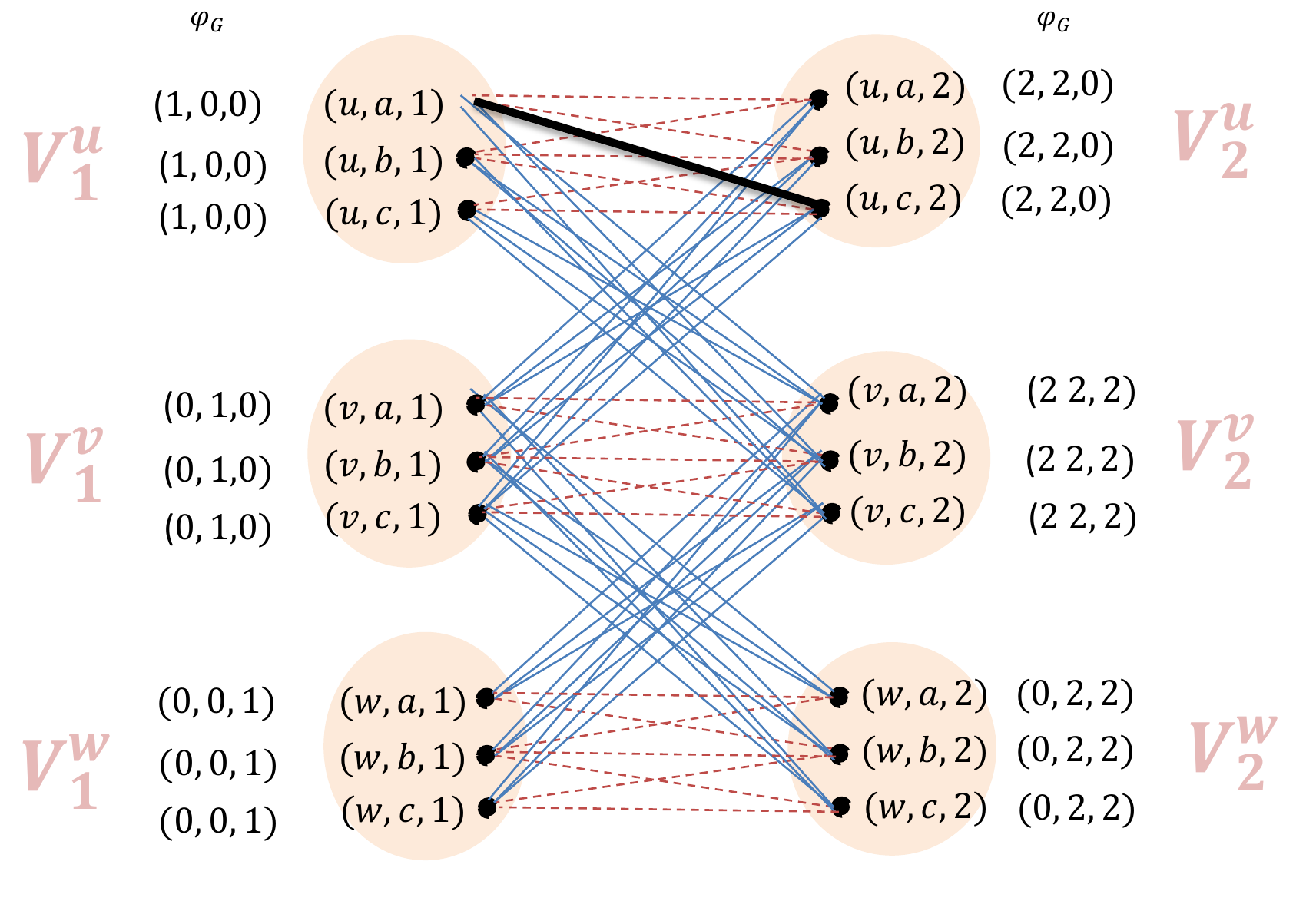}
    \label{fig:dimension4}
}
\subfigure[]{
  \includegraphics[height=0.22\textheight, clip=true, trim=3cm 6cm 2cm 1cm]{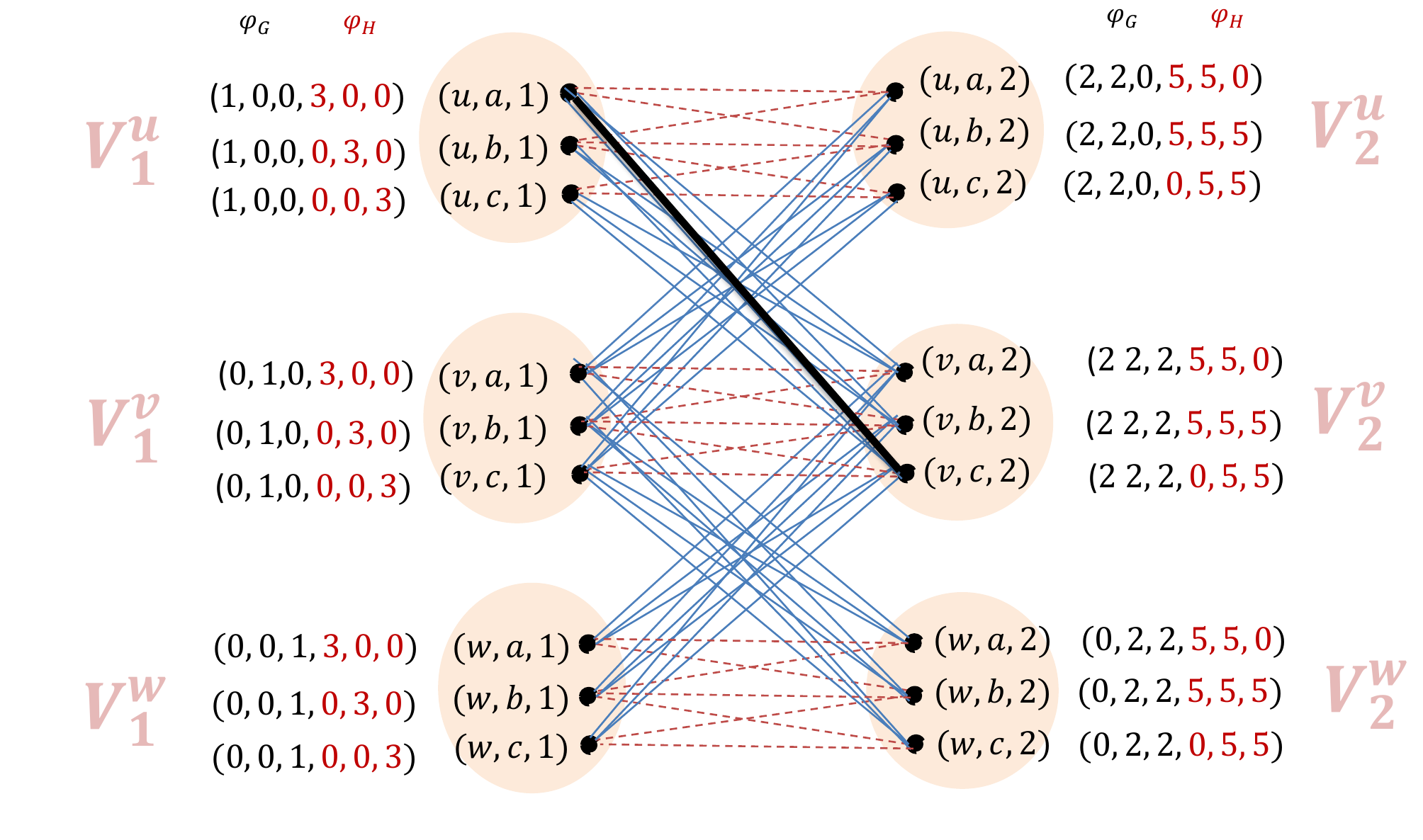}
    \label{fig:dimension5}
}
\end{minipage}}}
\caption{\subref{fig:dimension4} $\varphi_1$ (cf. Section~\ref{sec: simplified subadditivity}) introduces ``undesirable edges'' such as the bold edge in this picture.
\subref{fig:dimension5} $\varphi_2$ (cf. Section~\ref{sec: simplified subadditivity}) removes ``desirable edges'' such as the bold edge in this picture.}\label{fig:dimension_45}
\end{figure*}

\paragraph{Proof Idea} Before we show the construction of $\varphi$, let us show a few {\em failed attempts} to illustrate the intuition behind the construction (readers may feel free to skip this part to the definition of $\varphi$ below). We use Fig.~\ref{fig:dimension} as an example. Fig.~\ref{fig:dimension_45} and \ref{fig:dimension_67} might be also helpful as visual aids.

The first attempt is to use $\varphi_1(u, a, i)=\varphi_G(u, i)$ to realize $\bipp[G\cdot H]$. This obviously fails, simply because we did not use $\varphi_H$ at all: In Fig.~\ref{fig:dimension} (also see Fig.~\ref{fig:dimension4}), we have $\varphi_1(u, a, 1)<\varphi_1(u, c, 2)$, but $(u, a, 1)(u, c, 2)\notin E(\bipp[G\cdot H])$.  In other words, $\varphi_1$ ``introduces'' some ``undesirable edges'' -- edges $(u', a', 1)(v', b', 2)$ that are {\em not} in $\bipp[G\cdot H]$ but $\varphi_1(u', a', 1)<\varphi_1(v', b', 2)$.

%Fig.~\ref{fig:dimension4} in Appendix~\ref{sec:missing proof subadditivity} provides a full visualization of this.

%Intuitively, $\varphi_1$ fails because it does not take $\varphi_H$ into account. It is thus natural to define $\varphi_2(u, a, i)=\varphi_G(u, i)\varphi_H(a, i)$ which is a ``concatenation'' of $\varphi_G(u, i)$ and $\varphi_H(a, i)$ (thus, the dimension of $\varphi_2$ is $d_G+d_H$). It can be shown that there is no undesirable edge introduced by $\varphi_2$. However, $\varphi_2$ might ``remove'' some ``desirable edges'' -- edges  $(u', a', 1)(v', b', 2)$ that {\em is} in $\bipp[G\cdot H]$ but $\varphi_1(u', a', 1)\not <\varphi_1(v', b', 2)$. For example, in Fig.~\ref{fig:dimension} (also see Fig.~\ref{fig:dimension5}), $\varphi_2(u, a, 1)=(1, 0, 0, 3, 0, 0)$ and $\varphi_2(v, c, 2)=(2, 2, 2, 0, 5, 5)$ which means that $\varphi_2(u, a, 1)\not < \varphi_2(v, c, 2)$; however, $(u, a, 1)(v, c, 2)\in E(\bipp[G\cdot H])$.

A natural way to fix this is to define $\varphi_2(u, a, i)=\varphi_G(u, i)\varphi_H(a, i)$ which is a ``concatenation'' of $\varphi_G(u, i)$ and $\varphi_H(a, i)$ (thus, the dimension of $\varphi_2$ is $d_G+d_H$). It can be shown that there is no undesirable edge introduced by $\varphi_2$. However, $\varphi_2$ might ``remove'' some ``desirable edges'' -- edges  $(u', a', 1)(v', b', 2)$ that {\em are} in $\bipp[G\cdot H]$ but $\varphi_1(u', a', 1)\not <\varphi_1(v', b', 2)$.
For example, in Fig.~\ref{fig:dimension} (also see Fig.~\ref{fig:dimension5}), $\varphi_2(u, a, 1)\not < \varphi_2(v, c, 2)$, but $(u, a, 1)(v, c, 2)\in E(\bipp[G\cdot H])$.
%
%For example, in Fig.~\ref{fig:dimension} (also see Fig.~\ref{fig:dimension5}), $\varphi_2(u, a, 1)=(1, 0, 0, 3, 0, 0)$ and $\varphi_2(v, c, 2)=(2, 2, 2, 0, 5, 5)$ which means that $\varphi_2(u, a, 1)\not < \varphi_2(v, c, 2)$; however, $(u, a, 1)(v, c, 2)\in E(\bipp[G\cdot H])$.

%Fig.~\ref{fig:dimension5} in Appendix~\ref{sec:missing proof subadditivity}  shows a visualization of this.

We thus need a more clever way to combine $\varphi_G$ with $\varphi_H$. A crucial observation we found is that {\em if we concatenate them only at vertices that are independent in $G$, then we will not remove any desirable edges}. For example, in Fig.~\ref{fig:dimension1}, vertices $u$ and $w$ are independent in $G$. So, for any $x\in \{u, w\}$ and $r\in V(H)$, we will let $\varphi_3$ be as in Fig.~\ref{fig:phi three}.
\begin{figure*}
\setlength{\fboxsep}{0pt}
\centering{
\fbox{
\begin{minipage}{\columnwidth}
\noindent
{
\begin{align*}
\varphi_3(x, r, 1) &=\mbox{\cfbox{gray}{$\varphi_G(x, 1)$}{1.7cm}}\mbox{ \cfbox{blue}{$\varphi_H(r, 1)$}{1.7cm}}\mbox{ \cfbox{green}{$(0,0,0)$}{1.7cm}} &\mbox{and} &&
\varphi_3(x, r, 2) &=\mbox{\cfbox{gray}{$\varphi_G(x, 2)$}{1.7cm}}\mbox{ \cfbox{blue}{$\varphi_H(r, 2)$}{1.9cm}}\mbox{ \cfbox{green}{$(\infty,\infty,\infty)$}{1.9cm}} \\
\varphi_3(v, r, 1) &=\mbox{\cfbox{gray}{$\varphi_G(v, 1)$}{1.7cm}}\mbox{ \cfbox{blue}{$(0,0,0)$}{1.7cm}} \mbox{ \cfbox{green}{$\varphi_H(r, 1)$}{1.7cm}} &\mbox{and} &&
\varphi_3(v, r, 1) &=\mbox{\cfbox{gray}{$\varphi_G(v, 2)$}{1.7cm}}\mbox{ \cfbox{blue}{$(\infty,\infty,\infty)$}{1.9cm}}\mbox{ \cfbox{green}{$\varphi_H(r, 2)$}{1.9cm}}
\end{align*}
}
\end{minipage}
}}
\caption{Example of $\varphi_3$. Note that $x$ is node in $\{u, w\}$.}\label{fig:phi three}
\end{figure*}
In other words, every vertex starts with $\varphi_G$ (i.e., in the first (gray) boxes). Then, we ``attach'' $\varphi_H$ to $\varphi_G$ at vertices of the form $(x,r,i)$, where $x \in \set{u,w}$, while attaching ``trivial'' vectors ($(0, 0, 0)$ or $(\infty, \infty, \infty)$) at $(v,r,i)$ (i.e., in the second (blue) boxes in Fig~\ref{fig:phi three}). We then do the opposite. We attach $\varphi_H$ to vertices $(v,r,i)$ while attaching trivial vectors to other vertices (i.e., in the last (green) boxes). It can be checked (e.g., Fig.~\ref{fig:dimension7}) that $\varphi_3$ does realize $\bipp[G\cdot H]$.

\begin{figure*}
\setlength{\fboxsep}{0pt}
\centering{\fbox{\begin{minipage}{\textwidth}
\begin{center}
 \includegraphics[height=0.25\textheight, clip=true, trim=2cm 6cm 0cm 1cm]{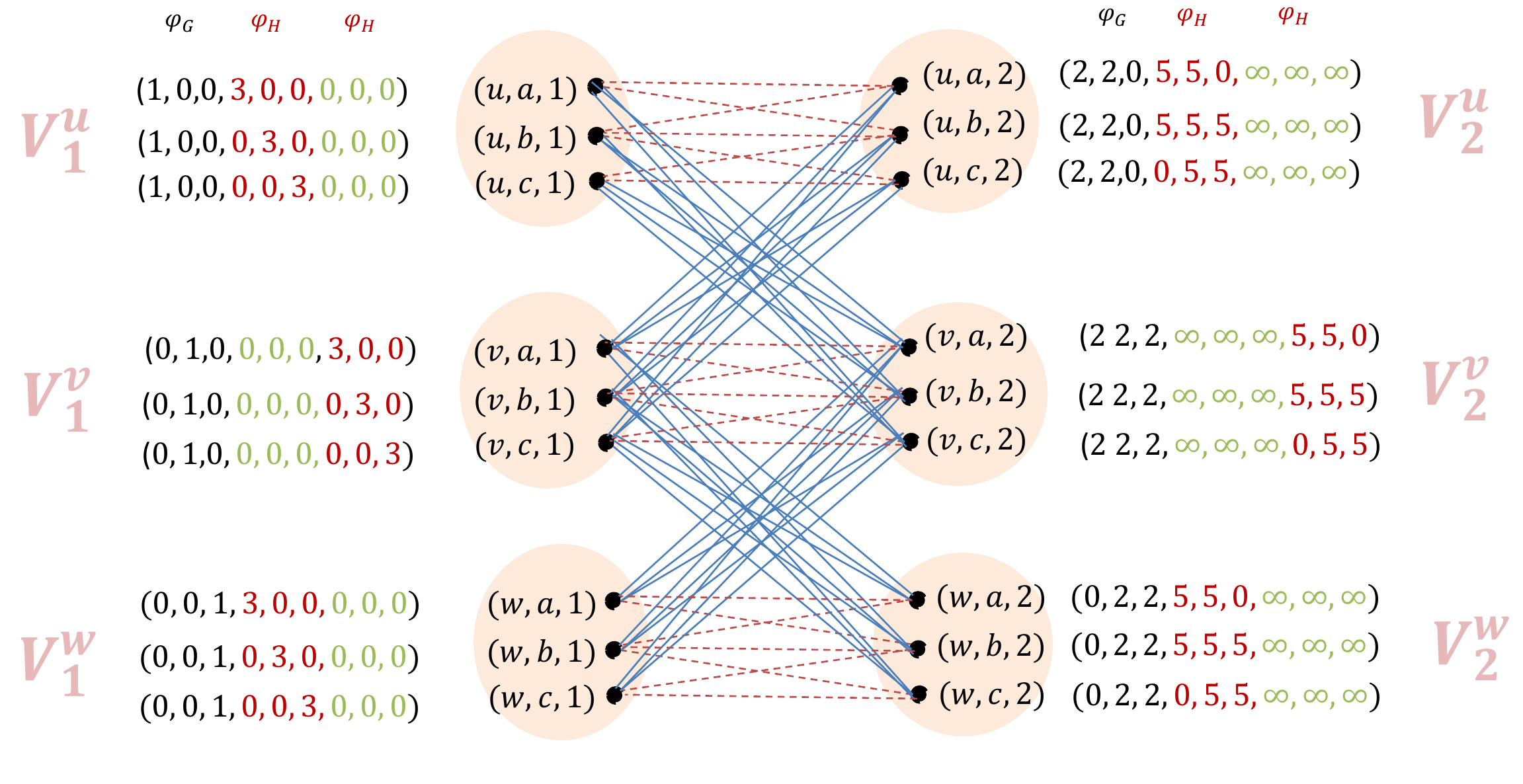}
\end{center}
\end{minipage}}}
\caption{$\varphi_3$ (cf. Section~\ref{sec: simplified subadditivity}) which realizes {$\bipp[G\cdot H]$}.}\label{fig:dimension_67}\label{fig:dimension7}
\end{figure*}

To summarize, the general idea of constructing $\varphi$ is to keep attaching $\varphi_H$ to $\varphi_G$ where each attachment must be done only on vertices that are {\em independent in $G$}. The dimension of $\varphi$ depends on how many times we attach $\varphi_H$. A natural way to minimize the number of attachments is to use $\chi(G)$ color classes since each color class contains independent vertices. This is why the dimension becomes $d_G+\chi(G)d_H$.

%
%Thus, we can minimize the number of attachments, each time we attach $\varphi_H$ on vertices corresponding to one color class of $G$. This will requires $\chi(G)$ attachments in total; so the dimension of $\varphi$ will be $d_G+\chi(G)d_H$. We now define $\varphi$ formally.
%
%%%%%%%%%%%%%%%%%%%%%%%%%%%%%%%%%%%%%%%%%%%%%%%%%%%%%%%%%%%%%%%%%%%%
%
%\begin{proof}
%Let $d_1 = \dimension{\bipp[G]}$, $d_2 = \dimension{\bipp[H]}$, and $\phi_1: \bipp[G] \rightarrow \R_{+}^{d_1}, \phi_2: \bipp[H] \rightarrow
%\R_{+}^{d_2}$ be mappings that realize the posets $\bipp[G]$ and $\bipp[H]$,
%respectively.

\paragraph{Constructing $\varphi$} %We construct a new mapping $\phi: V(\bipp[G \cdot H]) \rightarrow \R^{d_G + k d_H}$ that realizes the poset $\bipp[G \cdot H]$ as follows.
Let $\mathcal{C}: V(G) \rightarrow [k]$ be an optimal coloring of $G$, where $k=\chi(G)$.
The coordinates in $\R^{d_G+ k d_H}$ are viewed as $k+1$ {\em blocks}.
In the first block $\bset_0$, we have $d_G=\dimension{\bipp[G]}$ coordinates, and in the $k$
consecutive blocks $\bset_1, \ldots, \bset_{k}$, we have $d_H=\dimension{\bipp[H]}$
coordinates per block. We will define the coordinates of each vertex in $\bipp[G \cdot H]$ by describing the coordinates in each block. For each point ${\bf x} \in \R^{d_G + k d_H}$ and each block $\bset_j$, we refer to coordinates in block $\bset_j$ of ${\bf x}$ as ${\bf x}|_{\bset_j}\,.$

% For each poset element of the form $(i,a)$, we define its coordinates
% as in block $\bset_0$ as $\phi_1(i)$ and
% coordinates in block $\bset_{\mathcal{C}(i)}$ as $\phi_2(a)$.
% The rest of the coordinates are assigned to be zero.
For each vertex in $\bipp[G\cdot H]$ of the form $(u,a,1)$, we define its coordinates
as\danupon{For later: It might be better to change $u$ to $x$ and $a$ to $r$ but there are too many of them (also in appendix)!}
\begin{align*}
&\phi((u,a,1))|_{\bset_0} &&= \phi_G((u,1)),\\
&\phi((u,a,1))|_{\bset_{\mathcal{C}(u)}} &&= \phi_H((a,1)), &&\mbox{and}~~\\
&\phi((u,a,1))|_{\bset_\ell} &&= (0,\ldots,0), &&\mbox{otherwise.}
\end{align*}

% For each poset element of the form $(i,a)'$, the coordinates in the
% first block $\bset_0$ are exactly $\phi_1(i)$, and the coordinates in
% block $\bset_{\mathcal{C}(i)}$ are $\phi_2(a)$.
% Other coordinates are $\infty$.
For each vertex in $\bipp[G\cdot H]$ of the form $(u,a,2)$, we define its coordinates
as
\begin{align*}
&\phi((u,a,2))|_{\bset_0} &&= \phi_G((u,2)),\\
&\phi((u,a,2))|_{\bset_{\mathcal{C}(u)}} &&= \phi_H((a,2)), &&\mbox{and}~~\\
&\phi((u,a,2))|_{\bset_\ell} &&= (\infty, \ldots, \infty), &&\mbox{otherwise.}
\end{align*}
We finish the proof of Eq.\eqref{eq:subadditive dimension special} by the following lemma, %which can be easily proved by case analysis.
which can be proved by case analysis. % Bun: The word "easily" should be removed.

\begin{lemma}\label{lem:phi realizes B}
$\varphi$ realizes $\bipp[G\cdot H]$.
\end{lemma}
\begin{proof}
We will use the following properties:

\renewcommand{\theenumi}{P\arabic{enumi}}
%\begin{enumerate}
\squishnum
\item \label{property: embedding of G} For any $(u,1), (v,2) \in
  V(\bipp[G])$,
 $(u,1)(v,2) \in E(\bipp[G])$ $\iff$ $\phi_G((u,1)) < \phi_G((v,2))$.

\item \label{property: embedding of H} For any
  $(a,1)(b,2)\in{V(\bipp[H])}$,
  $(a,1)(b,2) \in E(\bipp[H])$ $\iff$ $\phi_H((a,1)) < \phi_H((b,2))$.

\item \label{property: coloring of G} For any $uv \in E(G)$, $\cset(u) \neq \cset(v)$.
%\end{enumerate}
\squishnumend
\renewcommand{\theenumi}{\arabic{enumi}}

We argue that for any vertices $(u,a,1),(v,b,2) \in V(\bipp[G\cdot H])$,
we have $(u,a,1)(v,b,2) \in E(\bipp[G\cdot H])$ if and only if
$\phi((u,a,1))< \phi((v,b,2))$. We analyze five possible cases.

\paragraph{Case 1: $u \neq v$ and $uv \in E(G)$}
    We will show that $\phi((u,a,1)) < \phi((v,b,2))$.
    First, $(u,a,1)(v,b,2) \in E(\bipp[G\cdot H])$ by construction, and
    $\mathcal{C}(u)\neq\mathcal{C}(v)$ by
    Property~\ref{property: coloring of G}.
    Next, consider the blocks $\bset_0$, $\bset_{\mathcal{C}(u)}$ and
    $\bset_{\mathcal{C}(v)}$. We have
%
%{\small
%\begin{align}
%    \phi((u,a,1))|_{\bset_0} &=  \phi_G((u,1))\\& <
%      \phi_G((v,2))=\phi((v,b,2))|_{\bset_0}\nonumber\\
%    \phi((u,a,1))|_{\bset_{\cset(u)}}&= \phi_H((a,1)) \\ &<  \infty
%      = \phi((v,b,2))|_{\bset(\cset(u))}\nonumber\\
%    \phi((u,a,1))|_{\bset_{\cset(v)}}&=0 \\&<
%           \phi_G((v,2))=\phi((v,b,2))|_{\bset_{\cset(v)}}\nonumber
%\end{align}
%}

\begin{align*}
    &\phi((u,a,1))|_{\bset_0} =  \phi_G((u,1)) <
      \phi_G((v,2))=\phi((v,b,2))|_{\bset_0}\nonumber\\
    &\phi((u,a,1))|_{\bset_{\cset(u)}}= \phi_H((a,1)) <  \infty
      = \phi((v,b,2))|_{\bset(\cset(u))}\nonumber\\
    &\phi((u,a,1))|_{\bset_{\cset(v)}}=0 <
           \phi_G((v,2))=\phi((v,b,2))|_{\bset_{\cset(v)}}\nonumber
\end{align*}
The first line is  because of Property~\ref{property: embedding of G} and the fact that $(u,1)(v,2)\in{E(\bipp[G])}$.
%
%\begin{itemize}
%
%    \item $\phi((u,a,1))|_{\bset_0} =  \phi_G((u,1)) <
%      \phi_G((v,2))=\phi((v,b,2))|_{\bset_0}$ because of
%    Property~\ref{property: embedding of G} and
%    $(u,1)(v,2)\in{E(\bipp[G])}$.
%
%    \item $\phi((u,a,1))|_{\bset_{\cset(u)}}= \phi_H((a,1)) <  \infty
%      = \phi((v,b,2))|_{\bset(\cset(u))}$.
%
%    \item $\phi((u,a,1))|_{\bset_{\cset(v)}}=0 <
%           \phi_G((v,2))=\phi((v,b,2))|_{\bset_{\cset(v)}}$.
%
%    \squishend
%
This proves the claim.

\paragraph{Case 2: $u \neq v$ and $uv \not\in E(G)$}
    We will show that $\phi((u,a,1)) \not< \phi((v,b,2))$.
    First, $(u,a,1)(v,b,2) \not\in E(\bipp[G \cdot H])$ by construction.
    Consider the block $\bset_0$.
    Because of Property~\ref{property: embedding of G} and
    $(u,1) (v,2) \not\in E(\bipp[G])$, we have
    $\phi((u,a,1))|_{\bset_0} =
     \phi_G((u,1))\not < \phi_G((v,2)) = \phi((v,b,2))|_{\bset_0}$.
     Thus, $\phi((u,a,1)) \not< \phi((v,b,2))$.

\paragraph{Case 3: $u =v$, $a\neq b$ and $ab \in E(H)$}
   We will show that $\phi((u,a,1)) < \phi((v,b,2))$.
   First, $(u,a,1) (v,b,2) \in E(\bipp[G \cdot H])$ by construction.
   Note that $\cset(u)=\cset(v)$ since $u=v$.
   Consider each block. We have

\begin{align}
\phi((u,a,1))|_{\bset_0} 
                 &= \phi_G((u,1))
                 = \phi_G((v,2))
                 = \phi((v,b,2))|_{\bset_0}\\
\label{eq:Case 3 line 2}\phi((u,a,1))|_{\bset_{\cset(u)}} 
                 &= \phi_H((a,1))
                 < \phi_H((b,2))
                 = \phi((v,b,2))|_{\bset_{\cset(v)}}\\
\label{eq:Case 3 line 3}\phi((u,a,1))|_{\bset_{\ell}} 
                 &= 0 \leq \infty
                 = \phi((v,b,2))|_{\bset_{\ell}}
\end{align}

Eq.\eqref{eq:Case 3 line 2} follows because $(a,1) (b,2) \in E(\bipp[H])$ and  Property~\ref{property: embedding of H}, and Eq.\eqref{eq:Case 3 line 3} follows from the settings of other blocks $\bset_{\ell}$. This proves the claim.

    % In this case, we have $(u,a,1) (v,b,2) \in E(\bipp[G \cdot H])$.
    % In the block $\bset_0$, we have $\phi((u,a,1))|_{\bset_0} = \phi_1((u,1)) = \phi_1((v,2)) =\phi((v,b,2))|_{\bset_0}$, and in
    % the block $\bset_{\mathcal{C}(u)}$ we have
    % $\phi((u,a,1))|_{\bset_{\cset(u)}} = \phi_2((a,1))< \phi_2((b,2)) = \phi((v,b,2))|_{\bset_{\cset(v)}}$ due to the fact that $(a,1) (b,2) \in E(\bipp[H])$ and Property~\ref{property: embedding of H}.
    % For other blocks $\bset_l$, we only compare $0$ to $\infty$, that is, $\phi((u,a,1))|_{\bset_l} = 0 \leq \infty = \phi((v,b,2))|_{\bset_l}$.

\paragraph{Case 4: $u=v$, $a\neq b$ and $ab \not\in E(H)$}
    We will show that $\phi((u,a,1)) \not< \phi((v,b,2))$.
    First, $(u,a,1) (v,b,2) \not\in E(\bipp[G\cdot H])$ by construction.
    Consider the block $\bset_{\mathcal{C}(u)}$.
    By Property~\ref{property: embedding of H},
    $\phi((u,a,1))|_{\bset_{\cset(u)}} = \phi_H((a,1)) \not<
    \phi_H((b,2)) = \phi((v,b,2))|_{\bset_{\cset(v)}}$ , thus proving the claim.

    % Then we have $(u,a,1) (v,b,2) \not\in E(\bipp[G\cdot H])$.
    % In the block $\bset_{\mathcal{C}(u)}$, we have
    % $\phi((u,a,1))|_{\bset_{\cset(u)}} = \phi_2((a,1)) \not< \phi_2((b,2)) = \phi((v,b,2))|_{\bset_{\cset(v)}}$, from Property~\ref{property: embedding of H}, implying that $\phi((u,a,1)) \not< \phi((v,b,2))$.

\paragraph{Case 5: $u=v$ and $a=b$}
  We will show that $\phi((u,a,1))< \phi((v,b,2))$.
  First, by the definition of the operator $\bipp$,
  $(u,a,1)(v,b,2) \in E(\bipp[G \cdot H])$.
  Next, consider each block.

\begin{align}
\phi((u,a,1))|_{\bset_0} 
  &= \phi_G((u,1))
  < \phi_G((v,2)) 
  = \phi((v,b,2))|_{\bset_0}\\
\label{eq:Case 5 Line 2}\phi((u,a,1))|_{\bset_{\cset(u)}} 
  &= \phi_H((a,1))
  <  \phi_H((b,2))= \phi((v,b,2))|_{\bset_{\cset(v)}}\\
\label{eq:Case 5 Line 3}\phi((u,a,1))|_{\bset_{\ell}}
  &= 0 
  \leq \infty 
  = \phi((v,b,2))|_{\bset_{\ell}}
\end{align}

Eq.\eqref{eq:Case 5 Line 2} follows from Property~\ref{property: embedding of H}, and Eq.\eqref{eq:Case 5 Line 3} follows from the settings of other blocks $\bset_{\ell}$.
This proves the claim.\qedhere
    % In this case, by the definition of the operator $\bipp$, we have
    % $(u,a,1)(v,b,2) \in E(\bipp[G \cdot H])$.
    % In the block $\bset_0$, we have $\phi((u,a,1))|_{\bset_0} = \phi_1((u,1)) < \phi_1((v,2)) = \phi((v,b,2))|_{\bset_0}$, and in
    % the block $\bset_{\mathcal{C}(u)}$, we have
    % $\phi((u,a,1))|_{\bset_{\cset(u)}} = \phi_2((a,1)) < \phi_2((b,2))= \phi((v,b,2))|_{\bset_{\cset(v)}}$, due to property~\ref{property: embedding of H}. Therefore, it must be the case that
    % $\phi((u,a,1))< \phi((v,b,2))$.
\end{proof}
 % Sec 3: Proof of Sub Additive Ineq.
% \section{Graph Products and Subadditivity Inequalities}
\section{Proof of the General Case of Theorem~\ref{theorem:subadditive}}
\label{sec: subadditivity full}

\danupon{I changed the section title and remove bracket from Theorem statement}
In this section, we prove the main theorem. We shall restate the main
theorem here.

\restatethm{Theorem}
{\ref{theorem:subadditive}}{
For any undirected graphs $G$, $H$ and $J$ and a height-two poset $\vec P$,\danupon{Add vector on top}

\begin{align*}
\induce{(G\vee H)\times J} &\leq \induce{G\times J}+\induce{H\times J}\tag{\ref{eq:subadditive induce}}\\
\sinduce{(G\vee H)\times J} &\leq \sinduce{G\times J}+\sinduce{H\times J}\tag{\ref{eq:subadditive semi-induce}}\\
\dimension{(G\cdot H)\hstrong \vec P} &\leq \dimension{G\hstrong \vec P}+\chi(G)\dimension{H\hstrong \vec P}+\dimension{\vec P}\tag{\ref{eq:subadditive dimension}}
\end{align*}

%{\footnotesize
%\begin{center}
%\begin{tabular}[width=\textwidth]{rlr}
%{\em (\ref{eq:subadditive induce})}& $\induce{(G\vee H)\times J}$
%  & $\leq \induce{G\times J}+\induce{H\times J}$
%  \\
%$\sinduce{(G\vee H)\times J}$
%  & $\leq \sinduce{G\times J}+\sinduce{H\times J}$
%  & {\em (\ref{eq:subadditive semi-induce})}\\
%$\dimension{(G\cdot H)\hstrong \vec P}$
%  & $\leq \dimension{G\hstrong \vec P}+$\\
%  & $\chi(G)\dimension{H\hstrong \vec P} +\dimension{\vec P}$
%  & {\em (\ref{eq:subadditive dimension})}
%\end{tabular}
%\end{center}
%}
%
where $\chi(G)$ is the chromatic number of $G$.
}

\subsection{Subadditivity of Induced Matching Number (Eq.\eqref{eq:subadditive induce}).}

In this section, we prove that  $\induce{(G\vee H)\times J} \leq \induce{G\times J}+\induce{H\times J}$.
%An intuitive way to prove this is to start from a simple observation. For any graph $G$, let $K_G$ be the complete graph whose vertex set is $V(G)$ and $\bar{K}_G$ is the graph obtained from $K_G$ by adding self-loops to all vertices.
%Recall that,
%
We first observe that we can decompose edges of $G\vee H$ into two sets:

\begin{align}\label{eq:decompose0}
E(G\vee H) = \{(u, a)(v, b) \mid \mbox{$uv\in E(G)$ or $ab\in E(H)$}\}
=E_1\cup E_2\
%&= \{(u, a)(v, b) \mid uv\in E(G)\} \cup \{(u, a)(v, b) \mid ab\in E(H)\}\\
%&= E(G \times \bar{K}_H) \cup E(K_G \hstrong H)
%%\\&= E(G\times K_H)\cup E(G\Box \overline{K_H}) \cup E(K_G\times H) \cup E(\overline{K_G}\Box H)
\end{align}

\noindent
where $E_1= \{(u, a)(v, b) \in E(G\vee H): uv\in E(G)\}$ and $E_2=\{(u, a)(v, b) \in E(G\vee H): ab\in E(H)\}$. For any $i\in \{1, 2\}$, define a subgraph $G_i$ of $G\vee H$ to be $G_i=(V(G\vee H), E_i)$. Note that $E((G \vee H) \times J) \subseteq E(G_1 \times J) \cup E(G_2 \times J)$.
\danupon{Why is this a claim?}
%
%This says that we can decompose the edges of $G\vee H$ into $E_1$ and $E_2$.

\begin{claim}\label{claim:decompose1}
$\induce{(G\vee H)\times J}\leq \induce{G_1\times J}+\induce{G_2\times J}.$
\end{claim}
\begin{proof}
Let $\mset$ be any induced matching in the graph $(G\vee H) \times J$. Let $\mset_1=\mset\cap E(G_1 \times J)$ and $\mset_2=\mset\cap E(G_2 \times J)$. By Eq.\eqref{eq:decompose0}, $\mset=\mset_1\cup \mset_2$.  Observe that $\mset_1$ and $\mset_2$ are induced matchings of $G_1$ and $G_2$, respectively, since they are induced matchings of $(G\vee H)\times J$ which is a super graph of $G_1 \times J$ and $G_2 \times J$. The claim follows.
%Thus, we can conclude that
%
%\begin{align}\label{eq:decompose1}
%\induce{(G\vee H)\times J}&\leq \induce{G_1}+\induce{G_2}.
%\end{align}

\end{proof}

Now, we try to write $G_1$ and $G_2$ as a product of two other graphs. For any set $X$ of vertices, we denote by $K_X$ a complete graph whose vertex set is $V(X)$.

\begin{lemma}
\label{lemma: G_1 G_2 as products of complete graphs}
$\induce{G_1 \times J}=\induce{(K_H\hstrong G)\times J}$ and $\induce{G_2 \times J}=\induce{(K_G \hstrong H)\times J}$.
\end{lemma}
%\begin{proof}
%By definition, $E(K_G\hstrong H)=\{(u, a)(v, b): ab\in E(H)\}$ which is equal to $E_2$. We just show that $K_G\hstrong H=G_2$, proving the second equality. To prove the first equality, observe that $E(K_H\hstrong G)=\{(a, u)(b, v) : uv\in E(G)\}$ which shows that $K_H\hstrong G$ is isomorphic to $G_1$ (i.e., by mapping any vertex $(a, u)\in V(K_H\hstrong G)$ to vertex $(u, a)\in V(G_1)$).
%\end{proof}
\begin{proof}
The lemma simply follows from the fact that $K_G\hstrong H$ is exactly the same as $G_2$ and $K_H\hstrong G$ is isomorphic to $G_1$. To see this, we simply observe that $E(K_G\hstrong H)=\{(u, a)(v, b): ab\in E(H)\ \wedge\ u, v\in V(G)\}$ which is exactly the same as $E_2$, and $E(K_H\hstrong G)=\{(a, u)(b, v) : uv\in E(G)\ \wedge\ a, b\in V(H)\}$ which is almost the same as $E_1$ except that vertices are in $V(H)\times V(G)$ instead of $V(G)\times V(H)$.

\end{proof}

The simple lemma above allows us to rewrite the equation in Claim~\ref{claim:decompose1} as
%
%{\footnotesize
\begin{align}\label{eq:decompose1-2}
\induce{(G\vee H)\times J} \leq \induce{(K_H\hstrong G)\times J} + \induce{(K_G\hstrong H)\times J}
\end{align}
%}

%
Now we need the following {\em associativity} property.
\begin{lemma}
\label{lemma: associativity of products}
For any graphs $X$, $Y$ and $Z$, $(X\hstrong Y)\times Z= X\hstrong(Y\times Z)$.
\end{lemma}
\begin{proof} The following equalities simply follow from the definition of $\times$ and $\hstrong$ (cf. Definition~\ref{def:product}).

\begin{align*}
E((X\hstrong Y)\times Z)
&= \{(x, y, z)(x', y', z'): (x,y)(x',y')\in E(X\hstrong Y) 
 \ \mbox{   and } zz'\in E(Z)\}\\
 &=\{(x, y, z)(x', y', z'): \mbox{($xx'\in E(X)$ or $x=x'$)}
 \ \mbox{   and $yy'\in E(Y)$ and $zz'\in E(Z)$}\}\\
&= X\hstrong(Y\times Z)
\end{align*}

\end{proof}
This allows us to rewrite Eq.\eqref{eq:decompose1-2} as
%
%{\footnotesize
\begin{align}\label{eq:decompose2}
\induce{(G\vee H)\times J} \leq \induce{K_H\hstrong (G\times J)} +\induce{K_G\hstrong (H\times J)}
\end{align}
%}

We finish our proof with the following lemma which says that the product of any graph $X$ with a complete graph $K_L$ will not increase the induced matching number of $X$. We note that in fact the equality could be achieved, but since it is not important to this proof, we only show the upper bound.

\begin{lemma}\label{lem:induce multiply with complete graph}
For any graph $X$ and any set $L$ of vertices, $\induce{K_L\hstrong X}\leq \induce{X}$.
\end{lemma}

\begin{proof}
Let $\mset$ be any induced matching in $K_L \hstrong X$. We construct a set of edges $\mset_X \subseteq E(X)$ by projection: for each edge $(i,x)(j,y) \in E(K_L \hstrong X)$, we add an edge $xy$ to $\mset_X$. To prove the lemma, it suffices to show that $\mset_X$ is an induced matching in graph $X$.

Assume for contrary that $\mset_X$ is not an induced matching, i.e., there exists $xy, x'y'\in \mset_X$ such that either (1) $x=x'$ (making $\mset_X$ not a matching) or (2) $xx'\in E(X)$ (making $\mset_X$ not an induced matching). (We note that (1) also includes the case where multiple edges are created, i.e., $x=x'$ and $y'=y'$.)
We shall use the following simple facts.

\begin{align}
xy\in \mset_X &\implies \exists i, j: (i, x)(j, y)\in \mset\label{eq:fact2}\\
xy\in E(X) &\implies \forall i, j: (i, x)(j, y)\in E(K_L\hstrong X)\label{eq:fact1}
\end{align}

By Eq.\eqref{eq:fact2}, the assumption that $xy, x'y'\in \mset_X$ implies that edges $(i,x)(j,y)$ and $(i',x')(j',y')$ belong to $\mset$ for some $i, j, i', j'\in V(K_L)$.

\noindent{\bf Case 1:} If $x=x'$, then we have that $xy'$ is also in $\mset_X\subseteq E(X)$ and thus Eq.\eqref{eq:fact1} implies that $(i, x)(j', y')\in E(K_L\hstrong X)$. This contradicts the fact that $\mset$ is an induced matching.

\noindent{\bf Case 2:} If $xx'\in E(X)$, then Eq.\eqref{eq:fact1} implies that $(i, x)(i', x')\in E(K_L\hstrong X)$ which again contradicts the fact that $\mset$ is an induced matching.

\end{proof}

Using Lemma~\ref{lem:induce multiply with complete graph}, we can rewrite Eq.\eqref{eq:decompose2} as
%
%\begin{align*}%\label{eq:decompose3}
$\induce{(G\vee H)\times J}\leq \induce{G\times J}+\induce{H\times J}$
%\end{align*}
%
as desired.

\subsection{Subadditivity of Semi-induced Matching Number (Eq.\eqref{eq:subadditive semi-induce}).}

In this section, we prove the subadditivity property of the
semi-induced matching number.
The proof closely follows the case of the induced matching number.

%First, we give a formal definition of a semi-induced matching.\danupon{I'm still not very happy that these definitions are in this section}
%
We prove the following subadditivity theorem for semi-induced matching which is equivalent to Eq.\eqref{eq:subadditive semi-induce}.\danupon{Later: Everything that uses this theorem should try to use Eq.~\eqref{eq:subadditive semi-induce} instead.}

\begin{theorem}
\label{thm: subadditivity for semi-induced matching}
For any graphs $G$ and $H$ and any total order $\sigma$ on $V((G \vee H) \times J) $, there exist bijections $\sigma_1$ on $V(G \times J)$ and $\sigma_2$ on $V(H \times J)$ such that
\[\sinducesigma{(G \vee H) \times J }{\sigma}  \leq \sinducesigma{G \times J}{\sigma_1} +
\sinducesigma{H \times J}{\sigma_2}\,.\]
\end{theorem}

The rest of this subsection is devoted to proving the above theorem. We first decompose edge set $E(G \vee H)$ into $E_1 \cup E_2$ where $E_1= \{(u, a)(v, b) : uv\in E(G)\ \wedge\ a,b\in V(H)\}$ and $E_2=\{(u, a)(v, b) : ab\in E(H)\ \wedge\ u,v\in V(G)\}$. For any $i\in \{1, 2\}$, define a subgraph $G_i$ of $G\vee H$ to be $G_i=(V(G\vee H), E_i)$.

\begin{claim}
For any bijection $\sigma: V((G \vee H) \times J) \rightarrow [|V((G \vee H) \times J)|]$, $\sinducesigma{(G \vee H) \times J}{\sigma} \leq \sinducesigma{G_1 \times J}{\sigma} +\sinducesigma{G_2 \times J}{\sigma} $.
\end{claim}

\begin{proof}
Let $\mset$ be any $\sigma$-semi-induced matching in $(G \vee H) \times J$. Let $\mset_1 = \mset \cap E(G_1 \times J)$ and $\mset_2 = \mset \cap E(G_2 \times J)$. It is clear that $\mset = \mset_1 \cup \mset_2$, and $\mset_1$ and $\mset_2$ are $\sigma$-semi-induced matchings.

\end{proof}

Next, we write $G_1$ and $G_2$ as $G_1 = G\hstrong  K_H$ and $G_2 = K_G \hstrong  H$\danupon{Later: Should be $G_1$ is isomorphic to ..., right? I think we should just state the lemma and sketch the proof.} as in Lemma~\ref{lemma: G_1 G_2 as products of complete graphs}. So, we have that $\sinducesigma{G_1 \times J}{\sigma} = \sinducesigma{(K_H \hstrong G) \times J}{\sigma'}$, for some $\sigma'$, and that $\sinducesigma{G_2 \times J}{\sigma} = \sinducesigma{(K_G \hstrong H) \times J}{\sigma}$ (we can use $\sigma$ in the second equality since $G_2\times J = K_G\hstrong H$, but we need a different mapping $\sigma'$ in the first equality since $G_1\times J$ is only isomorphic to $K_H\hstrong G$). Then, by applying associativity in Lemma~\ref{lemma: associativity of products}, we have that
\begin{align}\label{eq:semi induced 1}\sinducesigma{(G \vee H) \times J}{\sigma} \leq \sinducesigma{K_H \hstrong (G \times J)}{\sigma'}  +\sinducesigma{K_G \hstrong (H \times J)}{\sigma}.
\end{align}
The following lemma will finish the proof.
\begin{lemma}
For any graph $X$, set $L$, and a total order $\tau$ on $V(K_L \times X)$, there exists a total order $\tau'$ on $V(X)$ such that $\sinducesigma{K_L \hstrong X}{\tau} \leq  \sinducesigma{X}{\tau'}$.
\end{lemma}

\begin{proof}\danupon{This proof is still a bit sketchy}
Let $\mset$ be any $\tau$-semi-induced matching in $K_L \hstrong X$. We construct a set of edges $\mset_X \subseteq E(X)$ by adding to $\mset_X$
an edge $xy$ for each edge $(i,x)(j,y) \in E(K_L \hstrong X)$. To prove the lemma, it suffices to define a total order $\tau'$ on vertices $V(X)$ such that $\mset_X$ is a $\tau'$-semi-induced matching in the graph $X$. %We will again use Eq.~\eqref{eq:fact1} and Eq.~\eqref{eq:fact2}.
We will use Eq.~\eqref{eq:fact1} and Eq.~\eqref{eq:fact2}; we recall them here:
% as in Lemma~\ref{lem:induce multiply with complete graph}.
\bundit{Should we remove this?}

\begin{align*}
xy\in \mset_X &\implies \exists i, j: (i, x)(j, y)\in \mset\tag{\ref{eq:fact2}}\\
xy\in E(X) &\implies \forall i, j: (i, x)(j, y)\in E(K_L\hstrong X)\tag{\ref{eq:fact1}}
\end{align*}

Before defining $\tau'$, we first argue that $\mset_X$ is a matching in $X$. Suppose otherwise; i.e., $xy, xy'\in \mset_X$ for some $x$, $y$ and $y'$. Then, from Eq.\eqref{eq:fact2}, we have edges $(i,x)(j,y)$ and $(i',x)(j',y')$ in $\mset$ for some $i,j,i',j' \in V(K_L)$. This means that $xy, xy' \in E(G)$, and therefore, from Eq.\eqref{eq:fact1}, we must also have edges $(i,x)(j',y')$ and $(i',x)(j,y)$. This contradicts the fact that $\mset$ is $\tau$-semi-induced matching, i.e. no matter how we define $\tau$, this case cannot happen.

Now, we are ready to define the total order $\tau'$ on vertices $V(X)$. Since each vertex in $X$ appears at most once in $\mset_X$, for each vertex $x \in V(X)$ that appears in $\mset_X$, we define $\tau'(x) = \tau(i,x)$ where $(i,x)$ is the vertex that appears in $\mset$. Now, it is easy to check that $\mset_X$ is $\tau'$-semi-induced matching.

\end{proof}

From this lemma, we conclude that there exists a total order $\sigma_1$ on $V(K_H \hstrong (G \times J))$ such that $\sinducesigma{K_H \hstrong (G \times J)}{\sigma'} \leq \sinducesigma{G \times J}{\sigma_1}$, and there exists a total order $\sigma_2$ on $ V(K_G \hstrong (H \times J))$ such that $\sinducesigma{K_G \hstrong (H \times J)}{\sigma} \leq \sinducesigma{H \times J}{\sigma_2}$.
% Combining these inequalities with Eq.\eqref{eq:semi induced 1} leads to Theorem~\ref{thm: subadditivity for semi-induced matching}.
Theorem~\ref{thm: subadditivity for semi-induced matching} then follows by combining these inequalities with Eq.\eqref{eq:semi induced 1}.

\subsection{Subadditivity of Poset Dimension Number (Eq.\eqref{eq:subadditive dimension}).}

In this section, we show that $\dimension{(G\cdot H)\hstrong \vec P}\leq \dimension{G\hstrong \vec P}+\chi(G)\dimension{H \hstrong \vec P}+\dimension{\vec P}$. Definitions related to poset and dimension can be found in Section~\ref{sec:prelim}.
We first note that the extended tensor product between an undirected graph $G$ and a height-two poset $\vec{P}$ is still a poset (in fact, it is a height-two poset). So, the quantity $\dimension{G\hstrong \vec{P}}$ is well-defined. This fact is formalized and proved in the following lemma.

\begin{lemma}\label{lem:poset graph times poset}
For any graph $A$ and height-two poset $\vec P$, $A\hstrong \vec{P}$ is a height-two poset.
\end{lemma}
\begin{proof} Consider any vertex $(a, p)\in V(A)\times V(\vec{P})$. Observe that if $p$ is a minimal element in $\vec{P}$, then $(a, p)$ is also a minimal element in $A\hstrong \vec{P}$; otherwise, if there is a vertex $(a', p')\in V(A)\times V(\vec{P})$ such that $(a', p')(a, p)\in E(A\hstrong \vec{P})$, then $p'p\in E(\vec{P})$, which contradicts the fact that $p$ is minimal in $\vec{P}$. A similar argument shows that if $p$ is a maximal element in $\vec{P}$, then $(a, p)$ is also maximal in $A\hstrong \vec{P}$. Since {\em every} vertex $(a, p)\in V(A)\times V(\vec{P})$ is either a minimal or maximal element (or both) in $A\hstrong \vec{P}$, the graph product $A\hstrong \vec{P}$ is a height-two poset.

\end{proof}
%
%We now list a few lemmas which will be put together in the proof of Eq.\eqref{eq:subadditive dimension} at the end of this section.
%We start proving Eq.\eqref{eq:subadditive dimension} by writing $A\hstrong \vec P$ as an intersection of two height-two posets, whose dimension numbers are easier to compute.
%
Our proof of Eq.\eqref{eq:subadditive dimension} has two steps\danupon{Change from three to two}. In the first step (Lemma~\ref{lem:poset decompose}), we write the poset $(G \cdot H) \hstrong \vec{P}$ as the intersection of two other posets $\vec{P}_1$ and $\vec{P}_2$, where $\dimension{(G \cdot H) \hstrong \vec{P}} \leq \dimension{\vec{P}_1} + \dimension{\vec{P}_2}$. In the second step, we bound the dimensions of $\vec{P}_1$ and $\vec{P}_2$.

\paragraph{Step 1: Decomposition of poset} This step is summarized in the following lemma.

\begin{lemma}\label{lem:poset decompose}
Consider any undirected graph $A$ and a height-two poset $\vec{P}$.
%Let $K_A$ be a complete undirected graph whose vertex set is $V(A)$.
 Denote by $U$ and $V$ the set of minimal and maximal elements of $\vec P$, respectively. (Since $\vec{P}$ is of height two, $U \cup V = V(\vec{P})$.) Then, $E(A \hstrong \vec P)$ can be written as
$$E(A\hstrong \vec P) = E(A\hstrong \vec K_{U, V}) \cap E(K_A\hstrong \vec P)\,$$
where $\vec K_{U, V}$ is a complete height-two poset with $U$ and $V$ as the sets of minimal and maximal elements respectively, i.e., $E(\vec K_{U, V})=\{uv : u\in U, v\in V\}$.
\end{lemma}
\begin{proof} The lemma follows from simple logical implications as shown in Fig.~\ref{fig:proof of poset decomposition}.
\begin{figure*}
\setlength{\fboxsep}{0pt}
\centering{
\fbox{
\begin{minipage}{\columnwidth}
\noindent

\begin{align*}
E(A\hstrong \vec P)
&= \{(a, p)(a',p') : \text{($aa'\in E(A)$ or $a=a'$) and $pp'\in E(\vec P)$}\}\\
&= \{(a, p)(a',p') : \text{($aa'\in E(A)$ or $a=a'$) and $p\in U$ and $p'\in V$} \\
&\ \ \ \ \ \text{and ($a\neq a'$ or $a=a'$) and $pp'\in E(\vec P)$}\}\\
&=\{(a, p)(a',p') : \text{($aa'\in E(A)$ or $a=a'$) and $p\in U$ and $p'\in V$}\} \cap \\
&\ \ \ \ \ \{(a, p)(a',p') : \text{($a\neq a'$ or $a=a'$) and $pp'\in E(\vec P)$}\}\\
%&=\{(a, p)(a',p') : \text{($aa'\in E(A)$ or $a=a'$) and $pp'\in E(K_{P', P''})$}\}\\
%&~~~~~~\cap \{(a, p)(a',p') : \text{($a\neq a'$ or $a=a'$) and $pp'\in E(P)$}\}\\
&=E(A\hstrong \vec K_{U, V}) \cap E(K_A\hstrong \vec P)\,.
\end{align*}

\end{minipage}}}
\caption{Decomposition of Poset.}\label{fig:proof of poset decomposition}
\end{figure*}
The second equality is because $pp'\in E(\vec P)$ implies that $p\in U$ and $p'\in V$ and the fact that the statement ``$a\neq a'$ or $a=a'$'' is a true statement.

\end{proof}

%So, we can rewrite $(G\cdot H)\hstrong P=((G\cdot H)\hstrong K_{P', P''}) \times (K_{G\cdot H}\hstrong P)$.  Also note a simple fact:

Let $U$ and $V$ be as in the above lemma. This allows us to write

\begin{align}
\label{eq:poset1} E((G\cdot H)\hstrong \vec P) =E((G\cdot H)\hstrong \vec K_{U, V}) \cap E(K_{G\cdot H} \hstrong \vec P)\,.
\end{align}

Note that both $(G\cdot H)\hstrong \vec K_{U, V}$ and $K_{G\cdot H} \hstrong \vec P$ are height-two posets (by Lemma~\ref{lem:poset graph times poset}). Moreover, they have the same vertex set, which is $V(G)\times V(H)\times V(\vec P)$. We next apply the following lemma which relates graph intersection to poset dimension.

\begin{lemma}\label{lem:poset times to sum}
Let $\vec{P}_1$ and $\vec P_2$ be any height-two posets on the same vertex set $V'$. Let $\vec P=(V', E(\vec P_1)\cap E(\vec P_2))$. Then, $\vec P$ is a height-two poset. Moreover, $\dimension{\vec P}\leq \dimension{\vec P_1}+\dimension{\vec P_2}$.
\end{lemma}
\begin{proof}
The proof that $\vec P$ is a height-two poset is essentially the same as the proof of Lemma~\ref{lem:poset graph times poset}.
Consider any vertex $p\in V'$. Observe that if $p$ is a minimal element in $\vec P_1$, then it is also a minimal element in $\vec P$; otherwise, if there is a vertex $p'\in V'$ such that $p'p\in E(\vec P)$, then $p'p\in E(\vec P_1)$, which contradicts the fact that $p$ is minimal in $\vec P_1$. A similar argument shows that if $p$ is a maximal element in $\vec P_1$, then it is also maximal in $\vec P$. Since {\em every} vertex in $\vec P$ is either minimal or maximal (or both), $\vec P$ is a height-two poset.

We now argue that $\dimension{\vec P}\leq \dimension{\vec P_1}+\dimension{\vec P_2}$. Let $d_i=\dimension{\vec P_i}$ and let $\phi_i:V'\rightarrow \reals^{d_i}$ be a mapping that realizes poset $\vec P_i$. We define $\phi':V'\rightarrow \reals^{d_1+d_2}$ as a concatenation of $\phi_1$ and $\phi_2$. That is, for any $p\in V'$, we let $\phi'(p) = (\phi_1(p), \phi_2(p))$. We finish the proof by showing that $\phi'$ realizes $\vec P$ using the following simple logical implications.

\begin{align*}
pp'\in E(\vec P_1) \cap E(\vec P_2) &\iff pp'\in E(\vec P_1) \mbox{ and } pp'\in E(\vec P_2)\\
&\iff \phi_1(p) < \phi_1(p') \mbox{ and } \phi_2(p) < \phi_2(p')\\
&\iff \phi'(p)<\phi'(p')\,. 
\end{align*}

\end{proof}
Using the above lemma and Eq.\eqref{eq:poset1}, we get
\begin{align}
\dimension{(G\cdot H)\hstrong \vec P} \leq \dimension{(G\cdot H)\hstrong \vec K_{U, V}} + \dimension{K_{G\cdot H} \hstrong \vec P}.\label{eq:poset2}
\end{align}

\paragraph{Step 2: Bounding the dimensions} Our next step is to bound the dimension numbers of $(G\cdot H)\hstrong \vec K_{U, V}$ and $K_{G\cdot H} \hstrong \vec P$ separately. The nice thing is that these graph products are not in the general form anymore -- one of the graphs in each product is ``complete''. This allows us to bound the dimension numbers of these graphs as in the next two lemmas.
\begin{lemma}\label{lem:poset complete graph product} For any set $L$ of vertices and any height-two poset $\vec P$,
$\dimension{K_L\hstrong \vec P}\leq\dimension{\vec P}$.
\end{lemma}
\begin{proof}
%(SKETCH)
%$E(K_L\hstrong P)=\{(x, p)(x',p') : \text{$pp'\in E(P)$}\}$. Thus, it has the same structure as $P$ so just use the same embedding.\
%
Let $d=\dimension{\vec P}$ and $\phi:V(\vec P)\rightarrow \reals^d$ be a mapping that realizes $\vec P$. We define $\phi':V(K_L)\times V(\vec P)\rightarrow \reals^d$ as $\phi'(v, p)=\phi(p)$. We complete the proof with the fact that $\phi'$ realizes $K_L\times \vec P$, proved as follows.
\begin{align*}
&(v, p)(v', p') \in E(K_L\hstrong \vec P) \\
&\iff \text{($(vv'\in E(K_L))$ or $v=v'$) and $pp'\in E(\vec P)$}\\
&\iff pp' \in E(\vec P)\\
&\iff \phi(p)<\phi(p')\\
&\iff \phi'(v, p)<\phi'(v', p')\,. 
\end{align*}
\end{proof}

We note that the equality can be attained in Lemma~\ref{lem:poset complete graph product}, but it is not important to us. The same holds for the next lemma.

%\begin{lemma}
%For any graph $X$ and $Y$ and height-two poset $P$, $\dimension{(X\times Y)\hstrong P}=\dimension{X\hstrong P}+\dimension{Y\hstrong P}$.
%\end{lemma}
%\begin{proof}
%TO DO
%\end{proof}

%So, we only have to focus on computing $\dimension{(G\cdot H)\hstrong K_{P', P''}}$.

\begin{lemma}\label{lem:poset bipartite complete graph product} For any undirected graph $A$ and sets $U$ and $V$ of vertices,
$\dimension{A\hstrong \vec{K}_{U, V}}\leq \dimension{A\hstrong \vec{K}_2}$, where $\vec K_{U,V}$ is as in Lemma~\ref{lem:poset decompose}.
\end{lemma}
\begin{proof}
Recall that the set of minimal and maximal elements of $\vec K_{U, V}$ are $U$ and $V$, respectively. Let the minimal and maximal element of $\vec K_2$ be $u$ and $v$, respectively. Let $d=\dimension{A\hstrong \vec K_2}$ and $\phi: V(A)\times V(\vec K_2)\rightarrow \reals^d$ be a mapping that realizes $A\hstrong \vec K_2$.
Define a function $r:V(\vec K_{U, V})\rightarrow V(\vec K_2)$ as $r(i)=u$ for all $i\in U$ and $r(i)=v$ otherwise. Now, define a mapping $\phi':V(A) \times V(\vec K_{U, V})\rightarrow \reals^d$ as
$$\phi'(a, i)=\phi(a, r(i)).$$
We finish the lemma by showing that $\phi'$ realizes poset $A\hstrong \vec K_{U, V}$. Observe that for any $i$ and $i'$ in $V(\vec K_{U, V})$, we have $ii'\in E(\vec K_{U, V})\iff r(i)r(i')\in E(\vec K_2)$. Thus,
\begin{align*}
&(a, i)(a',  i')\in E(A\hstrong \vec K_{U,V})\\
&\iff (aa'\in E(A) \mbox{ or } a=a') \mbox{ and } ii'\in E(\vec K_{U, V})\\
&\iff (aa'\in E(A) \mbox{ or } a=a')\mbox{ and } r(i)r(i')\in E(\vec K_2)\\
&\iff  (a,r(i))(a', r(i')) \in E(A \hstrong \vec{K}_2)\\
&\iff \phi(a, r(i))<\phi(a', r(i'))\\
&\iff \phi'(a, i)<\phi'(a', i')
\end{align*}

\end{proof}

Applying Lemma \ref{lem:poset complete graph product} and \ref{lem:poset bipartite complete graph product} to Eq.\eqref{eq:poset2}, we get
{%\small
\begin{align*}
%\label{eq:poset3}
\dimension{(G\cdot H)\hstrong \vec P} &\leq \dimension{(G\cdot H)\hstrong \vec K_2}+ \dimension{\vec P} \,.
%\nonumber
\end{align*}
}
Finally, we apply Eq.\eqref{eq:subadditive dimension special} proved in Section~\ref{sec:poset-ineq} to get the desired inequality:

\begin{align*}
\dimension{(G\cdot H)\hstrong \vec P}
\leq \dimension{G\hstrong \vec P}%\label{eq:poset4}
+\chi(G)\dimension{H \hstrong \vec P}+\dimension{\vec P}.
\end{align*}

 % Sec 4: Proof of The General Case of Theorem 1.1.
%\section{Hardness of Bipartite (Semi-)Induced Matching and Poset Dimension}
\section{Hardness from Graph Products}
\label{sec: applications}\label{sec: proof of induced matching}

In this section, we show applications of subadditivity inequalities presented in Theorem~\ref{theorem:subadditive} in proving the tight hardness of approximating the induced matching number, the semi-induced matching number, and the poset dimension number.  Moreover, we prove the hardness of $d^{1/2-\epsilon}$ for approximating the induced matching number of $d$-regular bipartite graphs.
%
%More applications can be found in Section~\ref{sec:rest} and  Appendix~\ref{sec:pricing}.
%

% Throughout this section we will use $\bip[G]=G\times K_2$ and $\bipp[G]=G\hstrong K_2$ as defined in Section~\ref{sec: simplified subadditivity}.

\subsection{Bipartite Induced Matching.}\label{subsec:induced matching}\label{sec:hardness-semi-induced}
In this section, we prove that induced and semi-induced matching problems in bipartite graphs are hard to approximate to within a factor of $n^{1-\epsilon}$.
%
%Recall from Definition~\ref{def:semi induced number} that in the {\em semi-induced matching} problem, given a graph $G$, the goal is to find a $\sigma$-semi-induced matching of maximum size, for some permutation $\sigma:V($.
%
Recall that we use  $\bip[G]=G \times K_2$ and $\bipp[G]=G\hstrong K_2$.
%
%Also let $G^k=G\vee  \ldots \vee G$ where $G$ appears $k$ times.

We have already sketched the proof of the hardness of the induced matching problem in Section~\ref{sec:intro proof} and will give more detail here. We can actually say something stronger than just the hardness
of these problems. In fact, it is hard to distinguish between the case where an input graph $G$ has large $\induce{G}$ and small $\sinduce{G}$, as stated in Theorem~\ref{thm:semi-induced} below.

% Despite the existence of an induced matching of size $|V(G)|^{1-\epsilon}$, we cannot find, in polynomial time, a $\sigma$-induced matching of size $|V(G)|^{\epsilon}$. This holds even if an algorithm is allowed to choose its own permutation $\sigma$.

\begin{theorem}\label{thm:semi-induced}
Given any bipartite graph $G$ and $\epsilon >0$, unless $\NP \subseteq
\ZPP$, no polynomial-time algorithm can distinguish between the following two cases:
\squishlist
% \begin{itemize}
  \item (\yi) $\induce{G}\geq |V(G)|^{1-\epsilon}$.
  %In particular,     $\sinducesigma{G}{\sigma} \geq |V(G)|^{1-\epsilon}$ for any  $\sigma: V(G) \rightarrow V(G)$.

  \item (\ni) $\sinduce{G}\leq |V(G)|^{\epsilon}$.
% \end{itemize}
\squishend
\end{theorem}

Note that $\sinduce{G}\geq \induce{G}$; thus, Theorem~\ref{thm:semi-induced} implies that no polynomial-time algorithm can distinguish between the cases where $\induce{G}\geq |V(G)|^{1-\epsilon}$ and $\induce{G}\leq |V(G)|^{\epsilon}$ as well as the cases where $\sinduce{G}\geq |V(G)|^{1-\epsilon}$ and $\sinduce{G}\leq |V(G)|^{\epsilon}$. Theorem~\ref{thm:semi-induced} thus implies the hardness of both induced and semi-induced matching problems in bipartite graphs.

%\begin{proof}[Proof of Theorem~\ref{thm:semi-induced}]

\paragraph{Proof of Theorem~\ref{thm:semi-induced}}
Our proof is based on a reduction from the maximum independent set problem. As discussed earlier, we start from the result of~\cite{Hastad96} instead of~\cite{KhotPonnuswami}, to keep the parameters simple.
\begin{theorem}[\cite{Hastad96}]
\label{thm: independent set hardness}
Let $\epsilon >0$ be any constant. Given a graph $G$,
unless \NP = \ZPP, no polynomial-time algorithm can distinguish between the following two
cases:
\squishlist
\item (\yi) %There is an independent set in $G$ of size
  $\alpha(G)\geq |V(G)|^{1-\epsilon}$.
\item (\ni) %Any independent set in $G$ has size at most
$\alpha(G)\leq |V(G)|^{\epsilon}$.
%%\end{itemize}
\squishend
\end{theorem}

We start from a graph $G$ given by Theorem~\ref{thm: independent set hardness} and return as an output a graph
$\bipp[G^k]$, where $k=(1/\epsilon)$ and $G^k=G\vee G\vee \ldots \vee G$ (there are $k$ copies of $G$).
By construction, the number of vertices in $\bipp[G^k]$ is $n=2|V(G)|^k$.

If $G$ is a \yi, then we know that $\alpha(G)\geq |V(G)|^{1-\epsilon}$. We will use the following lemma, essentially due to~\cite{ElbassioniRRS09}, but since it is not explicitly stated in \cite{ElbassioniRRS09}, we shall provide the proof for completeness.

\parinya{I changed this since we don't need the upper bound part.}
\begin{lemma}[Implicit in \cite{ElbassioniRRS09}]\label{lem:elbassioni} For any graph $G$,
\[\alpha(G)\leq \induce{\bipp[G]} \,.\]
\end{lemma}

\begin{proof}

First, we show the lower bound of $\induce{\bipp[G]}$.
Let $I \subseteq V(G)$ be an independent set in $G$.
Clearly, the set of edges $\mset=\set{(u,1)(u,2):  u \in I}$ corresponding to $I$
is an induced matching in $\bip[G]$ since if there is $u, u'\in I$ such that $(u, 1)(u', 2)\in \mset$, then $uu'\in E(G)$, contradicting the fact that $I$ is an independent set.
Thus, $\bip[G]$ has an induced matching of size at least $\alpha(G)$. 
\begin{comment}
Next, we show the upper bound of $\induce{\bipp[G]}$.
Consider any induced matching $\mset$ in $\bipp[G]$.
We may write $\mset=\mset_1\cup \mset_2$, where $\mset_1$ consists of edges that are
also present in $\bip[G]$ and $\mset_2= \mset \setminus \mset_1$.

It can be seen that $|\mset_1|\leq\induce{\bip[G]}$ and
$|\mset_2|\leq\alpha(G)$.
The latter is because we can obtain an independent set in $G$ by
choosing each vertex $u$ corresponding to edges $(u,1)(u,2)$ in $\mset_2$.
Thus, the lemma follows.

\end{comment}
\end{proof}

We recall the following standard fact in graph theory. 

\begin{lemma}[Folklore; See e.g. \cite{TrevisanLecture}]\label{lem:independence number product} For any graphs $G$ and $H$, $\alpha(G\vee H)=\alpha(G)\alpha(H).$ In particular, $\alpha(G^k)=(\alpha(G))^k.$
\end{lemma}

It follows from the above lemmas that $\induce{\bipp[G^k]}\geq \alpha(G^k) \geq (\alpha(G))^k\geq |V(G)|^{k(1-\epsilon)}=\Omega(n^{1-\epsilon})$.

\medskip Now if $G$ is a \ni, we can invoke the next lemma, implicitly used in~\ref{lem:elbassioni}.
\begin{lemma}
\label{lemma: upper bounding the semi-induced matching for no-instance}
For any graph $G$,we have
\[
\sinduce{\bipp[G]}
     \leq \sinduce{\bipm[G]} + \alpha(G)\,.
\]
\end{lemma}
\begin{proof}
Let $\sigma$ be any total order on $V(G)$. Consider any $\sigma$-semi-induced matching $\mset$ in
$\bipp[G]$.
We may write $\mset=\mset_1\cup \mset_2$, where $\mset_1$ consists of
edges that are also present in $\bipm[G]$ and
$\mset_2=\mset\setminus\mset_1$.
It can be seen that $|\mset_1|\leq\sinducesigma{\bipm[G]}{\sigma}$
and $|\mset_2|\leq\alpha(G)$.
The former inequality is because $E(\bipm[G])\subseteq E(\bipp[G])$. The latter inequality is because we can define an independent set in $G$ by
choosing vertices corresponding to edges in $\mset_2$ (since every edge in $\mset_2$ is in the form $(v, 1)(v, 2)$ for some $v\in V(G)$). Since this is true for all $\sigma$, the lemma follows.

\end{proof}

%%%%%%%%%%%%%%%%%%%
% This is the old statement and proof of the above lemma where we have \sigma in the statement.
%%%%%%%%%%%%%%%%%%
\begin{comment}
\begin{lemma}
\label{lemma: upper bounding the semi-induced matching for no-instance}
For any graph $H$ and a permutation $\sigma: V(H) \rightarrow V(H)$,
we have
\[
  \alpha(H) \leq \sinducesigma{\bipp[H]}{\sigma}
     \leq \sinducesigma{\bipm[H]}{\sigma} + \alpha(H)\,.
\]
\end{lemma}
\begin{proof}
%First, the lower bound of $\sinducesigma{\bipp[H]}{\sigma}$ is implied by the first part of the proof of Lemma~\ref{lem:elbassioni} since the existence of an induced matching implies the existence of a $\sigma$-semi-induced matching for any $\sigma$.
%
First, note that $\sinducesigma{\bipp[H]}{\sigma}\geq \induce{\bipp[H]}\geq \alpha(H)$, where the first inequality is because the existence of an induced matching implies the existence of a $\sigma$-semi-induced matching for any $\sigma$, and the second inequality is by Lemma~\ref{lem:elbassioni}. We thus get the lower bound of $\sinducesigma{\bipp[H]}{\sigma}$.

To show the upper bound, consider any $\sigma$-semi-induced matching $\mset$ in
$\bipp[G]$.
We may write $\mset=\mset_1\cup \mset_2$, where $\mset_1$ consists of
edges that are also present in $\bipm[G]$ and
$\mset_2=\mset\setminus\mset_1$.
It can be seen that $|\mset_1|\leq\sinducesigma{\bipm[G]}{\sigma}$
and $|\mset_2|\leq\alpha(G)$.
The latter inequality is because we can define an independent set in $G$ by
choosing vertices corresponding to edges in $\mset_2$.

\end{proof}
\end{comment}

%Also note the following corollary.
\begin{corollary}[Immediate from Theorem~\ref{theorem:subadditive}]
\label{cor: main for semi induced matching}
For any integer $k$ and graph $G$, $\sinduce{\bipm[G^k]} \leq \sinduce{\bipm[G]}k\,.$
\end{corollary}

By applying Lemma~\ref{lemma: upper bounding the semi-induced matching for no-instance}, we have that
$\sinduce{\bipp[G^k]}\leq
 \sinduce{\bipm[G^k]} + \alpha(G)^k$, and
by invoking Corollary~\ref{cor: main for semi induced matching},
we have
$\sinduce{\bipp[G^k]}
  \leq \sinduce{\bipm[G]}k + \alpha(G)^k$.
Then we plug in $k=1/\epsilon$ and $\alpha(G) \leq |V(G)|^{\epsilon}$
and conclude that
$\sinduce{\bipp[G^k]} \leq
O(|V(G)|)\leq{n^{2\epsilon}}$ when $G$ is a \ni.
This completes the proof of Theorem~\ref{thm:semi-induced}.

To get a better hardness result, we start the reduction from Theorem~1 in~\cite{KhotPonnuswami} using the value of $k = \log^{\gamma} |V(G)|$ instead of $(1/\epsilon)$ (where $\gamma$ is as in Theorem~1 in~\cite{KhotPonnuswami}). This will give the hardness of $n/2^{\log^{3/4+ \gamma} n}$ under the assumption that $\NP \not \subseteq {\sf BPTIME}(n^{\poly \log n})$.

%%%%%%%%%%%%%%%%%%%%%%%%%%%%%%%%%%%%%%%%%%%%%%%%%%%%%%%%%%%%%%%%%%
%%%%%%%%%%%%%%%%%%%%%%%%%%%%%%%%%%%%%%%%%%%%%%%%%%%%%%%%%%%%%%%%%%
%%%%%%%%%%%%%%%%%%%%%%%%%%%%%%%%%%%%%%%%%%%%%%%%%%%%%%%%%%%%%%%%%%
%%%%%%%%%%%%%%%%%%%%%%%%%%%%%%%%%%%%%%%%%%%%%%%%%%%%%%%%%%%%%%%%%%

\subsection{Induced and Semi-induced Matchings on $d$-Regular Graphs.}
\label{sec:indmatching-d-regular}

Here we show the hardness result of the induced matching
problem on $d$-regular bipartite graphs.
For this, we need an instance $G$ of the maximum independent set
problem such that $G$ is $d$-regular.
It can be seen that the following hardness result follows
from Trevisan's construction in~\cite{Trevisan01} on the hardness of
the maximum independent set problem on bounded degree graphs.
As it is not guaranteed that an instance $G$ obtained from Trevisan's
construction has regular degree, we have to slightly modify the
construction in the same way as in~\cite{ChalermsookPricing} and~\cite{AndrewsCGKTZ10}.
%\parinya{I added the reference. Basically, this is the first paper that modifies Trevisan's construction in this way. Although the modification is simple, I think it's the right thing to do to cite the paper where it appears first time :-) }.
\begin{comment}
\begin{theorem}[\cite{Trevisan01}, statement follows~\cite{ChalermsookPricing}]
%\label{thm: bounded-degree MIS}
Given an $n$-variable 3SAT formula $\phi$, any sufficiently small
constant $\epsilon >0$ and any integer $\lambda >0 $, there is a
randomized algorithm that constructs a graph $G$ in which all vertices
have degrees $\Delta = 2^{\lambda  \poly (\frac 1 {\epsilon})}$ such
that w.h.p.:
\squishlist
%\begin{itemize}
  \item ({\sc Yes-Instance}) If $\phi$ is satisfiable, then $\alpha(G)\geq |V(G)|/\Delta^{\epsilon}$.
  \item ({\sc No-Instance}) If $\phi$ is not satisfiable, then $\alpha(G)\leq |V(G)|/\Delta^{1-\epsilon}$.
%\end{itemize}
\squishend
The construction size is $|V(G)| = n^{\lambda \poly(\frac
  1{\epsilon})}$, and the reduction runs in time $n^{\lambda
  \poly(\frac 1{\epsilon})}$.
%Moreover, the algorithm can be made deterministic with a running time of $2^{O(\Delta)} n^{\lambda \poly (\frac 1{\epsilon})}$.
\end{theorem}
\end{comment}

\begin{theorem}[\cite{Trevisan01}, modified from Theorem~4 in \cite{ChalermsookPricing}]
\label{thm: bounded-degree MIS}
Let $\lambda:\mathbb{N}\rightarrow\mathbb{N}$ be any function.
Assuming that $\NP \not\subseteq {\sf ZPTIME}(n^{O(\lambda(n))})$, there is no polynomial-time algorithm that can solve the following problem.

For any constant $\epsilon >0$ and any integer $q$, given a graph $G$ of size $q^{O(\lambda(q))}$ such that all vertices have degree $\Delta=2^{O(\lambda(q))}$, the goal is to distinguish between the following two cases:
\squishlist
%\begin{itemize}
  \item ({\sc Yes-Instance}) $\alpha(G)\geq |V(G)|/\Delta^{\epsilon}$.
  \item ({\sc No-Instance}) $\alpha(G)\leq |V(G)|/\Delta^{1-\epsilon}$.
%\end{itemize}
\squishend
\end{theorem}

%We would like to emphasize the properties of the graph we constructed in the proof of Theorem~\ref{thm: bounded degree induced matching} for future use (more specifically, we use these properties to prove the hardness of pricing problems in Section~\ref{sec:rest}).

This theorem gives the hardness result of the bounded degree version of the maximum independent set problem, and it allows us to use the value of $\lambda$ to specify the degree of vertices we want. For instance, if we use $\lambda(q) =c$ for some constant $c$, then we get the hardness of the constant-degree maximum independent set problem, and the hardness assumption is $\NP \not\subseteq \ZPP$. But, if we choose $\lambda(q) = O(\log \log q)$, then we have the hardness for the logarithmic-degree maximum independent set problem with the hardness assumption of $\NP \not\subseteq {\sf ZPTIME}(n^{O(\log \log n)})$.

We will need the following lemma which will also be used later to prove the hardness of pricing problems.

\begin{comment}
\begin{lemma}
\label{lemma: semi-induced matching in terms of graph products}
Let $\epsilon >0$, integer $\gamma$, and $\Delta= \Delta(\gamma)$ be parameters.
%There is a reduction that takes as input an $n$-variable 3SAT formula $\phi$ and constructs
Given a $\Delta$-regular graph $G$ on
$O(\gamma^{O(\log\Delta(\gamma))})$ vertices and an empty graph $H$ on
$\Delta(q)$ vertices, unless $\NP \subseteq {\sf ZPTIME}(n^{O(\log
  \Delta(n))})$, there is no polynomial-time algorithm that
distinguishes between the following cases:
\squishlist
\item (\yi) $\induce{\bipp[G \vee H]}\geq |V(G)| \Delta^{1-\epsilon}$.

\item (\ni) $\sinduce{\bipp[G \vee H]} \leq |V(G)| \Delta^{\epsilon}$.
\squishend

% Moreover, the degree of vertices in $G$ is $\Delta$, and $H$ is an empty graph on $\Delta$ vertices. The size of $G$ is $n^{O(\log \Delta)}$.
\end{lemma}
\end{comment}

\begin{lemma}
\label{lemma: semi-induced matching in terms of graph products}
Let $\Delta:\mathbb{N}\rightarrow \mathbb{N}$. Assuming that  $\NP \nsubseteq {\sf ZPTIME}(n^{O(\log
  \Delta(n))})$, there is no polynomial time algorithm that can solve the following problem: For any $\epsilon$ and integer $q$, given a $\Delta(q)$-regular graph $G$ on
$O(q^{O(\log\Delta(q))})$ vertices and an empty graph $H$ on
$\Delta(q)$ vertices, the goal is to distinguish between the following cases:
\squishlist
\item (\yi) $\induce{\bipp[G \vee H]}\geq |V(G)| (\Delta(q))^{1-\epsilon}$

\item (\ni) $\sinduce{\bipp[G \vee H]} \leq |V(G)| (\Delta(q))^{\epsilon}$.
\squishend
%
% Moreover, the degree of vertices in $G$ is $\Delta$, and $H$ is an empty graph on $\Delta$ vertices. The size of $G$ is $n^{O(\log \Delta)}$.
\end{lemma}

%We remark here that the value of $\Delta$ in the lemma can be dependent on $n$.

Note that we use $\Delta(q) = c$ for some constant $c$ in the proof of Theorem~\ref{thm: bounded degree induced matching}, while we choose $\Delta(q) =\poly \log q$ when proving the hardness of pricing problems in the next section.
From the lemma, the hardness of the induced matching problem on $d$-regular graphs follows immediately.\danupon{Say what assumptions we get from these parameters and say that they are natural.}

\begin{proof}[Of Lemma~\ref{lemma: semi-induced matching in terms of graph products}]
Let $\epsilon>0$ be a constant. Let $G$ be a $\Delta(q)$-regular graph obtained from Theorem~\ref{thm: bounded-degree MIS} (choosing $\lambda(q) = O(\log \Delta(q))$ so that we get the graph $G$ of degree $\Delta(q)$ and $|V(G)|= q^{O(\log \Delta(q))}$.)
We output a graph $\bipp[G \vee H]$, where
$H=\bar K_{\Delta(q)}$ is an empty graph on $\Delta(q)$ vertices. This finishes the construction.

Notice that the number of vertices in $\bipp[G \vee H]$ is $n = 2|V(G)||V(H)| = 2 |V(G)| \Delta(q)$. In the \yi, by Lemma~\ref{lem:elbassioni} and \ref{lem:independence number product}, we have that $\induce{\bipp[G \vee H]} \geq \alpha(G\vee H)\geq \alpha(G) \alpha(H) \geq |V(G)| \Delta(q)^{1-\epsilon}$ because $\alpha(G) \geq |V(G)|/\Delta(q)^{\epsilon}$ in the \yi.

In the \ni, we have $\sinduce{\bipp[G \vee H]} \leq \alpha(G \vee H) + \sinduce{\bip[G \vee H]}$ (by Lemma~\ref{lemma: upper bounding the semi-induced matching for no-instance}).
%
%In the \ni, we first fix a permutation $\sigma$. We have $\sinducesigma{\bipp[G \vee H]}{\sigma} \leq \alpha(G \vee H) + \sinducesigma{\bip[G \vee H]}{\sigma}$ (from Lemma~\ref{lemma: upper bounding the semi-induced matching for no-instance}).
%
The first term is at most $\alpha(G) \alpha(H) \leq |V(G)|\Delta(q)^{\epsilon}$ (by Lemma~\ref{lem:independence number product}). The second term is at most, by Eq.\eqref{eq:subadditive semi-induce} in Theorem~\ref{theorem:main}, $\sinduce{\bip[G]} + \sinduce{\bip[H]}\leq |V(G)| + \Delta(q) \leq 2|V(G)|$, for our choice of $\Delta(q)$. Therefore, in the \ni, the value of the solution is at most $O(|V(G)| \Delta(q)^{\epsilon})$.

Since $\lambda(q) = O(\log \Delta(q))$, we get the complexity assumption of $\NP \not\subseteq {\sf ZPTIME}(n^{O(\log \Delta(n))})$.

\end{proof}

\begin{theorem}
\label{thm: bounded degree induced matching}
Let $d \leq \poly\log n$ be any sufficiently large number. For any constant $\epsilon >0$, unless $\NP \subseteq {\sf ZPTIME}(n^{\poly \log n})$, it is hard to approximate the induced matching problem on $d$-regular graphs to within a factor of $d^{1/2 -\epsilon}$.
\begin{comment}
there is a number $Z$ such that it is hard to
distinguish between the following two cases:
%\begin{itemize}
\squishlist
\item (\yi) $\induce{G}\geq Z$

\item (\ni)  $\sinduce{G} \leq Z/d^{1/2-\epsilon}$.
%\end{itemize}
\squishend
\end{comment}
\end{theorem}
\begin{proof}
From Lemma~\ref{lemma: semi-induced matching in terms of graph products}, observe that the degree of each vertex in $\bipp[G \vee H]$ is $d = \Delta^2 +1$: each vertex $(v,a,1) \in \bipp[G \vee H]$ is connected to $(v,a,2)$ and other vertices $(u,b,2)$ for all $u \in V(G)$ and $b \in V(H)$. The gap between \yi and \ni is $\Delta^{1-2\epsilon} \geq d^{1/2 - O(\epsilon)}$. We use $\Delta \leq O(\poly \log n)$, so the running time of the reduction is $n^{O(\log \log n)}$.
\end{proof}

%%%%%%%%%%%%%%%%%%%%%%%%%%%%%%%%%%%%%%%%%%%%%%%%%%%%%%%%%%%%%%%%%%
%%%%%%%%%%%%%%%%%%%%%%%%%%%%%%%%%%%%%%%%%%%%%%%%%%%%%%%%%%%%%%%%%%
%%%%%%%%%%%%%%%%%%%%%%%%%%%%%%%%%%%%%%%%%%%%%%%%%%%%%%%%%%%%%%%%%%
%%%%%%%%%%%%%%%%%%%%%%%%%%%%%%%%%%%%%%%%%%%%%%%%%%%%%%%%%%%%%%%%%%

\subsection{Poset Dimension.}
We now prove the $n^{1-\epsilon}$-hardness of approximating poset
dimension. Note that here we use  $\bip[G]=G \times \vec{K}_2$ and $\bipp[G]=G\hstrong \vec{K}_2$.
We denote by $G^k=G\cdot G\cdot \,\ldots\, \cdot G$ where $G$ appears $k$ times.

\paragraph{Construction} We will need the following hardness result of the graph coloring problem, due to Feige and Kilian~\cite{FeigeK98}.
(In fact, there is a stronger hardness  result by Khot and Ponnuswami~\cite{KhotPonnuswami}, but we use the result of Feige and Kilian to keep the presentation simple.)

\begin{theorem}[\cite{FeigeK98}]
\label{thm: coloring hardness}
Let $\epsilon >0$ be any constant. Given a graph $G$, unless \NP = \ZPP, no polynomial-time algorithm can distinguish between the following two cases:
% (\yi) $\chi(G)\leq n^{\epsilon}$ and (\ni) $\chi(G)\geq n^{1-\epsilon}$.
%%\begin{itemize}
\squishlist
  \item (\yi) $\chi(G) \leq |V(G)|^{\epsilon}$.

  \item (\ni) $\chi(G) \geq |V(G)|^{1-\epsilon}$.
%\end{itemize}
\squishend
\end{theorem}

Our reduction starts from the instance $G$ given by Theorem~\ref{thm: coloring hardness}. Then we output $\bip[G^k]$ where $k = 1/\epsilon$. The construction size is $n = 2|V(G)|^k$.

\paragraph{Analysis} We need the following lemma, similar in spirit to
Lemma~\ref{lem:elbassioni}.
Since we state the lemma in our language, we provide the proof for completeness.

\begin{lemma}[Implicit in \cite{HJ07}]\label{lemma: poset dimension connections}
For any graph $G$, $\chi(G) \leq \dimension{\bip[G]} \leq \dimension{\bipp[G]} + \chi(G)$
\end{lemma}

% Since we state the lemma in our language, we provide its proof in \fullversion{Appendix~\ref{sec:proofs-apps}}{the full version} for completeness.
% We will also need the following lemma which bounds the chromatic number of the $k$-fold product of graphs.

%%%%%%%%%%%%
\begin{proof}
Recall that $\bipm[G]$ is almost identical to $\bipp[G]$ except that
$\bipm[G]$ has no edges of the form $(u, 1)(u,2)$ for all $u\in V(G)$.
We say that a mapping $\psi:\bipp[G] \rightarrow \R$ {\bf hits}
a vertex $u \in V(G)$ if $\psi((u,1)) > \psi((u,2))$, and
$\psi((v,1)) \leq \psi((w,2))$ whenever $vw \in E(G)$. In other words, $\psi$ is a linear order that ``reverses'' the direction of edge $(u,1)(u,2) \in \bipp[G]$.

The following claim was proved by Hegde and Jain in~\cite{HJ07}.
We restated it here in our terminology and also provide the proof for completeness.

\begin{claim}[\cite{HJ07}]
\label{claim: HJ}
Let $X \subseteq V(G)$. There is a mapping $\psi$ that hits all
vertices in $X$ if and only if $X$ is an independent set in $G$.
\end{claim}
\begin{proof}
One direction is easy to see. Suppose $X \subseteq V(G)$ contains $u,v$ such that $uv \in E(G)$, so the function $\psi$ hitting $\set{u,v}$ means that $\psi((u,2)) > \psi((u,1)) \geq \psi((v,2)) > \psi((v,1)) \geq \psi((u,2))$, an obvious contradiction. In short, we cannot expect to have a linear extension of a directed cycle.

Conversely, assume that $X$ is an independent set. We define function $\psi$ by processing vertices in $X$ in arbitrary order. When vertex $u \in X$ is considered, we set the values $\psi((u,1)) = 2$ and $\psi((u,2)) = 1$. After we finish, we set $\psi((u,1)) = 0$ and $\psi((u,2)) =  3$ for all other vertices $u$'s whose values $\psi$ were undefined. Now notice that the only way to violate the hitting property of $\psi$ is to have $\psi((u,1)) = 2$ and $\psi((v,2)) = 1$ for some $uv \in E(G)$, but this is impossible because $X$ is an independent set.

\end{proof}

Now we prove the inequality.

The lower bound follows immediately from Claim~\ref{claim: HJ}.
Let $\tilde \phi: V(\bip[G])\rightarrow\R^d$ be a mapping that realizes
$\bip[G]$.
For each coordinate $q$, define $\psi_{q}$ as
$\psi_{q}((u,1)) = \tilde\phi(((u,1))[q]$ and $\psi_{q}((u,2)) = \tilde \phi((u,2))[q]$, i.e. $\psi_{q}$ is function where we project the $q$th
coordinate of $\tilde \phi$.
Observe that, for each vertex $u \in V(G)$, there must be some $\psi_{q}$ that
hits the vertex $u$.
We argue that there is a valid coloring of $G$ using at most $d$
colors.
To see this, construct a coloring as follows.
For each vertex $u \in V(G)$, assign a color $q$ to $u$, where $q$
is the first coordinate such that $\psi_{q}$ hits the vertex $u$.
Claim~\ref{claim: HJ} guarantees that each color class is an
independent set.
Thus, the coloring is valid and $\chi(G) \leq d$.

To prove the upper bound, let $\phi$ be a function that realizes the
poset $\bipp[G]$.
Then, for any vertices $u \neq v$ of $G$, we have
$\phi((u,1))\leq\phi((v,2))$ if and only if $uv \in E(G)$. Then we need to extend $\phi$ into $\tilde \phi$ such that (i) Each vertex $u \in V(G)$ is hit by some coordinate of $\tilde \phi$, and (ii) For any two vertices of the form $(u,i)$ and  $(v,i)$ where $u \neq v$ and $i \in \set{1,2}$, we must have some coordinates $q, q'$ such that $\tilde \phi((u,1))[q] < \tilde \phi((v,1))[q]$ and $\tilde \phi((u,1))[q'] > \tilde \phi((v,1))[q']$. We only need two more coordinates to satisfy (ii). To deal with (i), it suffices to find a collection of functions $\psi_j$ such
that each vertex $u \in V(G)$ is hit by some $\psi_j$.
Then $\tilde\phi$ can be defined by concatenating $\phi$ with all the
mappings $\psi_j$.
Each such $\psi_j$ can be obtained from Claim~\ref{claim: HJ} by
defining, for each $j$, $\psi_j$ to be a map that hits all vertices in
a color class $j$.

\end{proof}

%%%%%%%%%%%%

We will also need the following lemma which bounds the chromatic number of the $k$-fold product of graphs.

\begin{lemma}[{\cite{LinialV89,GaoZ96,KlavzarH02} and \cite[Cor. 3.4.5]{Scheinerman1997fractional})}]
\label{lem:dimension number product} For any graph $G$ and any number $k$,
$\left(\frac{\chi(G)}{\log|V(G)|}\right)^k\leq \chi(G^k)\leq (\chi(G))^k$.
\end{lemma}

We are now ready to analyze the gap between the \yi and \ni.

Suppose that $G$ is a \ni. Then $\chi(G) \geq  |V(G)|^{1-\epsilon}$. By Lemma \ref{lemma: poset dimension connections} and \ref{lem:dimension number product} and for sufficiently large $|V(G)|$ (so that $\log  |V(G)| \leq |V(G)|^{\epsilon}$), we have that
$\dimension{\bipm[G^k]}\geq \chi(G^k)\geq \left(\frac{\chi(G)}{\log |V(G)|}\right)^k \geq \left(\frac{|V(G)|^{1-\epsilon}}{\log |V(G)|}\right)^k \geq \frac{n^{1-\epsilon}}{(2\log |V(G)|)^k} \geq n^{1-O(\epsilon)}.$

%****If $G$ is a \ni, then $\chi(G) \geq  |V(G)|^{1-\epsilon}$, so by Lemma \ref{lemma: poset dimension connections} and \ref{lem:dimension number product} and for sufficiently large $n$ (so that $\log n\leq n^{\epsilon}$), $\dimension{\bipm[G^k]}\geq \chi(G^k)\geq \left(\frac{\chi(G)}{\log |V(G)|}\right)^k\geq \frac{n^{1-\epsilon}}{\log^k |V(G)|}=\Omega(n^{1-2\epsilon}).$ This gives us the gap of $n^{1-O(\epsilon)}$.

For the \yi, we have that $\dimension{\bipm[G^k]}\leq \dimension{\bipp[G^k]}+\chi(G^k)$. By Lemma~\ref{lem:dimension number product}, the term $\chi(G^k)$ can be upper bounded by $\chi(G)^k \leq |V(G)|^{\epsilon k} = |V(G)| \leq n^{\epsilon}$ because $\chi(G) \leq |V(G)|^{\epsilon}$ in the \yi. We use the following claim to bound the term $\dimension{\bipp[G^k]}$.
% Due to the space limit, the proof is deferred to \fullversion{Appendix~\ref{sec:proofs-apps}}{the full version}.

\begin{claim}%[Immediate from Theorem~\ref{theorem:subadditive}]
\label{cor:poset dimension}
For any graph $G$ and integer $k$,
$\dimension{\bipp[G^k]} \leq \chi(G)\dimension{\bipp[G]}k+k.$
\end{claim}
%
%We will prove the claim below. Meanwhile, we finish the bound of the \yi cost.
\begin{proof}
Note that $\dimension{\vec{K}_2}=1$. By
Theorem~\ref{theorem:subadditive}, we have that

\begin{align*} \dimension{\bipp[G^k]} &\leq
\dimension{\bipp[G^{k-1}]}+\chi(G)\cdot\dimension{\bipp[G]} +1\\ 
&\leq
  \dimension{\bipp[G^{k-2}]}+2\cdot\chi(G)\cdot\dimension{\bipp[G]} +2\\ 
&\vdots \\
&\leq \dimension{\bipp[G]}+(k-1)\cdot\chi(G)\cdot\dimension{\bipp[G]} + (k-1)\\ 
&\leq k\cdot\chi(G)\cdot\dimension{\bipp[G]}+k
\end{align*}
\end{proof}

By Claim~\ref{cor:poset dimension} and the fact that 
$\dimension{\bipp[G]} \leq 2|V(G)|$, we have 
\[
\dimension{\bipp[G^k]}
  \leq \chi(G)\dimension{\bipp[G]}k+k 
  \leq 2|V(G)|^{\epsilon} |V(G)| k +k
\]
This implies that 
$\dimension{\bipp[G^k]} 
  \leq O(|V(G)|^{1+\epsilon}) 
  \leq n^{2\epsilon}$.
Therefore, $\dimension{\bip[G^k]} \leq n^{O(\epsilon)}$, implying 
the gap of $n^{1-O(\epsilon)}$.

\section{More Applications}
\label{sec:rest}

In this section, we present all other applications discussed in the introduction.
% : the maximum feasible subsystem problem and the donation
% center location problem.

\subsection{Maximum Feasible Subsystem with 0/1 Coefficients.}
\label{sec:mfs}

%Elbassioni et al.~\cite{ElbassioniRRS09} showed that the bipartite semi-induced matching problem is a special case of \MFS as formally stated in Theorem~\ref{thm: elbassioni MRFS}. This means that any hardness result for the semi-induced matching problem immediately implies the same hardness result for \MFS.

The following reduction follows the ideas implicit in Theorem~3.5 in \cite{ElbassioniRRS09}. For completeness, we include the proof in \fullversion{Appendix~\ref{sec: omitted proofs from rest}}{the full version}.

%We note that we will reduce from a variant of semi-induced matching problem where a permutation $\sigma$ is given as part of the input and we want to approximate $\sinducesigma{G}{\sigma}$ (the hardness of this variant also follows from Theorem~\ref{thm:semi-induced}); this helps making the proof cleaner.

\begin{theorem}[\cite{ElbassioniRRS09}]
\label{thm: elbassioni MRFS}
Consider an instance $G=(V_1\cup V_2, E)$ of the bipartite semi-induced matching problem. There is a polynomial time reduction that, for any $0<\beta\leq |V(G)|$, outputs an instance $\aset= (A,\ell, \mu)$ of {\sc Mrfs} satisfying the following properties:
\squishlist
    \item ({\sc Size}) Matrix $A$ is an $m$-by-$n$ matrix, where $m = |V_1|$, $n=|V_2|$ and  $L = \max_{i\in [n] } \set{\ell_i} \leq (\beta n)^{O(n)}$.

    %\item If there is an induced matching of size $r$ in $G$, then there is a solution $\x \in \R_{+}^n$ that satisfies $r$ constraints in $\aset$.

    \item (\yi) There is a solution $\x \in \R_{+}^n$ that satisfies at least $\induce{G}$ constraints in $\aset$.

    %\item For some constant $\beta>0$, any solution $\x \in \R_{+}^n$ that $\beta$-satisfies \danuponb{Define!} $r$ constraints (i.e., $\ell_i \leq a_i^T\x \leq \beta\mu_i$) in $\aset$ can be turned into a $\sigma'$-semi-induced matching of size $r$ for some $\sigma'$.\bundit{I changed the wording.}

    \item (\ni) There is no solution $\x \in \R_{+}^n$ that ``$\beta$-satisfies'' more than $\sinduce{G}$ constraints in $\aset$; i.e., $|\set{i:\ell_i \leq a_i^T\x \leq \beta\mu_i}| \leq \sinduce{G}$ for all $\x$.
\squishend
\end{theorem}
Now, we prove the hardness of approximating \MFS, which holds even in the following {\em bi-criteria} setting. For any instance $\aset$, we denote by $\opt(\aset)$ the maximum number of constraints that can be satisfied, i.e., $\opt(\aset) = \max_i |\set{i: \ell_i \leq a_i^T {\bf x} \leq \mu_i}|$. For any $0<\alpha, \beta \leq m$, we say that an algorithm is an $(\alpha, \beta)$-approximation algorithm if, for any instance $\aset$ of \MFS, the algorithm returns a solution ${\bf x}$ that $\beta$-satisfies at least $\opt(\aset)/\alpha$ constraints; i.e., $|\set{i: \ell_i \leq a_i^T\x \leq \beta\mu_i}| \geq \opt(\aset)/\alpha$. The non-bi-criteria setting (defined in Section~\ref{sec:intro approx}) is when $\beta=1$.

\begin{corollary} Let $\epsilon>0$ be any constant. There is no polynomial-time $(m^{1-\epsilon}, m+n)$-approximation algorithm for \MFS unless $\NP\subseteq \ZPP$. Moreover, when considering an approximation factor in terms of $L$, finding $(\log^{1-\epsilon} L, m+n)$-approximation algorithm cannot be done in polynomial time, unless $\NP \subseteq \ZPP$.
%For any $\alpha \leq m^{1-\epsilon}$ and $\beta \leq \min\set{m,n}$, there is no polynomial-time $(\alpha,\beta)$-approximation algorithm for \MFS, unless $\NP\subseteq \ZPP$.
\end{corollary}
\begin{proof}
We start from the graph $G$ given by Theorem~\ref{thm:semi-induced} and invoke Theorem~\ref{thm: elbassioni MRFS} on $G$. For \yi where $\induce{G}\geq |V(G)|^{1-\epsilon}$, we have that there is a solution ${\bf x}$ that satisfies $\induce{G}\geq |V(G)|^{1-\epsilon}=m^{1-\epsilon}$ constraints.
%
%, we have an induced matching in $G$ of size at least $m^{1-\epsilon}$, so there is a solution ${\bf x}$ that satisfies $m^{1-\epsilon}$ constraints.
%
In the \ni where $\sinduce{G}\leq |V(G)|^{\epsilon}$, there is no solution that $\beta$-satisfies $\sinduce{G}\leq |V(G)|^{\epsilon}\leq m^{\epsilon}$, for any $0<\beta\leq |V(G)|=m+n$. Thus, even when we allow the solution to $(m+n)$-satisfies the constraints, there is still an $m^{1-O(\epsilon)}$ gap.

Theorem~\ref{thm: elbassioni MRFS} guarantees that $L \leq 2^{O(n \log n)} \leq 2^{n^{1+\epsilon}}$, so the hardness of $n^{1-O(\epsilon)}$ can also be written as $\log^{1-O(\epsilon)} L$.
% if some solution ${\bf x}$ $\beta$-satisfies at least $m^{\epsilon}$ constraints, then we must have, for some permutation $\sigma'$, a $\sigma'$-semi-induced matching in $G$, which is impossible.
\end{proof}

We remark that the hardness factor can be improved to $m/2^{\log^{3/4+\gamma} m}$. Our bounds are nearly tight since it is trivial to get a factor of $m$-approximation, and since~\cite{ElbassioniRRS09} showed an $O(\log (nL))$-approximation algorithm.

\subsection{Pricing Problems.}

In this section, we revisit \UDP and \SMP and give an alternative proof of the hardness
results in~\cite{ChalermsookPricing}. As discussed in the introduction, our proof illustrates the insight that the maximum expanding sequence problem, which is equivalent to the bipartite semi-induced matching problem (see \fullversion{Appendix~\ref{sec:equivalence semi induced and expanding seq}}{the full version}), is the main source of hardness of these pricing problems. 

We start by defining the pricing problems we consider. In Unit-Demand Min-Buying Pricing (\UDP), we have a collection of items $\iset = [n]$ and customers $\cset$ where each consumer $c \in \cset$ is associated with set $S_c \subseteq [n]$ and a budget $B_c$.
Once the price $p:\iset \rightarrow\R_{+}$ is fixed, each consumer $c$
buys the cheapest item in $S_c$ if the price of such item is at most
$B_c$; otherwise, the consumer buys nothing.
Our goal is to set the price $p$ so that the profit is maximized.

In Single-Minded Pricing (\SMP), the setting is the same except that now each consumer $c$ buys the whole set $S_c$ of its items if $\sum_{i \in S_c} p(i) \leq B_c$; otherwise, the consumer $c$ buys nothing. Again, the goal is to set the prices $p$ so that the profit is maximized.

For any instance $\pset$ of \UDP or \SMP, we denote by $\opt(\pset)$ the revenue that can be obtained by an optimal price function.

%Our starting point of both hardness results is implicit in Theorem~\ref{thm: bounded degree induced matching}, restated in the following form.

Our contribution lies in proving the following theorem that makes connections between the bipartite semi-induced matching problem and pricing problems. The proof of this theorem borrows many ideas from~\cite{BriestK11,ChalermsookPricing}. \fullversion{The proof is included in Appendix~\ref{sec: omitted proofs from rest}.}{The proof is deferred to the full version.}

\begin{comment}
\begin{theorem}
\label{thm: semi-induced matching to UDP}
For any graph $G$ and $H$, there is a reduction with a running time of
$2^{O( |V(H)|)}\poly|V(G)|$, that transforms a graph
$G'$ of the form $G'=\bipp[G \vee H]$ into an instance $(\cset, \iset)$ of \UDP such that

\[\induce{G'} \leq \opt(\cset, \iset) \leq 2 \sinduce{G'} + O(|V(G)| |E(H)|)\]

Moreover, the number of consumers and items are
$|\cset| = 2^{O(|V(H)|)}|V(G)|$, and
$|\iset| = |V(G)||V(H)|$ respectively.
\end{theorem}

\begin{theorem}
\label{thm: semi-induced matching to SMP}
There is a reduction with a running time of $|V(G)|^{O(|V(H)|)})$,
that transforms graph $G'= \bipp[G \vee H]$ into an
instance $(\cset, \iset)$ of \SMP such that

\[\induce{G'} \leq \opt(\cset, \iset) \leq 2 \sinduce{G'} + O(|V(G)| |E(H)|)\]

Furthermore, the number of consumers and items are
$|\cset| = |V(H)|^{O(|V(H)|)}|V(G)|$ and $|\iset|  = |V(G)||V(H)|$, respectively.
\end{theorem}

\danuponb{I propose to combine two theorems as follows.}
%

\end{comment}

\begin{theorem}
\label{thm: semi-induced matching to UDP}
There are reductions with a running time of $|V(G)|^{O(|V(H)|)}$
that transform input graph $G'= \bipp[G \vee H]$ into an
instance $(\cset, \iset)$ of \UDP or \SMP such that

\[\induce{G'} \leq \opt(\cset, \iset) \leq 2 \sinduce{G'} + O(|V(G)|(1+ |E(H)|))\]

Furthermore, the number of consumers and items are
$|\cset| = |V(H)|^{O(|V(H)|)}|V(G)|$ and $|\iset|  = |V(G)||V(H)|$, respectively, and each consumer $c \in \cset$ satisfies $|S_c| \leq O(\Delta^2)$.
\end{theorem}

Note that the running time and the number of consumers for \UDP can be slightly improved with a more careful analysis as follows: The running time can be made $2^{O(|V(H)|)}\poly |V(G)|$, and the number of consumers is $2^{O(|V(H)|)}|V(G)|$.

Applying the subadditivity property (Theorems~\ref{theorem:subadditive}), the hardness of the induced and semi-induced matching problems (Lemma~\ref{lemma: semi-induced matching in terms of graph products}) and  Theorem \ref{thm: semi-induced matching to UDP}, we get the following result, which is an alternative proof of the result in \cite{ChalermsookPricing}.

\begin{theorem}
\label{thm: UDP and SMP hardness}
For any constant $\epsilon >0$, both \SMP and \UDP are hard to approximate
to within a factor of $\log^{1-\epsilon} |\cset|$ and $\paren{\max_{c \in \cset} |S_c|}^{1/2-\epsilon}$, where $\cset$ is the set of consumers, unless $\NP
\subseteq {\sf ZPTIME}(n^{O(\poly \log n)})$.
\end{theorem}

\begin{proof}
% We start by invoking Lemma~\ref{lemma: semi-induced matching in terms of graph products} to obtain the graphs $G$, $H$, and $G' = \bipp[G \vee H]$ such that it is hard to distinguish between when $\induce{G'}\geq |V(G)|\Delta^{1-\epsilon}$ and $\sinduce{G'}\leq |V(G)|\Delta^{\epsilon}$, for some parameter $\Delta$ whose value will be specified later. Recall that we use $n$ to denote the size of 3SAT instance, which is the starting point of the reduction, and that $|V(G)| =n^{O(\log \Delta)}$.

First, we take a $\Delta(q)$-regular graph $G$ on $O(q^{O(\log\Delta(q))})$ vertices and an empty graph $H$ on $\Delta(q)$ vertices as stated in Lemma~\ref{lemma: semi-induced matching in terms of graph products}.
Thus, assuming that  $\NP \nsubseteq {\sf ZPTIME}(n^{O(\log \Delta(n))})$, there is no polynomial-time algorithm that distinguishes between the case that $\induce{\bipp[G \vee H]}\geq |V(G)|(\Delta(q))^{1-\epsilon}$ and the case that $\sinduce{\bipp[G \vee H]}\leq |V(G)|(\Delta(q))^{\epsilon}$, for function $\Delta$ whose value will be specified later.
% Recall that we use $n$ to denote the size of 3SAT instance, which is the starting point of the reduction, and that $|V(G)| =n^{O(\log \Delta)}$.

%Then we apply Theorem~\ref{thm: semi-induced matching to UDP} on $\bipp[G \vee H]$ to obtain an instance $(\cset, \iset)$ of \UDP. This means that in the \yi, the optimal revenue from $(\cset, \iset)$ is at least $\induce{G'} \geq |V(G)| (\Delta(q))^{1-\epsilon}$. Additionally, the optimal revenue in the \ni is at most $2\sinduce{G'}  + O(|V(G)||E(H)|)\leq 2|V(G)|(\Delta(q))^{\epsilon}$.\danuponb{I put comments many times already. I still don't understand this. For example, do you actually mean $3|V(G)|(\Delta(q))^{\epsilon}$?} In the \ni, the first term is at most $2|V(G)|(\Delta(q))^{\epsilon}$, and since $H$ is an empty graph, the second term is at most $O(|V(G)|) \leq |V(G)|\Delta^{\epsilon}$. This implies the gap of $(\Delta(q))^{1-O(\epsilon)}$ between the two cases.

Then we apply Theorem~\ref{thm: semi-induced matching to UDP} on $\bipp[G \vee H]$ to obtain an instance $(\cset, \iset)$ of \UDP. This means that in the \yi, the optimal revenue from $(\cset, \iset)$ is at least $\induce{G'} \geq |V(G)| (\Delta(q))^{1-\epsilon}$. Additionally, the optimal revenue in the \ni is at most $2\sinduce{G'}  + O(|V(G)|(1+|E(H)|))$, which is $O(|V(G)|(\Delta(q))^{\epsilon})$ because $\sinduce{G'}\leq |V(G)|(\Delta(q))^{\epsilon}$ (by Lemma~\ref{lemma: semi-induced matching in terms of graph products}), and the term $|V(G)|(1+|E(H)|)$ is at most $O(|V(G)|) \leq |V(G)|\Delta^{\epsilon}$ since $H$ is an empty graph. This implies the gap of $(\Delta(q))^{1-O(\epsilon)}$ between the two cases.

Now, we choose $\Delta(q) = \log^b q$ where $b =O(1/\epsilon)$. So, $|V(H)| = \Delta(q) = \log^b q$ and $|V(G)| = q^{O(\log \log q)}$. By Theorem~\ref{thm: semi-induced matching to UDP}, the number of consumers is bounded by $|\cset| \leq 2^{O(|V(H)|)} |V(G)| \leq 2^{\log^{b+1} q} = 2^{(\Delta(q))^{1+\epsilon}}$, so we obtain the gap of $(\Delta(q))^{1-O(\epsilon)}\geq \log^{1- O(\epsilon)} |\cset|$ as desired.
% The running time of the reduction is $2^{O(\Delta)} \poly |V(G)| \leq n^{\poly \log n}$, where $n$ is the size of 3SAT formula from Lemma~\ref{lemma: semi-induced matching in terms of graph products}. \danuponb{I added ``from ...'' but it's still confusing.}

Lemma~\ref{lemma: semi-induced matching in terms of graph products} holds with the assumption that $\NP \nsubseteq {\sf ZPTIME}(n^{O(\log \Delta(n))})={\sf ZPTIME}(n^{O(\log \log n)})$, and the running time of the reduction in Theorem~\ref{thm: semi-induced matching to UDP} is $|V(G)|^{O(|V(H)|)} = q^{\log^b q} = q^{\poly \log q}$, so the hardness assumption we need is $\NP \not\subseteq {\sf ZPTIME}(n^{\poly \log n})$.
\fullversion{

Now to get the hardness in terms of $k^{1/2-\epsilon} = \paren{\max_{c \in \cset} |S_c|}^{1/2 -\epsilon}$, notice that $k \leq O(\Delta(q)^2)$. Therefore, the hardness in terms of $k$ is $(\Delta(q))^{1-O(\epsilon)} = k^{1/2 - O(\epsilon)}$. This holds for any $\Delta(q) \leq \poly \log q$.
}{}

\end{proof}

\subsection{Donation Center Location Problem.}
\label{sec:donation}

The Donation Center Location {\sc (Dcl)} problem is defined as
follows. The input consists of a directed bipartite graph $G=(A \cup L, E)$ with
edges directed from the set $A$ of {\em agents} to the set $L$ of {\em donation
centers}. Each center $\ell \in L$ has a capacity $c_{\ell} \in \Z$ that represents the maximum number of clients that can be served, and each
vertex $a \in A$ has a strictly ordered preference ranking of its
neighbor in $L$. We are interested in choosing a subset $L' \subseteq L$ of
centers to open and an assignment of subset $A' \subseteq A$ of agents to
centers such that: (1)~The number of agents assigned to any center $\ell$
is at most $c_{\ell}$, and (2)~Each $a \in A'$ is assigned to its
highest-ranked neighbor in $L'$. Our goal is to maximize the number of
satisfied agents.

We therefore write an instance of {\sc Dcl} as a triple $\pset = (G=(A \cup L, E), \set{c_{\ell}}_{{\ell} \in L}, \set{\preceq_a}_{a \in A})$, where the relation $\preceq_a$ represents the rank relation of agent $a \in A$. Denote by $\opt(\pset)$ the optimal value of the instance $\pset$. The following theorem makes a connection between {\sc Dcl} and the semi-induced matching problem.

\begin{theorem}
\label{thm: semi-induced to DCL}
Let $G'$ be a bipartite graph. There is a polynomial time reduction that transforms $G'$ into an instance $\pset=(G,\set{c_{\ell}}, \set{\preceq_a})$ of {\sc Dcl} such that:
\[\induce{G'} \leq \opt(\pset) \leq \sinduce{G'} \]
\begin{comment}
\squishlist
  \item (\yi) If there is an induced matching of size $K$ in $\bipp[G]$, then there is a solution to {\sc Dcl} that satisfies $K$ agents.

  \item (\ni) If there is a solution to {\sc Dcl} that satisfies $K$ agents, then there is a $\sigma$-semi-induced matching of size $K$ for some ordering $\sigma$.
\squishend
\end{comment}
Moreover, $|V(G)| = |V(G')|$, $c_{\ell} = 1$ for all $\ell \in L$, and $\preceq_a = \preceq^*$ for all $a \in A$, where $\preceq^*$ is some global preference.
\end{theorem}

\begin{proof}
Given a bipartite graph $G'= (V_1\cup V_2, E)$ where $V_1 = \set{(u,1): u \in [n]}$ and $V_2 = \set{(u,2): u\in [n]}$, we create an instance of {\sc Dcl} as follows. Each vertex $(u,1)$ represents a center $\ell(u)$, and each vertex $(v,2)$ represents an agent $a(v)$. The capacity of each center $\ell(u)$ is $c_{\ell(u)} = 1$, and each agent uses a global preference $\preceq^*$ that satisfies $\ell(u) \preceq^* \ell(u')$ if and only if $u >u'$ are integers.

First, we prove the lower bound of $\opt(\pset) \geq \induce{G'}$.
Given any induced matching $\mset$ in $G'$, we argue that there is a solution of value $|\mset|$ to the {\sc Dcl} instance. For each edge $(u,1)(v,2) \in \mset$, we open the center $\ell(u)$ and assign the agent $a(v)$ to $\ell(u)$. Now,  we only need to argue that all agents that are matched in the matching $\mset$ are satisfied. For each $(v,2)$, where $(u,1)(v,2) \in \mset$, assume that $(v,2)$ prefers some other center $\ell(u')$ to $\ell(u)$ that is currently open. This means that there must be an edge $(u',1)(v,2) \in E$, contradicting the fact that $\mset$ is an induced matching.

To prove the upper bound, given any solution $L' \subseteq L, A' \subseteq A$ and assignment $\phi: A' \rightarrow L'$, we show that we can construct a $\sigma$-semi-induced matching $\mset$ in $G'$ such that $|\mset| = |A'|$. Consider the following set:
\[\mset = \set{(u,1)(v,2): \ell(u) \in L' \mbox{ and } \phi(a(v)) = \ell(u)}\]
The set $\mset$ is indeed a matching because each agent is only assigned once, and each center has unit capacity. It is sufficient to show that the matching $\mset$ is $\sigma$-semi-induced for total order $\sigma$ defined by $\sigma(u) < \sigma(u')$ if and only if $u <u'$. Assume that this is not a $\sigma$-semi-induced matching. Then there must be an edge $(u',1)(v,2)$, but $\phi(a(v)) = \ell(u)$ and $u > u'$. This means that the agent $a(v)$ prefers $\ell(u')$ to $\ell(u)$, but $a(v)$ was assigned to $\ell(u)$ instead. This contradicts the fact that the solution is feasible.

\end{proof}

\begin{corollary}
Let $\epsilon >0$ be any constant. Unless $\NP \subseteq \ZPP$, it is hard to approximate {\sc Dcl} to within a factor of $n^{1-\epsilon}$ where $n$ is the number of vertices in the input graph.
\end{corollary}

\begin{proof}
First, we take a graph $G'$ as in Theorem~\ref{thm:semi-induced}, and
we then invoke Theorem~\ref{thm: semi-induced to DCL} on $G'$ to obtain an instance $\pset = (G =(A \cup L, E), \set{c_l}, \set{\preceq_a})$ of {\sc DCL}. In the \yi, where $\induce{G'} \geq |V(G')|^{1-\epsilon}$, we have that $\opt(\pset) \geq |V(G)|{1-\epsilon}$. In the \ni, where $\sinduce{G'} \leq |V(G')|^{\epsilon}$, we have $\opt(\pset) \leq |V(G)|^{\epsilon}$. Therefore, we obtain a gap of $|V(G)|^{1-2\epsilon}$ as desired.

\end{proof}

\subsection{Boxicity, Cubicity, and Threshold Dimension.}\label{sec:boxicity}

%%% BUN: I put the long one back
\begin{comment}
Adiga et~al.~\cite{AdigaBC10} showed that the problems of approximating the boxicity, cubicity and threshold dimension of graphs are at least as hard as approximating the poset dimension number (within a logarithmic factor). We thus get a tight hardness of $n^{1-\epsilon}$ for all these problems by combining the reductions in \cite{AdigaBC10} with our hardness of poset dimension. We provide some details in \fullversion{Appendix~\ref{sec:boxicity-appendix}}{the full version} for completeness.

%Since the hardness of Boxicity, Cubicity, and Threshold Dimension directly follows from the result of Adiga et~al.~\cite{AdigaBC10}, and due to space limitation, we only briefly discuss the result here. We provide definitions and an outline of the proofs in \fullversion{Appendix~\ref{sec:boxicity-appendix}}{the full version} for completeness.

%Adiga et~al. \cite{AdigaBC10} shows that there is a polynomial-time algorithm that transforms any poset $\vec{P}$ into a graph $G_{\vec P}$ such that ${\sf box}(G_{\vec P}) \leq \dimension{\vec{P}} \leq 2 {\sf box}(G_{\vec P})$. This implies the hardness of approximating boxicity of graphs. Since cubicity is known to be within a logarithmic factor of boxicity, approximating cubicity is also as hard as boxicity (up to a factor of $\log n$)~\cite{AdigaBC10}. Also, \cite{AdigaBC10} constructs graph $G'_{\vec P}$ such that the threshold dimension of $G'_{\vec P}$ is the same as the dimension of poset $\vec{P}$, hence implying the hardness of approximating threshold dimension.
\end{comment}

We start by giving the definitions of the problems and related notions in graph theory.

\begin{definition}[Intersection Graph]
Given a graph $G=(V,E)$, we say that a set system $\set{S_v}_{v \in V(G)} $ is a {\em set system representation} of $G$ if
\[(\forall u,v \in V(G)) uv \in E  \mbox{ iff } S_u \cap S_v \neq \emptyset\]
\end{definition}

It is well-known that any graph $G$ can be represented by a set system: for each vertex $u \in V(G)$, we define set $S_u$ to contain edges incident to $u$. We are interested in a set system representation where each set corresponds to a geometric object.

\begin{definition}[Boxicity and Cubicity]
We say that the {\em boxicity} of a graph $G$ is at most $d$ if $G$ can be represented by a set system $\set{S_v}_{v \in V(G)}$ such that each set $S_v$ is a $d$-dimensional axis-parallel hyper-rectangle in $\R^d$. Similarly, we say that the {\em cubicity} of $G$ is at most $d$ if each set $S_v$ is a unit cube in $\R^d$.
\end{definition}

In other words, the boxicity of $G$, denoted by ${\sf box}(G)$, is the minimum dimension $d$ such that we can represent each node $v \in V(G)$ as a $d$-dimensional rectangle in the geometric setting. It is known that the boxicity is at most one and two in interval graphs and planar graphs, respectively.

\begin{definition}[Threshold Dimension]
A graph $G$ is a {\em threshold graph} if there is a real number $\eta$ and weight function $w: V(G) \rightarrow \R$ such that $uv \in E(G)\Leftrightarrow w(u) + w(v) \geq \eta$. For any graph $G$, the {\em threshold dimension} of $G$ is the minimum $k$ such that there exist threshold graphs $G_1,\ldots, G_k$ where $E(G) = \bigcup_{i=1}^k E(G_i)$.
\end{definition}

Adiga et~al.~\cite{AdigaBC10} show that the problems of approximating boxicity, cubicity, and threshold dimension are at least as hard as poset dimension (within a constant factor). We get the tight hardness of these problems by combining the reductions in \cite{AdigaBC10} with our hardness of poset dimension. We provide an outline of their reductions here and the readers to~\cite{AdigaBC10} for the complete proof.

First, they show that there is a polynomial-time algorithm that transforms any poset $\vec{P}$ into a graph $G_{\vec P}$ such that ${\sf box}(G_{\vec P}) \leq \dimension{\vec{P}} \leq 2 {\sf box}(G_{\vec P})$. This implies the hardness of approximating boxicity of graphs. Since cubicity is known to be within a logarithmic factor of boxicity, approximating cubicity is also as hard as boxicity (up to a factor of $\log n$). Also, they construct graph $G'_{\vec P}$ such that the threshold dimension of $G'_{\vec P}$ is the same as the dimension of poset $\vec{P}$, hence implying the hardness of approximating threshold dimension.

 % Sec 6: More Applications

\section{Conclusion and Open Problems}

We have shown that simple techniques based on graph products are powerful tools in proving hardness of approximation. While some of these results are tight, some problems are still open. In particular, it remains to close the gap between $d^{1/2 -\epsilon}$ and $O(d)$ for the semi-induced matching problem on $d$-regular graphs. Also, there is a gap between $O(k)$-approximation (see~\cite{Balcan-Blum-Pricing}) and $k^{1/2-\epsilon}$-hardness for the $k$-hypergraph vertex pricing problem.
% (see~\cite{Balcan-Blum-Pricing} for an algorithm).

It is also interesting to further investigate the power of our techniques in proving hardness of approximation or even in other types of lower bounds. A potential starting point is to look at problems which share common structures to those studied in this paper. For example, \UDP is a special case of the multi-user Stackelberg network pricing problem, so graph products can be used to prove the hardness of this problem. However, the approximability of the single-user version is still wide open as there is a large gap between $(2-\epsilon)$-hardness and $O(\log n)$-approximation algorithm. In fact, the proof of the $(2-\epsilon)$-hardness can be viewed as a reduction from the independent set problem, but we found no graph product techniques that can be used to boost the hardness further. (See, e.g., \cite{BriestCKLN10,BriestHK12} for detail.) %Can graph product techniques be extended to get an improved result for this problem?

% \section{Open problems}
% Approximate treewidth, Incidence poset dimension, Strong edge coloring.

%\pagebreak

%{\footnotesize
\bibliographystyle{plain}
\bibliography{ind-matching}

\begin{thebibliography}{10}

\bibitem{AdigaBC10}
Abhijin Adiga, Diptendu Bhowmick, and L.~Sunil Chandran.
\newblock The hardness of approximating the boxicity, cubicity and threshold
  dimension of a graph.
\newblock {\em Discrete Appl. Math.}, 158(16):1719--1726, 2010.

\bibitem{AdigaBC11}
Abhijin Adiga, Diptendu Bhowmick, and L.~Sunil Chandran.
\newblock Boxicity and poset dimension.
\newblock {\em SIAM J. Discrete Math.}, 25(4):1687--1698, 2011.

\bibitem{AdigaCM11}
Abhijin Adiga, L.~Sunil Chandran, and Rogers Mathew.
\newblock Cubicity, degeneracy, and crossing number.
\newblock In {\em FSTTCS}, pages 176--190, 2011.

\bibitem{AmbuhlMMS08}
Christoph Amb{\"u}hl, Monaldo Mastrolilli, Nikolaus Mutsanas, and Ola Svensson.
\newblock Precedence constraint scheduling and connections to dimension theory
  of partial orders.
\newblock {\em Bulletin of the EATCS}, 95:37--58, 2008.

\bibitem{AndrewsCGKTZ10}
Matthew Andrews, Julia Chuzhoy, Venkatesan Guruswami, Sanjeev Khanna, Kunal
  Talwar, and Lisa Zhang.
\newblock Inapproximability of edge-disjoint paths and low congestion routing
  on undirected graphs.
\newblock {\em Combinatorica}, 30(5):485--520, 2010.

\bibitem{Balcan-Blum-Pricing}
Maria-Florina Balcan and Avrim Blum.
\newblock Approximation algorithms and online mechanisms for item pricing.
\newblock {\em Theor. Comput.}, 3(1):179--195, 2007.

\bibitem{BonifaciKMS11}
Vincenzo Bonifaci, Peter Korteweg, Alberto Marchetti-Spaccamela, and Leen
  Stougie.
\newblock Minimizing flow time in the wireless gathering problem.
\newblock {\em ACM Transactions on Algorithms}, 7(3):33, 2011.

\bibitem{BriestCKLN10}
Patrick Briest, Parinya Chalermsook, Sanjeev Khanna, Bundit Laekhanukit, and
  Danupon Nanongkai.
\newblock Improved hardness of approximation for stackelberg shortest-path
  pricing.
\newblock In {\em WINE}, pages 444--454, 2010.

\bibitem{BriestHK12}
Patrick Briest, Martin Hoefer, and Piotr Krysta.
\newblock Stackelberg network pricing games.
\newblock {\em Algorithmica}, 62(3-4):733--753, 2012.

\bibitem{BriestK11}
Patrick Briest and Piotr Krysta.
\newblock Buying cheap is expensive: Approximability of combinatorial pricing
  problems.
\newblock {\em SIAM J. Comput.}, 40(6):1554--1586, 2011.

\bibitem{Cameron89}
Kathie Cameron.
\newblock Induced matchings.
\newblock {\em Discrete Appl. Math.}, 24(1-3):97--102, 1989.

\bibitem{CameronH06}
Kathie Cameron and Pavol Hell.
\newblock Independent packings in structured graphs.
\newblock {\em Math. Program.}, 105(2-3):201--213, 2006.

\bibitem{ChalermsookPricing}
Parinya Chalermsook, Julia Chuzhoy, Sampath Kannan, and Sanjeev Khanna.
\newblock Improved hardness results for profit maximization pricing problems
  with unlimited supply.
\newblock In {\em APPROX-RANDOM}, pages 73--84, 2012.

\bibitem{ChandranFS08}
L.~Sunil Chandran, Mathew~C. Francis, and Naveen Sivadasan.
\newblock Boxicity and maximum degree.
\newblock {\em J. Comb. Theory, Ser. B}, 98(2):443--445, 2008.

\bibitem{ChandranS07}
L.~Sunil Chandran and Naveen Sivadasan.
\newblock Boxicity and treewidth.
\newblock {\em J. Comb. Theory, Ser. B}, 97(5):733--744, 2007.

\bibitem{DuckworthMZ05}
William Duckworth, David Manlove, and Michele Zito.
\newblock On the approximability of the maximum induced matching problem.
\newblock {\em J. Discrete Algorithms}, 3(1):79--91, 2005.

\bibitem{ElbassioniRRS09}
Khaled~M. Elbassioni, Rajiv Raman, Saurabh Ray, and Ren{\'e} Sitters.
\newblock On the approximability of the maximum feasible subsystem problem with
  0/1-coefficients.
\newblock In {\em SODA}, pages 1210--1219, 2009.

\bibitem{EvenGMT84}
S.~Even, O.~Goldreich, S.~Moran, and P.~Tong.
\newblock On the np-completeness of certain network testing problems.
\newblock {\em Networks}, 14(1):1--24, 1984.

\bibitem{FeigeK98}
Uriel Feige and Joe Kilian.
\newblock Zero knowledge and the chromatic number.
\newblock {\em J. Comput. Syst. Sci.}, 57(2):187--199, 1998.

\bibitem{FelsnerLT10}
Stefan Felsner, Ching~Man Li, and William~T. Trotter.
\newblock Adjacency posets of planar graphs.
\newblock {\em Discrete Math.}, 310(5):1097--1104, 2010.

\bibitem{RobertsBoxicity}
F.S.Roberts.
\newblock On the boxicity and cubicity of a graph.
\newblock {\em Recent Progresses in Combinatorics}, pages 301--310, 1969.

\bibitem{GaoZ96}
Guogang Gao and Xuding Zhu.
\newblock Star-extremal graphs and the lexicographic product.
\newblock {\em Discrete Math.}, 152(1-3):147--156, 1996.

\bibitem{GareyJ79}
M.~R. Garey and David~S. Johnson.
\newblock {\em Computers and Intractability: A Guide to the Theory of
  NP-Completeness}.
\newblock W. H. Freeman, 1979.

\bibitem{GotthilfL05}
Zvi Gotthilf and Moshe Lewenstein.
\newblock Tighter approximations for maximum induced matchings in regular
  graphs.
\newblock In {\em WAOA}, pages 270--281, 2005.

\bibitem{GuruswamiPricing}
Venkatesan Guruswami, Jason~D. Hartline, Anna~R. Karlin, David Kempe, Claire
  Kenyon, and Frank McSherry.
\newblock On profit-maximizing envy-free pricing.
\newblock In {\em SODA}, pages 1164--1173. SIAM, 2005.

\bibitem{GuruswamiR09}
Venkatesan Guruswami and Prasad Raghavendra.
\newblock Hardness of solving sparse overdetermined linear systems: A 3-query
  pcp over integers.
\newblock {\em TOCT}, 1(2), 2009.

\bibitem{HammackIK11}
Richard Hammack, Wilfried Imrich, and Sandi Klav{\v{z}}ar.
\newblock {\em Handbook of product graphs}.
\newblock Discrete Math. Appl. (Boca Raton). CRC Press, Boca Raton, FL, second
  edition, 2011.

\bibitem{Hastad96}
Johan H{\aa}stad.
\newblock Clique is hard to approximate within n$^{1-\epsilon}$.
\newblock In {\em FOCS}, pages 627--636, 1996.

\bibitem{HJ07}
Rajneesh Hegde and Kamal Jain.
\newblock The hardness of approximating poset dimension.
\newblock {\em Electron. Notes Discrete Math.}, 29:435--443, 2007.

\bibitem{HuangS09}
Chien-Chung Huang and Zoya Svitkina.
\newblock Donation center location problem.
\newblock In {\em FSTTCS}, pages 227--238, 2009.

\bibitem{Johnson81}
David~S. Johnson.
\newblock The {NP}-completeness column: An ongoing guide.
\newblock {\em J. Algorithms}, 2(4):393--405, 1981.

\bibitem{JooSSM10}
Changhee Joo, Gaurav Sharma, Ness~B. Shroff, and Ravi~R. Mazumdar.
\newblock On the complexity of scheduling in wireless networks.
\newblock {\em EURASIP J. Wireless Comm. and Networking}, 2010, 2010.

\bibitem{KhotPonnuswami}
Subhash Khot and Ashok~Kumar Ponnuswami.
\newblock Better inapproximability results for maxclique, chromatic number and
  min-3lin-deletion.
\newblock In {\em ICALP (1)}, pages 226--237, 2006.

\bibitem{KlavzarH02}
Sandi Klavzar and Hong-Gwa Yeh.
\newblock On the fractional chromatic number, the chromatic number, and graph
  products.
\newblock {\em Discrete Math.}, 247(1-3):235--242, March 2002.

\bibitem{KumarMS04}
Ravi Kumar, Uma Mahadevan, and D.~Sivakumar.
\newblock A graph-theoretic approach to extract storylines from search results.
\newblock In {\em KDD}, pages 216--225, 2004.

\bibitem{LawlerV81}
Eugene~L. Lawler and Oliver Vornberger.
\newblock The partial order dimension problem is {NP}-complete.
\newblock Manuscript, 1981.

\bibitem{LinialV89}
Nathan Linial and Umesh~V. Vazirani.
\newblock Graph products and chromatic numbers.
\newblock In {\em FOCS}, pages 124--128, 1989.

\bibitem{Milosavljevic11}
Nikola Milosavljevic.
\newblock On complexity of wireless gathering problems on unit-disk graphs.
\newblock In {\em ADHOC-NOW}, pages 308--321, 2011.

\bibitem{Raskhodnikova10}
Sofya Raskhodnikova.
\newblock Transitive-closure spanners: A survey.
\newblock In {\em Property Testing}, pages 167--196, 2010.

\bibitem{Rusmevichientong03}
Paat Rusmevichientong.
\newblock A non-parametric approach to multi-product pricing: Theory and
  application.
\newblock {\em Ph. D. thesis, Stanford University}, 2003.

\bibitem{Rusetal}
Paat Rusmevichientong, Benjamin Van~Roy, and Peter~W. Glynn.
\newblock A nonparametric approach to multiproduct pricing.
\newblock {\em Oper. Res.}, 54:82--98, January 2006.

\bibitem{Scheinerman1997fractional}
E.R. Scheinerman and D.H. Ullman.
\newblock {\em Fractional graph theory: a rational approach to the theory of
  graphs}.
\newblock Wiley-Intersci. Ser. Discrete Math. Optim. Wiley, 1997.

\bibitem{Schnyder89}
Walter Schnyder.
\newblock {Planar graphs and poset dimension}.
\newblock {\em Order}, 5:323--343, 1989.

\bibitem{Schnyder90}
Walter Schnyder.
\newblock Embedding planar graphs on the grid.
\newblock In {\em SODA}, pages 138--148, 1990.

\bibitem{StockmeyerV82}
Larry~J. Stockmeyer and Vijay~V. Vazirani.
\newblock Np-completeness of some generalizations of the maximum matching
  problem.
\newblock {\em Inf. Process. Lett.}, 15(1):14--19, 1982.

\bibitem{Trevisan01}
Luca Trevisan.
\newblock Non-approximability results for optimization problems on bounded
  degree instances.
\newblock In {\em STOC}, pages 453--461, 2001.

\bibitem{TrevisanLecture}
Luca Trevisan.
\newblock CS294: PCP and Hardness of Approximation, Lecture 5, 2006.

\bibitem{TrotterBook01}
W.T. Trotter.
\newblock {\em Combinatorics and Partially Ordered Sets: Dimension Theory}.
\newblock Johns Hopkins Studies in the Mathematical Sciences. Johns Hopkins
  University Press, 2001.

\bibitem{Yannakakis82}
Mihalis Yannakakis.
\newblock {The Complexity of the Partial Order Dimension Problem}.
\newblock {\em SIAM J. Algebra Discr.}, 3(3):351--358, 1982.

\bibitem{Zito99}
Michele Zito.
\newblock Induced matchings in regular graphs and trees.
\newblock In {\em WG}, pages 89--100, 1999.

\end{thebibliography}
%}

%\fullversion{

\pagebreak
\appendix
\section*{Appendix}

\section{Omitted Proofs from Section~\ref{sec:rest}}
\label{sec: omitted proofs from rest}

\subsection{Proof of Theorem~\ref{thm: elbassioni MRFS}}

\danupon{The statement and proof is changed in many places. Please check carefully (especially the use of $m$ and $n$). If this is correct then we have to change the statement in the main paper and maybe other places as well.}

\restatethm{Theorem}{\ref{thm: elbassioni MRFS}}
{
Consider an instance $G=(V_1\cup V_2, E)$ of the bipartite semi-induced matching problem. Let $m=|V_1|$ and $n=|V_2|$. There is a polynomial time reduction that, for any $\beta>0$, outputs an instance $\aset= (A,\ell, \mu)$ of {\sc Mrfs} satisfying the following properties:
\squishlist
    \item ({\sc Size}) Matrix $A$ is an $m$-by-$n$ matrix and  $L = \max_{i\in [m] } \set{\ell_i} = (\beta n)^{O(m)}$.

    \item (\yi) There is a solution $\x \in \R_{+}^n$ that satisfies at least $\induce{G}$ constraints in $\aset$.

    \item (\ni) There is no solution $\x \in \R_{+}^n$ that ``$\beta$-satisfies'' more than $\sinduce{G}$ constraints in $\aset$; i.e., $|\set{i:\ell_i \leq a_i^T\x \leq \beta\mu_i}| \leq \sinduce{G}$ for all $\x$.
\squishend
}

\begin{proof}
%\danupon{I don't understand why this paragraph is here.}
%For the sake of presentation, we invoke $\sigma$ to map vertices of $G$ into a simpler form. To be precise, we write the set of vertices of $G$ as $V_1=\set{(u,1):u\in[n]}$ and $V_2=\set{(u,2):u\in[m]}$ in such a way that  $u <u'\iff\sigma((u,1))<\sigma((u',1))$ and $w <w'\iff\sigma((w,2))<\sigma((w',2))$.
%
\danupon{We can simplify the proof by replace $V_1\cup V_2$ by $U\cup V$ and $(v, 1)$ by $u$ and $(v, 2)$ by $v$. Also use $n_1, n_2$ instead of $m, n$}
For the sake of presentation, we represent the set of vertices of $G$ as $V_1=\set{(u,1):u\in[m]}$ and $V_2=\set{(u,2):u\in[n]}$.

We define a linear system consisting of $n$ variables $\{x_{(w, 2)}\}_{w\in [n]}$ and the following $m$ constraints:
\begin{equation}\label{eq:constraints}
\forall u\in [m], \quad (\beta n)^{3u-1} \leq \sum_{w: (u, 1)(w,2) \in E(G)} x_w \leq (\beta n)^{3 u}
\end{equation}
Formally, for each $(u,w)\in[m]\times[n]$, define $a_{u,w}=1$ if $(u, 1)(w,2) \in E(G)$ and $a_{u,w}=0$ otherwise.
Then we create a constraint $u$ for each vertex $(u,1) \in V_1$ as $\ell_u \leq \sum_{(w,2) \in V_2} a_{u,w} x_w \leq \mu_u$ where $\ell_u=(\beta n)^{3u-1}$ and $\mu_u=(\beta n)^{3 u}$.

%$(\beta n)^{3u-1} \leq \sum_{(w,2) \in V_2} a_{u,w} x_w \leq (\beta n)^{3 u}$.

By the construction, the number of constraints is $m=|V_1|$, and the number of variables is $n= |V_2|$. Also, notice that $L = \max_{(u,1) \in V_1} \ell_u = (\beta n)^{O(m)}$. This proves the first property.

Let $\mset =\set{(u_i, 1)(w_i,2): i =1,\ldots,r }$ be an induced matching of size $r$ in $G$. We can define the following solution for linear system $\aset$: For each $i =1,\ldots, r$, we have $x_{w_i} = \ell_{u_i}$, and $x_{w'} = 0$ for all other $w'\in [n]$. It suffices to show that a constraint $u_i$ is satisfied for all $i \in [r]$. To see this, consider any constraint $u_i$, where $i\in[r]$. Only variables $x_{w_j}$ with $(u_i,1)(w_j,2) \in E(G)$ participate in this constraint, and the only variable with positive value is $x_{w_i}= \ell_{u_i}$; otherwise, it would contradict the fact that $\mset$ is an induced matching. This proves the second property.

Now, to prove the third property, assume that we have a solution $\x$ that $\beta$-satisfies $r$ constraints, for some $0<\beta\leq |V(G)|$; i.e, there exists a subset $V^*_1 \subseteq V_1$, denoted by $V_1^* = \set{(u_1, 1),\ldots, (u_r, 1)}$, such that
$$\forall (u_i, 1)\in V_1^* \quad (\beta n)^{3u_i-1} \leq \sum_{w: (u_i, 1)(w,2) \in E(G)} x_w \leq \beta(\beta n)^{3 u_i}\,.$$
%
%
%Let $V_1^* = \set{u_1,\ldots, u_r}$ be the index set of satisfied constraints.
We note the following claim.

\begin{claim}
For any $(u_i, 1)\in V_1^*$, there exists $w_i\in [n]$ such that $x_{w_i} \geq \ell_{u_i}/n$ and $(u_i,1)( w_i,2) \in E(G)$.
\end{claim}
\begin{proof}
Consider the constraint $u_i$, which involves variables $x_{w'}$ for all $w'\in [n]$ such that $(u_i,1)( w',2) \in E(G)$. Since there are at most $n$ such variables and $\sum_{w': (u_i, 1)(w',2) \in E(G)} x_{w'} \geq \ell_{u_i}$, one of the variables $x_{w'}$ must have a value of at least $\ell_{u_i}/n$. Thus, $w_i=w'$ is the desired index, proving the claim.
\qedhere
\end{proof}
Next, we define a set of matching $\mset$ by $\set{(u_{i},1)(w_i,2)}_{i=1}^r$ where $w_i$ is as in the above claim.
It is not difficult to check that $\mset$ is a matching: For any $u_i>u_j$, note that $x_{w_i}\geq \ell_{u_i}/n = (\beta n)^{3u_i-1}/n>\beta (\beta n)^{3 u_j}=\beta \mu_{u_j}\geq x_{w_j}$; thus, $w_i\neq w_j$.

We define a total order $\sigma$ on $V(G)$ as $\sigma(v) = v$ for all $(v, 1) \in V_1$ and $\sigma(w) = m + w$ for all $(w, 2) \in V_2$.
We claim that $\mset$ is a $\sigma$-semi-induced matching.
To see this, assume that it is not. Then, by the definition of $\sigma$-semi-induced matching and the fact that $G$ is bipartite, there must be some edge $(u_i, 1)(w_j,2) \in E(G)$ for some $i$ and $j$ such that $u_i < u_j$. Observe that
$$\sum_{w': (u_i,1)( w',2)\in E(G)} x_{w'} \geq x_{w_j} \geq \ell_{u_j}/n = (\beta n)^{3u_j-1}/n>\beta (\beta n)^{3 u_i}=\beta \mu_{u_i}.$$
This means that constraint $u_i$ is violated by more than a factor of $\beta$, contradicting the assumption that $\x$ $\beta$-satisfies constraints corresponding to vertices in $V_1^*$. This proves the third claim and completes the proof of Theorem~\ref{thm: elbassioni MRFS}.
\end{proof}

\subsection{Equivalence between semi-induced matching and maximum expanding sequence}\label{sec:equivalence semi induced and expanding seq}

\begin{comment}
Briest and Krysta~\cite{BriestK11} defined a new problem, that they called maximum expanding sequence,
and prove that a restricted version of \MES can be reduced to both \SMP and \UDP.
By giving a reduction from balanced bipartite independent set to
maximum expanding sequence, they obtained the hardness results for both
\UDP and \SMP.
\end{comment}

\danupon{To discuss: Should we change in the main paper that MES is equivalent to $\sigma$-semi-induced matching?}

In this section, we show that the maximum expanding sequence problem is in fact equivalent to the
semi-induced matching problem.

%We choose to work with the semi-induced matching problem, instead of the maximum expanding sequence problem, in order to be consistent with the rest of the paper. Finally, we show the hardness of a special case of maximum semi-induced matching, thus implying the hardness of both \UDP and \SMP.

\paragraph{Maximum Expanding Sequence} We are given an ordered collection of sets
$\sset = \set{S_1,\ldots, S_m}$ over the ground elements $\uset$.
An {\em expanding sequence} $\phi = (\phi(1),\ldots, \phi(\ell))$ of
length $\ell$ is a selection of sets $S_{\phi(1)}, \ldots,
S_{\phi(\ell)}$ such that, for all $j: 2 \leq j \leq \ell$, we have $\phi(j-1)<\phi(j)$\danupon{I added this} and
$S_{\phi(j)} \not\subseteq \bigcup_{j' < j} S_{\phi(j')}$.
Our objective is to compute an expanding sequence of maximum
length.

Given an instance $(\sset, \uset)$ of the maximum expanding sequence problem, we denote by $\opt(\sset, \uset)$ the optimal value of the instance. The following theorem shows that this problem is equivalent to the maximum semi-induced matching problem.

\begin{theorem}
Let $(\sset, \uset)$ be an instance of the maximum expanding sequence problem. Then there is a polynomial-time reduction that constructs an instance $(G, \sigma)$ of the $\sigma$-semi-induced matching problem such that $\sinducesigma{G}{\sigma} = \opt(\sset, \uset)$. Conversely, given an instance $(G, \sigma)$ of the semi-induced matching problem, we can construct $(\sset, \uset)$ such that $\opt(\sset, \uset) = \sinducesigma{G}{\sigma}$.
\end{theorem}

\begin{proof}
We only prove one direction of the reduction. It will be clear from the description that the converse also holds.
Given an instance $(\sset, \uset)$ of the expanding sequence problem, we construct the bipartite graph $G=(V_1 \cup V_2, E)$ where $V_1 = \set{(i,1): i \in [|\sset|]}$ and $V_2 = \set{(i,2): i \in [|\uset|]}$. Each $S_i \in \sset$ corresponds to $(i,1) \in V_1$ and each element $j \in \uset$ corresponds to vertex $(j,2)$ in $V_2$. The set $S_i$ contains $j$ if and only if $(i,1) (j,2) \in E$, and finally the total order $\sigma$ is defined as follows:
\begin{itemize}
\item $\sigma((i,1)) = i$ for all $0\leq i\leq |\sset|-1$, and
\item $\sigma((j,2)) = |\sset| +j$ for all $0\leq j\leq |\uset|-1$.
\end{itemize}
Notice that the total order $\sigma$ put the order of vertices in $V_1$ before those in $V_2$, and the ordering of vertices in $V_1$ are ordered according to their corresponding sets.

\parinya{I rephrased this paragraph.}We now claim that expanding sequences in $(\sset, \uset)$ are equivalent to
$\sigma$-semi-induced matchings in $G$.
For any expanding sequence,
$S_{\phi(1)},\ldots,S_{\phi(\ell)}$, we define the $\sigma$-semi-induced matching as follows: For each $j=1,\ldots, \ell$, we have an edge $(\phi(j),1)(\psi(j), 2)$ where $\psi(j)$ is defined as an arbitrary element in $S_{\phi(j)} \setminus \paren{\bigcup_{j' < j} S_{\phi(j')}}$ (we know such element exists due to the property of expanding sequences). It is now easy to check that the set $\mset= \set{(\phi(j), 1)(\psi(j), 2)}$ is indeed a $\sigma$-semi-induced matching.
\qedhere
\end{proof}

\subsection{Proof of Theorem~\ref{thm: semi-induced matching to UDP}}
%%%%%%%%%%%%%%%%%%%%%
%%%% Added by BUN %%%
%%%%%%%%%%%%%%%%%%%%%
\newcommand{\OptUDP}{\opt_{\mathsf{UDP}}}
\newcommand{\OptSMP}{\opt_{\mathsf{SMP}}}

First, we describe our construction. We view the graph $G' = \bipp[G
\vee H]$ as a bipartite graph $(V_1 \cup V_2, E)$ where $V_1= \set{(u,a,1): u \in V(G), a \in
  V(H)}$ and $V_2 = \set{(u,a,2): u \in V(G), a \in V(H)}$.
For convenience, we may think of vertices in $V(G)$ and $V(H)$ as integers (so that their ordering and arithmetic can be naturally done).
For each vertex $(u,a,2) \in V_2$, we have an item $I(u,a)$.
So, $\iset=\set{I(u,a): u \in V(G), a \in V(H)}$, and hence $|\iset| = |V(G)||V(H)|$.
For each vertex $(u,a,1) \in V_1$, we have $n^{3a}$ consumers $\cset(u,a) = \set{c(u,a,r)}_{r=1}^{n^{3a}}$, and each such consumer in this set has budget $B_{c(u,a,r)} = 1/n^{3a}$ and an associated set $S_{c(u,a,r)} =
\set{I(v,b): (u,a,1)(v,b,2) \in E}$.
The final set of consumers is $\cset = \bigcup_{u,a} \cset(u,a)$.
This completes the construction. The instance $(\cset, \iset)$ here will be used as both \UDP and \SMP instances. 

We first show that the optimal revenue we receive from the above instance is at least the size of the maximum induced matching, for both \UDP and \SMP. Let $\OptUDP(\cset, \iset)$ and $\OptSMP(\cset, \iset)$ denote the optimal value on instance $(\cset, \iset)$ of \UDP and \SMP, respectively.

\begin{lemma} The followings hold for \UDP and \SMP:
\begin{itemize}
\item $\OptUDP(\cset, \iset) \geq \induce{G'}$
\item $\OptSMP(\cset, \iset) \geq \induce{G'}$
\end{itemize}
\end{lemma}
\begin{proof}
Let $\mset$ be an induced matching of cardinality $K$.
First, for each $(u,a,1)(v,b,2)\in \mset$, we set the price of $I(v,b)$ to $p(I(v,b)) = 1/n^{3a}$.
% The rest of the items have their prices set to $\infty$ in the case of \UDP and to $0$ in the case of \SMP.
Next, we set prices of the other items.
For \UDP, we set their prices to $\infty$.
For \SMP, we set their prices to $0$.
It is clear that this price function is well-defined because $\mset$ is a matching.

Now, we argue that the revenue that can be made from the price function $p$ is $K$ for both \UDP and \SMP.
It suffices to show that, for each pair $(u,a)$ such that $(u,a,1)$ is matched in $\mset$, each consumer in the set $\cset(u,a)$ pays the price of $1/n^{3a}$.
Consider any edge $(u,a,1)(v,b,2) \in \mset$.
% Notice that, in the case of \UDP, for any consumer $c \in \cset(u,a)$, we have exactly one item in $S_c$
% with finite price of $1/n^{3a}$, i.e., item $I(v,b)$; otherwise, the matching $\mset$ would not have been an induced matching. Similarly, in the case of \SMP, for any consumer $c \in \cset(u,a)$, we have exactly one item in $S_c$
% with non-zero price of $1/n^{3a}$, i.e., item $I(v,b)$. Therefore, the total profit made from consumers in $\cset(u,a)$ is exactly one for both \UDP and \SMP. The lemma immediately follows.
For \UDP, any consumer $c \in \cset(u,a)$ has exactly one item $I(v,b)\in S_c$ with finite price of $1/n^{3a}$; otherwise, $\mset$ would not be an induced matching.
Similarly, for \SMP, any consumer $c \in \cset(u,a)$ has exactly one item in $I(v,b)\in S_c$ with non-zero price of $1/n^{3a}$.
Therefore, the total profit made from consumers in $\cset(u,a)$ is exactly one for both \UDP and \SMP, and the
lemma follows.
\end{proof}

The next lemma proves the upper bound of the revenue.

\begin{lemma}
\label{lemma: upper bound of OPT}
The followings hold for \UDP and \SMP:
\begin{itemize}
\item $\OptUDP(\cset, \iset) \leq 2 \sinduce{G'} + |V(G)|(|E(H)|+1)$
\item $\OptSMP(\cset, \iset) \leq 2 \sinduce{G'} + O(|V(G)||E(H)|)$
\end{itemize}
\end{lemma}

\begin{proof}
Let $p^*$ be any {\em optimal} price function for either \UDP or \SMP, and let $K$ be the revenue that we obtain from $p^*$. Our goal is to show that we can find a semi-induced matching of cardinality $K/2 - |V(G)|(|E(H)|+1)$ in $G'$.

First, observe that, for any $u \in V(G), a \in V(H)$, all consumers in
$\cset(u,a)$ pay exactly the same price since they desire the same set of items and have the same budget. So, we will treat these consumers as a bundle and refer to index $(u,a)$ as a representative of all consumers in $\cset(u,a)$. 

We need the following notion of tight index, which intuitively captures the consumers who spend a sufficiently large fraction of their budgets in the solution $p^*$. 

\begin{definition}[Tight Index]
We say that an index $(u,a)$ is {\em tight} if consumers in $\cset(u,a)$ pay between $1/n^{3a+1}$ and $1/n^{3a}$ (i.e., between $1/n$ and $1$ proportion of their budgets).
%For the case of \UDP, we say that an index $(u,a)$ is {\em tight} if consumers in $\cset(u,a)$ pay exactly the amount of $1/n^{3a}$ (i.e., their budget).
\end{definition}

\begin{claim}
The number of tight indices is at least $K/2$.
\end{claim}
\begin{proof}
Assume for a contradiction that the number of tight indices is less than $K/2$.
Consider a non-tight index $(u,a)$ such that consumers in $\cset(u, a)$ {\em can afford} to buy items:
\begin{itemize}
\item For \UDP, $\displaystyle\min_{(v, b): (u,a,1)(v,b,2)\in E(G')}  p^*(I(v, b))\in [0, 1/n^{3a+1})$.
\item For \SMP, $\displaystyle\sum_{(v, b): (u,a,1)(v,b,2)\in E(G')}  p^*(I(v, b))\in [0, 1/n^{3a+1})$.
\end{itemize}
We call these indices {\em feasible non-tight} indices.

Observe that we earn a profit of strictly larger than $K/2$ from consumers corresponding to these feasible non-tight indices in both \UDP and \SMP, and this is because the price function $p^*$ yields a revenue of $K$ which only comes from either consumers with tight or feasible non-tight indices, and we get strictly less than $K/2$ for tight indices. 

Now, define a new price function $p' = 2p^*$.
By using the price function $p'$, the revenue that we earn from the feasible non-tight indices will be twice.
So, we would have revenue strictly more than $K$ for both \UDP and \SMP. This contradicts the optimality of $p^*$.
\end{proof}

%Now, let $\Lambda$ be the set of all tight indices; so $|\Lambda| \geq K/2$ by the above claim. We would like to ``recover'' a large semi-induced matching of $G'$ from the set $\Lambda$.

Our goal is to ``recover'' a large semi-induced matching of $G'$ from the tight indices. We will show that the set of ``recoverable indices'', which is large, is exactly the following set of {\em canonical tight indices}.  Note that, in both \UDP and \SMP, for any tight index $(u, a)$, there must be an item $I(v, b)$ such that $(u,a,1)(v,b,2)\in E(G')$ (i.e., the item that consumers in $\cset(u,a)$ want to buy) and $1/n^{3a+2}\leq p^*(I(v, b))\leq 1/n^{3a}$ (i.e., $I(v, b)$ is expensive but affordable by consumers in $\cset(u, a)$). We say that $(u, a)$ is {\em canonical} if and only if $I(u, a)$ is the only such item.
To be precise, the canonical tight index is defined as below.
\begin{definition}[Canonical Tight Index]
We say that a tight index $(u, a)$ is {\em canonical} if for any $(v, b)$ such that $(u,a,1)(v,b,2)\in E(G')$, we have that $1/n^{3a+2}\leq p^*(I(v, b))\leq 1/n^{3a}$ if and only if $u=v$ and $a=b$.
\end{definition}
First, we show that the number of canonical tight indices is large.
\begin{claim}
There are at least $K/2-|V(G)|(|E(H)|+1)$ canonical tight indices.
\end{claim}
\begin{proof}
Let $(u, a)$ be any non-canonical tight index. Recall that since $(u, a)$ is tight, there must be an item $I(v, b)$ such that $(u,a,1)(v,b,2)\in E(G')$ and $1/n^{3a+2}\leq p^*(I(v, b))\leq 1/n^{3a}$. Since $(u, a)$ is non-canonical, it is {\em not} the case that both $u=v$ and $a=b$. Consequently, since $(u,a,1)(v,b,2)\in E(G')$, either $uv\in E(G)$ or $ab\in E(H)$. We say that $(u, a)$ is {\em $G$-non-canonical} if $uv\in E(G)$ and {\em $H$-non-canonical} if $ab\in E(H)$. Note that every non-canonical tight index must be either $G$-non-canonical or $H$-non-canonical (or both).

First, we claim that the number of $G$-non-canonical tight indices is at most $|V(G)|$. In particular, we claim that for any $u\in V(G)$, there is {\em at most} one $G$-non-canonical tight index of the form $(u, a, 1)$. To see this, let us assume for a contradiction that there are two $G$-non-canonical tight indices $(u, a)$ and $(u, a')$ for some $u\in V(G)$ and $a, a'\in V(H)$ such that $a<a'$.

For the case of \UDP, observe that since $(u, a')$ is $G$-non-canonical and tight, there exists an index $(v, b')$ such that (1) $uv\in E(G)$ and (2) $1/n^{3a'+2}\leq p^*(I(v, b'))\leq 1/n^{3a'}$. The first property implies that $(u, a, 1)(v, b', 2)\in E(G')$. Consequently, by the second property, consumers in $\cset(u, a)$ pays {\em at most} $1/n^{3a'}<1/n^{3a+1}$, contradicting the assumption that index $(u, a)$ is tight.

For the case of \SMP, observe that since $(u, a)$ is $G$-non-canonical and tight, there exists an index $(v, b)$ such that (1) $uv\in E(G)$ and (2) $1/n^{3a+2}\leq p^*(I(v, b))\leq 1/n^{3a}$. The first property implies that $(u, a', 1)(v, b, 2)\in E(G')$. Consequently, by the second property, consumers in $\cset(u, a)$ pays {\em at least} $1/n^{3a+2}>1/n^{3a'}$, contradicting the assumption that index $(u, a')$ is tight.

Secondly, we claim that the number of $H$-non-canonical tight indices is at most $|V(G)||E(H)|$. To see this, recall that for any $(u, a)$ that is $H$-non-canonical and tight, there exists an index $(v, b)$ such that (1) $ab\in E(H)$ and (2) $1/n^{3a+2}\leq p^*(I(v, b))\leq 1/n^{3a}$. Observe that, by the first condition, any $H$-non-canonical tight index $(u, a)$ must be in the following set $\Phi(u)=\bigcup_{a, b:\ ab\in E(H)} \{(u, a)\}.$
%
%Consider any $v\in V(G)$. Let $\Phi(v)$ be the set of tight indices $(u, a)$ that is $H$-non-canonical because of $v$; formally,
%\[\Phi(v)=\{(u, a)\mid \mbox{$(u, a)$ is tight, $\exists b$ s.t. $ab\in E(H)$, and $1/n^{3a+2}\leq p^*(I(v, b))\leq 1/n^{3a}$}\}.\]
%Observe that the number of $H$-non-canonical tight indices is at most $\sum_{v\in V(G)} |\Phi(v)|$ since each such index must be $H$-non-canonical because of some vertex $v\in V(G)$. Observe further that $|\Phi(v)|\leq |E(H)|$
%
Obviously, $|\Phi(u)|\leq |E(H)|$. It follows that the number of $H$-non-canonical tight indices is at most $\sum_{u\in V(G)} |\Phi(u)|\leq |V(G)||E(H)|$ as claimed.

Now, we have already shown that there are at most $|V(G)|(|E(H)|+1)$ non-canonical tight indices; since the number of tight indices is at least $K/2$ by the previous claim, we have that the number of canonical tight indices is at least $K/2-|V(G)|(|E(H)|+1)$ as desired.
\end{proof}

We finish the proof by showing that we can recover a large semi-induced matching from canonical tight indices.

\begin{claim}
$\sinduce{G'}$ is at least the number of canonical tight indices. In other words, $\sinduce{G'}\geq K/2-|V(G)|(|E(H)|+1)$.
\end{claim}
\begin{proof} Let $(u_1, a_1), (u_2, a_2) \ldots (u_t, a_t)$, for some $t$, be the canonical tight indices. Order them in such a way that $a_1\leq a_2\leq \ldots \leq a_t$. Let $\mset=\{(u_i, a_i, 1)(u_i, a_i, 2)\}_{i=1\ldots t}$.

For the case of \UDP, let $\sigma$ be any total ordering of nodes in $G'$ such that $\sigma(u_1, a_1, 1)<\sigma(u_2, a_2, 1) <\ldots < \sigma(u_t, a_t, 1)$ and $\sigma(u_i, a_i, 2) > \sigma(u_t, a_t, 1)$ for all $i$. We claim that $\mset$ is a $\sigma$-semi-induced matching in $G'$. In particular, we show that for any $i<j$, $(u_i, a_i, 1)(u_j, a_j, 2)\notin E(G')$. To see this, note that since $(u_j, a_j)$ is canonical and tight, $1/n^{3a_j+2}\leq p^*(I(u_j, a_j))\leq 1/n^{3a_j}$.

Consider two possible cases: either (1) $a_i<a_j$ or (2) $a_i=a_j$.
In the first case, we have that $p^*(I(u_j, a_j))\leq 1/n^{3a_j}<1/n^{3a_i+1}$. Thus, if  $(u_i, a_i, 1)(u_j, a_j, 2)\in E(G')$, then consumers in $\cset(u_i, a_i)$ will pay strictly less than $1/n^{3a_i+1}$, contradicting the fact that $(u_i, a_i)$ is tight.
In the second case, we have that $1/n^{3a_i+2}\leq p^*(I(u_j, a_j))\leq 1/n^{3a_i}$. Thus, if  $(u_i, a_i, 1)(u_j, a_j, 2)\in E(G')$, then $(u_i, a_i)$ is not a canonical index, a contradiction.
Since both cases lead to a contradiction, we have that $(u_i, a_i, 1)(u_j, a_j, 2)\notin E(G')$, and thus $\mset$ is a $\sigma$-semi-induced matching as claimed.

\medskip For the case of \SMP, let $\sigma'$ be any total ordering such that $\sigma'(u_1, a_1, 1)>\sigma'(u_2, a_2, 1) >\ldots > \sigma'(u_t, a_t, 1)$ and $\sigma'(u_i, a_i, 2) > \sigma'(u_1, a_1, 1)$ for all $i$. We claim that $\mset$ is a $\sigma'$-semi-induced matching in $G'$. In particular, we show that for any $i>j$, $(u_i, a_i, 1)(u_j, a_j, 2)\notin E(G')$. To see this, note again that since $(u_j, a_j)$ is canonical and tight, $1/n^{3a_j+2}\leq p^*(I(u_j, a_j))\leq 1/n^{3a_j}$.

Consider two possible cases: either (1) $a_i>a_j$ or (2) $a_i=a_j$.
In the first case, we have that $p^*(I(u_j, a_j))\geq 1/n^{3a_j+2}>1/n^{3a_i}$. Thus, if  $(u_i, a_i, 1)(u_j, a_j, 2)\in E(G')$, then consumers in $\cset(u_i, a_i)$ pay strictly more than $1/n^{3a_i}$, contradicting the fact that $(u_i, a_i)$ is tight.
In the second case, we have that $1/n^{3a_i+2}\leq p^*(I(u_j, a_j))\leq 1/n^{3a_i}$. Thus, if  $(u_i, a_i, 1)(u_j, a_j, 2)\in E(G')$, then $(u_i, a_i)$ is not a canonical index, a contradiction.
Since both cases lead us to a contradiction, we have that $(u_i, a_i, 1)(u_j, a_j, 2)\notin E(G')$, and thus $\mset$ is a $\sigma'$-semi-induced matching as claimed.
\end{proof}

The above claim completes the proof of the lemma.
\end{proof}

\end{document}